\documentclass[11pt]{article}
\usepackage{fullpage}
\usepackage{amsmath,amsthm,amsfonts,amssymb,dsfont,bm}
\usepackage{enumerate,color,xcolor}
\usepackage{graphicx}
\usepackage[colorlinks,linkcolor=cyan,citecolor=blue]{hyperref}
\usepackage{url}

\numberwithin{equation}{section}
\theoremstyle{plain}

\newtheorem{theorem}{Theorem}
\newtheorem{corollary}[theorem]{Corollary}
\newtheorem{definition}[theorem]{Definition}
\newtheorem{lemma}[theorem]{Lemma}
\newtheorem{proposition}[theorem]{Proposition}

\newtheorem{remark}[theorem]{Remark}

\begin{document}
\begin{center}
{\bf\Large VIX options in Bergomi models}
\end{center}

\author{}
\begin{center}
  Desen Guo\,\footnote{Department of Mathematics, Florida State University, United States of America; gd19k@fsu.edu},
{Dan Pirjol}\,\footnote{School of Business, Stevens Institute of Technology, United States of America;
  dpirjol@gmail.com},
  Lingjiong Zhu\,\footnote{Department of Mathematics, Florida State University, United States of America; zhu@math.fsu.edu
 }
\end{center}

\begin{center}
 \today
\end{center}

\begin{abstract}
We present a study of the leading-order asymptotics for VIX option prices in Bergomi models in the short-maturity
and small volatility-of-volatility regimes. 
Both out-of-the-money (OTM) and at-the-money (ATM) asymptotics are considered for one-factor, two-factor Bergomi and $N$-factor models. 
The leading-order asymptotics are obtained in closed-form, which are translated into predictions for the small-maturity asymptotics of the VIX implied volatility.
Numerical illustrations are provided to illustrate the efficiency of the closed-form asymptotic formulas.
\end{abstract}

\section{Introduction}

The CBOE Volatility Index (VIX) is the main volatility benchmark of the U.S. stock market, which provides a measure of the implied volatility of options with a maturity of 30 days on the S\&P 500 index. The VIX is defined in terms of an expectation in the risk-neutral measure 
\begin{equation}\label{VIX:formal:defn}
\mathrm{VIX}_t^2 = - \frac{2}{\tau} \mathbb{E}\left[\log\left(\frac{S_{t+\tau}}{S_t}\right)\Big|\mathcal{F}_t\right]+\frac{2}{\tau}\mathbb{E}\left[\int_{t}^{t+\tau}\frac{dS_{s}}{S_{s-}}\Big|\mathcal{F}_{t}\right],    
\end{equation} 
where $S_{t}$ is the S\&P 500 equity index  at time $t$ and $\tau = 30$ days. The expectation is computed by replication in terms of market-observed SPX option prices; see e.g. the VIX White Paper \cite{VIXwp} for the details of the methodology.
Since 2022, CBOE has also started reporting 
the CBOE 1-day Volatility Index \texttt{(VIX1D)} \cite{VIX1Dwp}, which is an analog of the VIX index computed using the PM-settled weekly SPX options, which mature on the same day and the next day $(\tau=1$ day) as the index date.

The volatility index VIX is used by market participants to speculate on and hedge volatility risk. Several volatility derivatives that can be used for this purpose are traded on the CBOE Exchange: futures contracts on VIX have been traded since 2004, and 
VIX options have been traded since 2006. 

In view of the popularity of these contracts, a great deal of work has been devoted in the literature to the valuation of volatility derivatives. 
By modeling the instantaneous variance $V_t$ as a stochastic process, Detemple and Osakwe (2000) \cite{Detemple2000} studied both European and American volatility options pricing under several popular diffusion models for $V_t$. 
Carr et al. (2005) \cite{Carr2005} presented results for volatility options under pure jump models with independent increments.
Sepp (2008) \cite{Sepp2008a,Sepp2008} priced volatility derivatives under a square root volatility model with jumps. 
Goard and Mazur (2013) \cite{Goard2013} derived
analytical results for VIX options in the 3/2 stochastic volatility model, and 
Baldeaux and Badran (2014) \cite{Baldeaux2014} extended this model for VIX option pricing by adding jumps. 
Kokholm and Stisen (2015) \cite{Kokholm2015} studied
the joint pricing of VIX and SPX options with stochastic 
volatility and jump models.

The pricing of volatility derivatives has also been extended to rough stochastic volatility models, where the volatility is driven by a fractional Brownian motion. For example, Horvath et al. \cite{Horvath2020} introduced the class of modulated Volterra processes that can accommodate observed VIX smiles. Jacquier et al. (2021) \cite{Jacquier2021}
derive short-maturity SPX and VIX option prices for a wide class of multi-factor models of this type. 
An empirical analysis of the SPX and VIX option markets under rough and stochastic volatility models was given 
by R\o{}mer (2022) \cite{Romer2022}.

We also mention the martingale optimal transport approach which was applied to the problem of simultaneous calibration to the SPX and VIX implied volatility smiles in Guyon (2020) \cite{Guyon2020}. This approach is model-independent and aims to calibrate the joint distribution of the underlying (SPX) and of the VIX at several maturities of interest, under appropriate martingale constraints. 
Guyon (2024) \cite{Guyon2024} obtained
the exact joint calibration of SPX options, VIX futures
and VIX options by solving a dispersion-constrained martingale
Schr\"{o}dinger problem.

Finally, we mention a few other recent works on VIX option pricing. 
Cao et al. (2020) \cite{Cao2020} studied the pricing
of VIX and target volatility options under affine GARCH models.
Tong and Huang (2021) \cite{Tong2021}
studied the pricing of VIX option pricing
with realized volatility using GARV and GARCH models.
Al\`{o}s et al. (2022) \cite{Alos2022VIX} studied the implied volatility of 
volatility derivatives such as VIX options using methodologies built upon Malliavin calculus techniques, and obtained short-maturity at-the-money implied volatility level and skew. 
Abi Jaber et al. (2022) \cite{AbiJaber2022} 
considered joint SPX-VIX calibration by modeling the volatility as a polynomial of order five in an Ornstein-Uhlenbeck model.
Yuan (2022) \cite{Yuan2022} studied a new reduced-form model for the pricing of VIX derivatives that includes an independent stochastic jump intensity factor and cojumps in the level and variance of VIX, while allowing the mean of VIX variance to be time varying.
Cuchiero et al. (2024) \cite{Cuchiero2023}
studied SPX and VIX options with signature-based models.
Pirjol et al. (2024) \cite{VIXpaper} studied the VIX and European options pricing 
for local-stochastic volatility models in the short-maturity regime, 
and Pirjol and Zhu (2025) \cite{PZ2025} focused on the VIX options in the SABR model.
Finally, Guo et al. (2026) \cite{VIX-jumps-paper} studied short-maturity VIX and European options pricing 
for local-stochastic volatility models in the presence of jumps.

An alternative approach to modeling stochastic volatility are the 
\textit{variance curve} models. They were proposed by Bergomi in a series of papers \cite{Bergomi,Bergomi-II,Bergomi-III}. 
See \cite{BergomiBook} for a textbook treatment. 
In this approach, one models the dynamics of the entire curve of the forward variance $\xi_t^T := \mathbb{E}[V_T|\mathcal{F}_t]$. This approach is sometimes called second generation stochastic volatility model, in contrast to the approach based on modeling the joint dynamics of the asset price and instantaneous variance $V_t=\xi_t^t$, which is called the first generation type of models. 

The original formulation of the Bergomi models assumed a Markovian dynamics. Recently, they have been extended by allow for rough volatility, which are not Markovian. Although their implementation and simulation is more challenging, 
efficient numerical methods for implementing these models have been proposed, such as the quadrature approach of Bergomi \cite{Bergomi-II}, the semi-analytic approach of Ould Aly \cite{OuldAly} and the quantization approach proposed in \cite{Kyakutwika}. For pricing and calibration, it is useful to also have analytical approximations for pricing variance options, VIX futures and options. 

Finally, we give a brief survey of the literature on analytical expansions for derivatives pricing in Bergomi models.
The implied volatility of European options in these models was obtained by an expansion in volatility-of-volatility by Bergomi and Guyon (2012) \cite{Bergomi-Guyon} and by Guyon (2021) \cite{Guyon2021}. The implied volatility of VIX options in these models was also studied.
Guyon (2022) \cite{Guyon} used a small volatility of volatility expansion to obtain analytical formulas for computing VIX
futures and options in multi-factor Bergomi models. 
Bourgey, De Marco and Gobet \cite{Bourgey} derived weak-approximation techniques  for VIX option pricing in the mixed Bergomi models 
and applied them recently in Liao, Agarwal and Bourgey (2026) \cite{Liao} to obtain expansions for the VIX implied volatility for both standard and rough Bergomi-type models. 
Lacombe, Muguruza and Stone \cite{Lacombe} use large deviations theory method to study the short-maturity asymptotics in rough Bergomi models. A short-maturity expansion for the implied volatility of VIX options around the ATM region was given using Malliavin calculus techniques, by Alos et al \cite{Alos2022VIX}. 

In this paper, we are interested in studying the VIX option pricing
in Bergomi models. 
We will study VIX option pricing for one-factor, two-factor and $N$-factor Bergomi models 
in both the short-maturity regime and the small volatility-of-volatility (vol-of-vol) regime.

The rest of the paper is organized as follows. 
In Section~\ref{sec:Bergomi}, we review the one-factor and $N$-factor Bergomi models.
We present our main results in Section~\ref{sec:main}.
In particular, we study the pricing of VIX options for the short-maturity regime
in Section~\ref{sec:short} and for the small vol-of-vol regime in Section~\ref{sec:small}.
The asymptotic results are translated into predictions for the small-maturity asymptotics of the VIX implied volatility. Explicit results for the expansion of the VIX implied volatility in log-moneyness are obtained in Section~\ref{sec:VIXsmile}.
We provide numerical experiments in Section~\ref{sec:numerical} to demonstrate the efficiency
of the theoretical results derived in Section~\ref{sec:main} and Section~\ref{sec:VIXsmile}.
Finally, a brief background on large deviations theory is provided in Appendix~\ref{sec:LDP}
and all the technical proofs are presented in Appendix~\ref{sec:proofs}.

\subsection{Bergomi model}\label{sec:Bergomi}

In this section, we review the Bergomi models from the literature that will be used as the underlying of our studies on VIX options. 

\subsubsection{One-factor Bergomi model}

We start by considering the one-factor 
Bergomi model \cite{Bergomi-II}. Under the risk-neutral probability measure $\mathbb{Q}$:
\begin{align}
&\frac{d\xi_{t}^{u}}{\xi_{t}^{u}}=\omega e^{-k(u-t)}dZ_{t},
\\
&\frac{dS_{t}}{S_{t}}=(r-q)dt+\sqrt{\xi_{t}^{t}}dW_{t},
\end{align}
where $Z_{t},W_{t}$ are standard Brownian motions that might be correlated, 
$\omega,k>0$ and $r,q$ are the instantaneous interest rate
and dividend yield. 
We assume that $\xi_{0}^{u}>0$ and continuous in $u$ for every $u\geq 0$.
It is known that $\xi_{t}^{u}$ admits a one-dimensional Markov representation:
\begin{equation}\label{xi:one}
\xi_{t}^{u}=\xi_{0}^{u}f^{u}(t,X_{t}),
\end{equation}
where
\begin{equation}
f^{u}(t,x):=\exp\left(\omega e^{-k(u-t)}x-\frac{\omega^{2}}{2}e^{-2k(u-t)}v_{t}\right),
\end{equation}
where $X_{t}:=\int_{0}^{t}e^{-k(t-s)}dZ_{s}$ satisfies the Ornstein-Uhlenbeck process
\begin{equation}\label{X:one}
dX_{t}=-kX_{t}dt+dZ_{t},\qquad X_{0}=0,
\end{equation}
with 
\begin{equation}\label{v:one}
v_{t}:=\mathrm{Var}(X_{t})=\frac{1-e^{-2kt}}{2k}.
\end{equation}

\subsubsection{Two-factor Bergomi model}
\label{sec:2FBergomi}

The two-factor Bergomi model has the following dynamics  under the risk-neutral probability measure $\mathbb{Q}$ \cite{Bergomi-III}:
\begin{align}
&\frac{d\xi_{t}^{u}}{\xi_{t}^{u}}=\omega\alpha_{\theta}\left(\theta_{1}e^{-k_{1}(u-t)}dZ_{t}^{1}+\theta_{2}e^{-k_{2}(u-t)}dZ_{t}^{2}\right),
\\
&\frac{dS_{t}}{S_{t}}=(r-q)dt+\sqrt{\xi_{t}^{t}}dW_{t},
\end{align}
where $Z_{t}^{1},Z_{t}^{2},W_{t}$ are standard Brownian motions that might be correlated, 
and the correlation between $Z_{t}^{1}$ and $Z_{t}^{2}$ is $\rho$ and 
$k_{1},k_{2}>0$, $\theta_{1},\theta_{2}\in[0,1]$ with $\theta_{1}+\theta_{2}=1$
and $\alpha_{\theta}$ is the normalizing factor given by
\begin{equation}
\alpha_{\theta}=\left(\theta_{1}^{2}+2\rho\theta_{1}\theta_{2}+\theta_{2}^{2}\right)^{-\frac{1}{2}}.
\end{equation}
We assume that $\xi_{0}^{u}>0$ and continuous in $u$ for every $u\geq 0$.
It is known that
\begin{align}\label{xi:two}
\xi_{t}^{u}=\xi_{0}^{u}f^{u}(t,x_{t}^{u}),
\end{align}
with
\begin{equation}\label{f2Fdef}
f^{u}(t,x):=\exp\left(\omega x-\frac{\omega^{2}}{2}v_{t}(u)\right),
\end{equation}
and
\begin{equation}
x_{t}^{u}:=\alpha_{\theta}\left(\theta_{1}e^{-k_{1}(u-t)}X_{t}^{1}+\theta_{2}e^{-k_{2}(u-t)}X_{t}^{2}\right),
\end{equation}
where $X_{t}^{i}$, $i=1,2$ are Ornstein-Uhlenbeck processes:
\begin{equation}\label{X:two}
dX_{t}^{i}=-k_{i}X_{t}^{i}dt+dZ_{t}^{i},\qquad X_{0}^{i}=0,\qquad i=1,2,
\end{equation}
such that
\begin{align}
\left(X_{t}^{1},X_{t}^{2}\right)\sim\mathcal{N}\left(0,
\left(\begin{array}{cc}
v_{t}^{1} & v_{t}^{1,2}
\\
v_{t}^{1,2} & v_{t}^{2}
\end{array}\right)\right),
\end{align}
where
\begin{align}\label{v:i:two}
v_{t}^{i}:=\frac{1-e^{-2k_{i}t}}{2k_{i}},\quad i=1,2,
\qquad
v_{t}^{1,2}:=\rho\frac{1-e^{-(k_{1}+k_{2})t}}{k_{1}+k_{2}},
\end{align}
and
\begin{align}\label{v:two}
v_{t}(u):=\mathrm{Var}(x_{t}^{u})=\alpha_{\theta}^{2}\left(\theta_{1}^{2}e^{-2k_{1}(u-t)}v_{t}^{1}
+\theta_{2}^{2}e^{-2k_{2}(u-t)}v_{t}^{2}+2\theta_{1}\theta_{2}e^{-(k_{1}+k_{2})(u-t)}v_{t}^{1,2}\right).
\end{align}

\subsubsection{$N$-factor Bergomi model}

The $N$-factor Bergomi model \cite{Bergomi-III} is constructed by allowing $N$ stochastic drivers. This is given by the dynamics under the risk-neutral probability measure $\mathbb{Q}$:
\begin{align}
&\frac{d\xi_{t}^{u}}{\xi_{t}^{u}}=\omega\alpha_{\theta}\sum_{i=1}^{N}\theta_{i}e^{-k_{i}(u-t)}dZ_{t}^{i},
\\
&\frac{dS_{t}}{S_{t}}=(r-q)dt+\sqrt{\xi_{t}^{t}}dW_{t},
\end{align}
where $Z_{t}^{1},Z_{t}^{2},\ldots,Z_{t}^{N},W_{t}$ are standard Brownian motions that might be correlated, 
and the correlation between $Z_{t}^{i}$ and $Z_{t}^{i}$ is $\rho_{ij}$ for $i\neq j$ and 
$k_{i}>0$, $\theta_{i}\in[0,1]$ with $\sum_{i=1}^{N}\theta_{i}=1$
and $\alpha_{\theta}$ is the normalizing factor given by
\begin{equation}
\alpha_{\theta}=\left(\sum_{i=1}^{N}\theta_{i}^{2}+2\rho\sum_{1\leq i<j\leq N}\theta_{i}\theta_{j}\right)^{-\frac{1}{2}}.
\end{equation}
We assume that $\xi_{0}^{u}>0$ and continuous in $u$ for every $u\geq 0$.
It is known that
\begin{align}\label{xi:N}
\xi_{t}^{u}=\xi_{0}^{u}f^{u}(t,x_{t}^{u}),
\end{align}
with
\begin{equation}
f^{u}(t,x):=\exp\left(\omega x-\frac{\omega^{2}}{2}v_{t}(u)\right),
\end{equation}
and
\begin{equation}
x_{t}^{u}:=\alpha_{\theta}\sum_{i=1}^{N}\theta_{i}e^{-k_{i}(u-t)}X_{t}^{i},
\end{equation}
where $X_{t}^{i}$, $i=1,2,\ldots,N$ are Ornstein-Uhlenbeck processes:
\begin{equation}\label{X:N}
dX_{t}^{i}=-k_{i}X_{t}^{i}dt+dZ_{t}^{i},\qquad X_{0}^{i}=0,\qquad i=1,2,\ldots,N,
\end{equation}
such that $(X_{t}^{1},\ldots,X_{t}^{N})$ is an $N$-dimensional Gaussian random vector
with mean $0$ and covariance matrix $\Sigma_{t}:=(v_{t}^{ij})_{1\leq i,j\leq N}$, 
with
\begin{align}\label{v:ij:N}
v_{t}^{ii}:=\frac{1-e^{-2k_{i}t}}{2k_{i}},\quad i=1,2,\ldots,N,
\qquad
v_{t}^{ij}:=\rho_{ij}\frac{1-e^{-(k_{i}+k_{j})t}}{k_{i}+k_{j}},\qquad i\neq j,
\end{align}
and
\begin{align}
v_{t}(u):=\mathrm{Var}(x_{t}^{u})=\alpha_{\theta}^{2}\left(\sum_{i=1}^{N}\theta_{i}^{2}e^{-2k_{i}(u-t)}v_{t}^{ii}
+2\sum_{1\leq i<j\leq N}\theta_{i}\theta_{j}e^{-(k_{i}+k_{j})(u-t)}v_{t}^{ij}\right).
\end{align}


\section{Main Results}\label{sec:main}

In the Bergomi models, the volatility VIX$_T$ index at time $T$ is expressed as an average over the variance curve 
\begin{equation}
\mathrm{VIX}_{T}^{2}=\frac{1}{\tau}\int_{T}^{T+\tau}\xi_{T}^{u}du\,,
\end{equation}
where $\tau = 30/360$.
The VIX futures price with maturity $T$ is given by an expectation in the risk-neutral measure
\begin{equation}
F_V(T) = \mathbb{E}[\mathrm{VIX}_T]\,.
\end{equation}
We are interested in pricing the VIX call and put options:
\begin{align}
&C(T,\omega)=e^{-rT}\mathbb{E}[(\mathrm{VIX}_{T}-K)^{+}],
\\
&P(T,\omega)=e^{-rT}\mathbb{E}[(K-\mathrm{VIX}_{T})^{+}],
\end{align}
where $K>0$ is the strike price and the notations $C(T,\omega)$ and $P(T,\omega)$ emphasize
the dependence on maturity $T$ and vol-of-vol $\omega$. We will study the pricing of the VIX options in the short-maturity regime
and the small vol-of-vol regime respectively. 

A VIX call option is OTM if $K> F_V(T)$ and ITM if $K<F_V(T)$.
In the same way, a VIX put option is OTM if $K<F_V(T)$ and ITM if $K>F_V(T)$.
Denote 
\begin{equation}\label{F0def}
F^2_0(T) = \frac{1}{\tau}\int_T^{T+\tau} \xi_0^u du\,.
\end{equation}
$F_0(T)$ has a financial interpretation as the fair strike of a forward starting variance swap spanning the time interval $[T,T+\tau]$. 
In the Bergomi models the VIX futures price is expanded in volatility of volatility as \cite{Guyon}
\begin{equation}
F_V(T) = F_{0}(T) \left(1 + c_1 \omega^2 + c_2 \omega^4 + O\left(\omega^6\right)\right)\,,
\end{equation}
where the coefficients $c_i$ are given in equation (2.10) in \cite{Guyon} for the one-factor Bergomi model and (3.10) in \cite{Guyon} for the two-factor Bergomi model.

\subsection{The short-maturity regime}\label{sec:short}

In this section, we consider the short-maturity regime as $T\rightarrow 0$.

\subsubsection{One-factor Bergomi model}

First, let us consider the one-factor Bergomi model.
We recall from \eqref{xi:one} that
\begin{equation}
\xi_{T}^{u}=\xi_{0}^{u}e^{\omega e^{-k(u-T)}X_{T}-\frac{\omega^{2}}{2}e^{-2k(u-T)}v_{T}},
\end{equation}
where $(X_{t})_{t\geq 0}$ is the Ornstein-Uhlenbeck process given in \eqref{X:one} and $v_{T}$ is the variance of $X_{T}$ given in \eqref{v:one}.
Since $v_{T}\rightarrow 0$ and $X_{T}\rightarrow 0$ a.s. as $T\rightarrow 0$, we have that
\begin{equation}
\mathrm{VIX}_{T}^{2}=\frac{1}{\tau}\int_{T}^{T+\tau}\xi_{T}^{u}du
=\frac{1}{\tau}\int_{T}^{T+\tau}\xi_{0}^{u}e^{\omega e^{-k(u-T)}X_{T}-\frac{\omega^{2}}{2}e^{-2k(u-T)}v_{T}}du
\rightarrow\frac{1}{\tau}\int_{0}^{\tau}\xi_{0}^{u}du,
\end{equation}
a.s. as $T\rightarrow 0$.

Therefore, in the short-maturity regime ($T\rightarrow 0$), 
the VIX futures price approaches 
\begin{align}
\lim_{T\to 0} F_V(T) = F_0(0) = \sqrt{\frac{1}{\tau}\int_{0}^{\tau}\xi_{0}^{u}du}\,,
\end{align}
where $F_0(T)$ is defined in \eqref{F0def}. We denote $F_0(0)=F_0$ for simplicity.
Thus, in the short-maturity limit, a VIX call option is OTM if $K > F_0$, 
ATM if $K=F_0$
and ITM if $K< F_0$; 
a VIX put option is OTM if $K < F_0$, 
ATM if $K=F_0$
and ITM if $K> F_0$.

We have the following result that provides the leading-order asymptotics
for OTM VIX options in the short-maturity regime.

\begin{theorem}\label{thm:short:OTM:one:factor}
Denote $F_0=\sqrt{\frac{1}{\tau}\int_{0}^{\tau}\xi_{0}^{u}du}$.
In the one-factor Bergomi model, we have the following short-maturity asymptotics for OTM VIX options.

(i) For an OTM VIX call option $K > F_0$, we have
\begin{equation}
\lim_{T\rightarrow 0}T\log C(T,\omega)=-\frac{1}{2}x_{+}^{2},
\end{equation}
where $x_{+}>0$ is the unique positive value such that
\begin{equation}
\frac{1}{\tau}\int_{0}^{\tau}\xi_{0}^{u}e^{\omega e^{-ku}x_{+}}du=K^{2}.
\end{equation}

(ii) For an OTM VIX put option $K < F_0$, we have
\begin{equation}
\lim_{T\rightarrow 0}T\log P(T,\omega)=-\frac{1}{2}x_{-}^{2},
\end{equation}
where $x_{-}<0$ is the unique negative value such that
\begin{equation}
\frac{1}{\tau}\int_{0}^{\tau}\xi_{0}^{u}e^{\omega e^{-ku}x_{-}}du=K^{2}.
\end{equation}
\end{theorem}

We have the following result that provides the leading-order asymptotics
for ATM VIX options in the short-maturity regime.

\begin{theorem}\label{thm:short:ATM:one:factor}
Assume that there exists some $C>0$ such that
$\sup_{0\leq u\leq\tau}|\xi_{0}^{T+u}-\xi_{0}^{u}|\leq CT$ for any sufficiently small $T$.
If $K = F_{0}$, then
\begin{equation}
\lim_{T\rightarrow 0}\frac{C(T,\omega)}{\sqrt{T}}
=\lim_{T\rightarrow 0}\frac{P(T,\omega)}{\sqrt{T}}
=\frac{\frac{1}{\tau}\int_{0}^{\tau}\xi_{0}^{u}\omega e^{-ku}du}{2\sqrt{2\pi}\left(\frac{1}{\tau}\int_{0}^{\tau}\xi_{0}^{u}du\right)^{1/2}}.
\end{equation}
\end{theorem}


\subsubsection{Two-factor Bergomi model}

Next, let us consider the two-factor Bergomi model.
We recall from \eqref{xi:two} that
\begin{equation}
\xi_{T}^{u}=\xi_{0}^{u}e^{\omega x_{T}^{u}-\frac{\omega^{2}}{2}v_{T}(u)}, 
\end{equation}
where
\begin{equation}
x_{T}^{u}:=\alpha_{\theta}\left(\theta_{1}e^{-k_{1}(u-T)}X_{T}^{1}+\theta_{2}e^{-k_{2}(u-T)}X_{T}^{2}\right),
\end{equation}
where $(X_{t}^{i})_{t\geq 0}$ are the Ornstein-Uhlenbeck processes defined in \eqref{X:two} and
\begin{align}
v_{T}(u):=\mathrm{Var}(x_{T}^{u})=\alpha_{\theta}^{2}\left(\theta_{1}^{2}e^{-2k_{1}(u-T)}v_{T}^{1}
+\theta_{2}^{2}e^{-2k_{2}(u-T)}v_{T}^{2}+2\theta_{1}\theta_{2}e^{-(k_{1}+k_{2})(u-T)}v_{T}^{1,2}\right),
\end{align}
where $v_{T}^{1}$, $v_{T}^{2}$, $v_{T}^{1,2}$ are defined in \eqref{v:i:two}.

Since $v_{T}^{1},v_{T}^{2},v_{T}^{1,2}\rightarrow 0$ and $X_{T}^{1},X_{T}^{2}\rightarrow 0$ a.s. as $T\rightarrow 0$, we have that
\begin{equation}
\mathrm{VIX}_{T}^{2}=\frac{1}{\tau}\int_{T}^{T+\tau}\xi_{T}^{u}du
\rightarrow\frac{1}{\tau}\int_{0}^{\tau}\xi_{0}^{u}du=F_{0}^{2},
\end{equation}
a.s. as $T\rightarrow 0$.

Therefore, in the short-maturity regime ($T\rightarrow 0$), 
the VIX call option is OTM if $F_{0}<K$, 
ATM if $F_{0}=K$
and ITM if $F_{0}>K$; 
the VIX put option is OTM if $F_{0}>K$, 
ATM if $F_{0}=K$
and ITM if $F_{0}<K$.

We have the following result that provides the leading-order asymptotics
for OTM VIX options in the short-maturity regime.

\begin{theorem}\label{thm:short:OTM:two:factor}
Denote $F_0=\sqrt{\frac{1}{\tau}\int_{0}^{\tau}\xi_{0}^{u}du}$.
In the two-factor Bergomi model we have the following short-maturity asymptotics for OTM VIX options.

(i) For an OTM VIX call option $K>F_0$ we have
\begin{equation}
\lim_{T\rightarrow 0}T\log C(T,\omega)=-\inf_{x_{1},x_{2}:\frac{1}{\tau}\int_{0}^{\tau}\xi_{0}^{u}e^{\omega\alpha_{\theta}(\theta_{1}e^{-k_{1}u}x_{1}+\theta_{2}e^{-k_{2}u}x_{2})}du=K^{2}}\frac{x_{1}^{2}+x_{2}^{2}-2x_{1}x_{2}\rho}{2(1-\rho^{2})}.
\end{equation}

(ii) For an OTM VIX put option $K<F_0$ we have
\begin{equation}
\lim_{T\rightarrow 0}T\log P(T,\omega)=-\inf_{x_{1},x_{2}:\frac{1}{\tau}\int_{0}^{\tau}\xi_{0}^{u}e^{\omega\alpha_{\theta}(\theta_{1}e^{-k_{1}u}x_{1}+\theta_{2}e^{-k_{2}u}x_{2})}du=K^{2}}\frac{x_{1}^{2}+x_{2}^{2}-2x_{1}x_{2}\rho}{2(1-\rho^{2})}.
\end{equation}
\end{theorem}

We have the following result that provides the leading-order asymptotics
for ATM VIX options in the short-maturity regime.

\begin{theorem}\label{thm:short:ATM:two:factor}
Assume that there exists some $C>0$ such that
$\sup_{0\leq u\leq\tau}|\xi_{0}^{T+u}-\xi_{0}^{u}|\leq CT$ for any sufficiently small $T$.
If $F_{0}=K$, then
\begin{equation}
\lim_{T\rightarrow 0}\frac{C(T,\omega)}{\sqrt{T}}
=\lim_{T\rightarrow 0}\frac{P(T,\omega)}{\sqrt{T}}
=\frac{(v(\tau))^{1/2}}{2\sqrt{2\pi}\left(\frac{1}{\tau}\int_{0}^{\tau}\xi_{0}^{u}du\right)^{1/2}},
\end{equation}
where
\begin{align}
v(\tau)
&:=\left(\frac{1}{\tau}\int_{0}^{\tau}\xi_{0}^{u}\omega\alpha_{\theta}\theta_{1}e^{-k_{1}u}du\right)^{2}
+\left(\frac{1}{\tau}\int_{0}^{\tau}\xi_{0}^{u}\omega\alpha_{\theta}\theta_{2}e^{-k_{2}u}du\right)^{2}
\nonumber
\\
&\qquad+2\left(\frac{1}{\tau}\int_{0}^{\tau}\xi_{0}^{u}\omega\alpha_{\theta}\theta_{1}e^{-k_{1}u}du\right)\left(\frac{1}{\tau}\int_{0}^{\tau}\xi_{0}^{u}\omega\alpha_{\theta}\theta_{2}e^{-k_{2}u}du\right)\rho.
\end{align}
\end{theorem}


\subsubsection{$N$-factor Bergomi model}

Next, let us consider the $N$-factor Bergomi model.
We recall from \eqref{xi:N} that
\begin{equation}
\xi_{T}^{u}=\xi_{0}^{u}e^{\omega x_{T}^{u}-\frac{\omega^{2}}{2}v_{T}(u)}, 
\end{equation}
where
\begin{equation}
x_{T}^{u}:=\alpha_{\theta}\sum_{i=1}^{N}\theta_{i}e^{-k_{i}(u-T)}X_{T}^{i},
\end{equation}
where $(X_{t}^{i})_{t\geq 0}$ are Ornstein-Uhlenbeck processes given in \eqref{X:N} and
\begin{align}
v_{T}(u):=\mathrm{Var}(x_{T}^{u})=\alpha_{\theta}^{2}\left(\sum_{i=1}^{N}\theta_{i}^{2}e^{-2k_{i}(u-T)}v_{T}^{ii}
+2\sum_{1\leq i<j\leq N}\theta_{i}\theta_{j}e^{-(k_{i}+k_{j})(u-T)}v_{T}^{ij}\right),
\end{align}
where $v_{T}^{ii},v_{T}^{ij}$ are given in \eqref{v:ij:N}.

Since $v_{T}^{ii},v_{T}^{ij}\rightarrow 0$ and $X_{T}^{i}\rightarrow 0$ a.s. as $T\rightarrow 0$, we have that
\begin{equation}
\mathrm{VIX}_{T}^{2}=\frac{1}{\tau}\int_{T}^{T+\tau}\xi_{T}^{u}du
\rightarrow\frac{1}{\tau}\int_{0}^{\tau}\xi_{0}^{u}du=F_{0}^{2},
\end{equation}
a.s. as $T\rightarrow 0$.
Therefore, in the short-maturity regime ($T\rightarrow 0$), 
the VIX call option is OTM if $F_{0}<K$, 
ATM if $F_{0}=K$
and ITM if $F_{0}>K$; 
the VIX put option is OTM if $F_{0}>K$, 
ATM if $F_{0}=K$
and ITM if $F_{0}<K$.

We have the following result that provides the leading-order asymptotics
for OTM VIX options in the short-maturity regime.

\begin{theorem}\label{thm:short:OTM:N:factor}
(i) If $F_{0}<K$, then
\begin{equation}
\lim_{T\rightarrow 0}T\log C(T,\omega)=-\inf_{x=(x_{1},\ldots,x_{N}):\frac{1}{\tau}\int_{0}^{\tau}\xi_{0}^{u}e^{\omega\alpha_{\theta}\sum_{i=1}^{N}\theta_{i}e^{-k_{i}u}x_{i}}du=K^{2}}\frac{1}{2}x^{\top}\hat{\Sigma}_{0}^{-1}x,
\end{equation}
where $\hat{\Sigma}_{0}$ is an $N\times N$ symmetric matrix with $(i,j)$-th entry $\rho_{ij}1_{i=j}$.

(ii) If $F_{0}>K$, then
\begin{equation}
\lim_{T\rightarrow 0}T\log P(T,\omega)=-\inf_{x=(x_{1},\ldots,x_{N}):\frac{1}{\tau}\int_{0}^{\tau}\xi_{0}^{u}e^{\omega\alpha_{\theta}\sum_{i=1}^{N}\theta_{i}e^{-k_{i}u}x_{i}}du=K^{2}}\frac{1}{2}x^{\top}\hat{\Sigma}_{0}^{-1}x.
\end{equation}
\end{theorem}

We have the following result that provides the leading-order asymptotics
for ATM VIX options in the short-maturity regime.

\begin{theorem}\label{thm:short:ATM:N:factor}
Assume that there exists some $C>0$ such that
$\sup_{0\leq u\leq\tau}|\xi_{0}^{T+u}-\xi_{0}^{u}|\leq CT$ for any sufficiently small $T$.
If $F_{0}=K$, then
\begin{equation}
\lim_{T\rightarrow 0}\frac{C(T,\omega)}{\sqrt{T}}
=\lim_{T\rightarrow 0}\frac{P(T,\omega)}{\sqrt{T}}
=\frac{(v(\tau))^{1/2}}{2\sqrt{2\pi}\left(\frac{1}{\tau}\int_{0}^{\tau}\xi_{0}^{u}du\right)^{1/2}},
\end{equation}
where
\begin{align}
v(\tau)
&:=\sum_{i=1}^{N}\left(\frac{1}{\tau}\int_{0}^{\tau}\xi_{0}^{u}\omega\alpha_{\theta}\theta_{i}e^{-k_{i}u}du\right)^{2}
\nonumber
\\
&\qquad
+2\sum_{1\leq i<j\leq N}\left(\frac{1}{\tau}\int_{0}^{\tau}\xi_{0}^{u}\omega\alpha_{\theta}\theta_{i}e^{-k_{i}u}du\right)\left(\frac{1}{\tau}\int_{0}^{\tau}\xi_{0}^{u}\omega\alpha_{\theta}\theta_{j}e^{-k_{j}u}du\right)\rho_{ij}.
\end{align}
\end{theorem}

\subsection{The small vol-of-vol regime}\label{sec:small}

In this section, we consider the small vol-of-vol regime as $\omega\rightarrow 0$.

\subsubsection{One-factor Bergomi model}

First, let us consider the one-factor Bergomi model.
We recall from \eqref{xi:one} that
\begin{equation}
\xi_{T}^{u}=\xi_{0}^{u}e^{\omega e^{-k(u-T)}X_{T}-\frac{\omega^{2}}{2}e^{-2k(u-T)}v_{T}},
\end{equation}
where $(X_{t})_{t\geq 0}$ is the Ornstein-Uhlenbeck process given in \eqref{X:one} and $v_{T}$ is the variance of $X_{T}$ given in \eqref{v:one}.
Therefore,
\begin{align}
\mathrm{VIX}_{T}^{2}=\frac{1}{\tau}\int_{T}^{T+\tau}\xi_{T}^{u}du
&=\frac{1}{\tau}\int_{T}^{T+\tau}\xi_{0}^{u}e^{\omega e^{-k(u-T)}X_{T}-\frac{\omega^{2}}{2}e^{-2k(u-T)}v_{T}}du
\nonumber
\\
&\rightarrow\frac{1}{\tau}\int_{T}^{T+\tau}\xi_{0}^{u}du
=F^2_0(T),
\end{align}
a.s. as $\omega\rightarrow 0$.

Therefore, in the small vol-of-vol regime ($\omega\rightarrow 0$), 
the VIX call option is OTM if $F_{0}(T)<K$, 
ATM if $F_{0}(T)=K$
and ITM if $F_{0}(T)>K$; 
the VIX put option is OTM if $F_{0}(T)>K$, 
ATM if $F_{0}(T)=K$
and ITM if $F_{0}(T)<K$.

We have the following result that provides the leading-order asymptotics
for OTM VIX options in the small vol-of-vol regime.

\begin{theorem}\label{thm:small:OTM:one:factor}
(i) If $F_{0}(T)<K$, then
\begin{equation}
\lim_{\omega\rightarrow 0}\omega^{2}\log C(T,\omega)=-\frac{ky_{+}^{2}}{1-e^{-2kT}},
\end{equation}
where $y_{+}>0$ is the unique positive value such that
\begin{equation}
\frac{1}{\tau}\int_{T}^{T+\tau}\xi_{0}^{u}\exp\left(e^{-k(u-T)}y_{+}\right)du=K^{2}.
\end{equation}

(ii) If $F_{0}(T)>K$, then
\begin{equation}
\lim_{\omega\rightarrow 0}\omega^{2}\log P(T,\omega)=-\frac{ky_{-}^{2}}{1-e^{-2kT}},
\end{equation}
where $y_{-}<0$ is the unique negative value such that
\begin{equation}
\frac{1}{\tau}\int_{T}^{T+\tau}\xi_{0}^{u}\exp\left(e^{-k(u-T)}y_{-}\right)du=K^{2}.
\end{equation}
\end{theorem}

We have the following result that provides the leading-order asymptotics
for ATM VIX options in the small vol-of-vol regime.

\begin{theorem}\label{thm:small:ATM:one:factor}
If $F_{0}(T)=K$, then
\begin{equation}
\lim_{\omega\rightarrow 0}\frac{C(T,\omega)}{\omega}
=\lim_{\omega\rightarrow 0}\frac{P(T,\omega)}{\omega}
=\frac{\int_{T}^{T+\tau}\xi_{0}^{u}e^{-k(u-T)}du}{2\sqrt{2\pi\tau}\left(\int_{T}^{T+\tau}\xi_{0}^{u}du\right)^{1/2}}
\left(\frac{1-e^{-2kT}}{2k}\right)^{1/2}.
\end{equation}
\end{theorem}


\subsubsection{Two-factor Bergomi model}

Next, let us consider the two-factor Bergomi model.
We recall from \eqref{xi:two} hat
\begin{equation}
\xi_{T}^{u}=\xi_{0}^{u}e^{\omega x_{T}^{u}-\frac{\omega^{2}}{2}v_{T}(u)}, 
\end{equation}
where
\begin{equation}
x_{T}^{u}:=\alpha_{\theta}\left(\theta_{1}e^{-k_{1}(u-T)}X_{T}^{1}+\theta_{2}e^{-k_{2}(u-T)}X_{T}^{2}\right),
\end{equation}
where $(X_{t}^{i})_{t\geq 0}$ are the Ornstein-Uhlenbeck processes defined in \eqref{X:two} and
\begin{align}
v_{T}(u):=\mathrm{Var}(x_{T}^{u})=\alpha_{\theta}^{2}\left(\theta_{1}^{2}e^{-2k_{1}(u-T)}v_{T}^{1}
+\theta_{2}^{2}e^{-2k_{2}(u-T)}v_{T}^{2}+2\theta_{1}\theta_{2}e^{-(k_{1}+k_{2})(u-T)}v_{T}^{1,2}\right),
\end{align}
where $v_{T}^{1}$, $v_{T}^{2}$, $v_{T}^{1,2}$ are defined in \eqref{v:i:two}.
Therefore,
\begin{equation}
\mathrm{VIX}_{T}^{2}=\frac{1}{\tau}\int_{T}^{T+\tau}\xi_{T}^{u}du
\rightarrow\frac{1}{\tau}\int_{T}^{T+\tau}\xi_{0}^{u}du
=F_{0}^{2}(T),
\end{equation}
a.s. as $\omega\rightarrow 0$.

Therefore, in the small vol-of-vol regime ($\omega\rightarrow 0$), 
the VIX call option is OTM if $F_{0}(T)<K$, 
ATM if $F_{0}(T)=K$
and ITM if $F_{0}(T)>K$; 
the VIX put option is OTM if $F_{0}(T)>K$, 
ATM if $F_{0}(T)=K$
and ITM if $F_{0}(T)<K$.

We have the following result that provides the leading-order asymptotics
for OTM VIX options in the small vol-of-vol regime.

\begin{theorem}\label{thm:small:OTM:two:factor}
(i) If $F_{0}(T)<K$, then
\begin{equation}
\lim_{\omega\rightarrow 0}\omega^{2}\log C(T,\omega)=-\inf_{x_{1},x_{2}:\frac{1}{\tau}\int_{T}^{T+\tau}\xi_{0}^{u}e^{\alpha_{\theta}(\theta_{1}e^{-k_{1}(u-T)}x_{1}+\theta_{2}e^{-k_{2}(u-T)}x_{2})}du=K^{2}}\frac{v_{T}^{2}x_{1}^{2}+v_{T}^{1}x_{2}^{2}-2x_{1}x_{2}v_{T}^{1,2}}{2\left(v_{T}^{1}v_{T}^{2}-\left(v_{T}^{1,2}\right)^{2}\right)},
\end{equation}
where
\begin{align}\label{v:T:i}
v_{T}^{i}:=\frac{1-e^{-2k_{i}T}}{2k_{i}},\quad i=1,2,
\qquad
v_{T}^{1,2}:=\rho\frac{1-e^{-(k_{1}+k_{2})T}}{k_{1}+k_{2}}.
\end{align}

(ii) If $F_{0}(T)>K$, then
\begin{equation}
\lim_{\omega\rightarrow 0}\omega^{2}\log P(T,\omega)=-\inf_{x_{1},x_{2}:\frac{1}{\tau}\int_{T}^{T+\tau}\xi_{0}^{u}e^{\alpha_{\theta}(\theta_{1}e^{-k_{1}(u-T)}x_{1}+\theta_{2}e^{-k_{2}(u-T)}x_{2})}du=K^{2}}\frac{v_{T}^{2}x_{1}^{2}+v_{T}^{1}x_{2}^{2}-2x_{1}x_{2}v_{T}^{1,2}}{2\left(v_{T}^{1}v_{T}^{2}-\left(v_{T}^{1,2}\right)^{2}\right)},
\end{equation}
where $v_{T}^{1},v_{T}^{2},v_{T}^{1,2}$ are given in \eqref{v:T:i}.
\end{theorem}

We have the following result that provides the leading-order asymptotics
for ATM VIX options in the small vol-of-vol regime.

\begin{theorem}\label{thm:small:ATM:two:factor}
If $F_{0}(T)=K$, then
\begin{equation}
\lim_{\omega\rightarrow 0}\frac{C(T,\omega)}{\omega}
=\lim_{\omega\rightarrow 0}\frac{P(T,\omega)}{\omega}
=\frac{(\hat{v}(\tau))^{1/2}}{2\left(\frac{1}{\tau}\int_{T}^{T+\tau}\xi_{0}^{u}du\right)^{1/2}}\frac{1}{\sqrt{2\pi}},
\end{equation}
where
\begin{align}
\hat{v}(\tau)&:=\left(\frac{1}{\tau}\int_{T}^{T+\tau}\xi_{0}^{u}\alpha_{\theta}\theta_{1}e^{-k_{1}(u-T)}du\right)^{2}v_{T}^{1}
+\left(\frac{1}{\tau}\int_{T}^{T+\tau}\xi_{0}^{u}\alpha_{\theta}\theta_{2}e^{-k_{2}(u-T)}du\right)^{2}v_{T}^{2}
\nonumber
\\
&\qquad+2\left(\frac{1}{\tau}\int_{T}^{T+\tau}\xi_{0}^{u}\alpha_{\theta}\theta_{1}e^{-k_{1}(u-T)}du\right)\left(\frac{1}{\tau}\int_{T}^{T+\tau}\xi_{0}^{u}\alpha_{\theta}\theta_{2}e^{-k_{2}(u-T)}du\right)v_{T}^{1,2},
\end{align}
where $v_{T}^{1},v_{T}^{2},v_{T}^{1,2}$ are given in \eqref{v:T:i}.
\end{theorem}


\subsubsection{$N$-factor Bergomi model}

Next, let us consider the $N$-factor Bergomi model.
We recall from \eqref{xi:N} that
\begin{equation}
\xi_{T}^{u}=\xi_{0}^{u}e^{\omega x_{T}^{u}-\frac{\omega^{2}}{2}v_{T}(u)}, 
\end{equation}
where
\begin{equation}
x_{T}^{u}:=\alpha_{\theta}\sum_{i=1}^{N}\theta_{i}e^{-k_{i}(u-T)}X_{T}^{i},
\end{equation}
where $(X_{t}^{i})_{t\geq 0}$ are Ornstein-Uhlenbeck processes given in \eqref{X:N} and
\begin{align}
v_{T}(u):=\mathrm{Var}(x_{T}^{u})=\alpha_{\theta}^{2}\left(\sum_{i=1}^{N}\theta_{i}^{2}e^{-2k_{i}(u-T)}v_{T}^{ii}
+2\sum_{1\leq i<j\leq N}\theta_{i}\theta_{j}e^{-(k_{i}+k_{j})(u-T)}v_{T}^{ij}\right),
\end{align}
where $v_{T}^{ii},v_{T}^{ij}$ are given in \eqref{v:ij:N}.
Therefore,
\begin{equation}
\mathrm{VIX}_{T}^{2}=\frac{1}{\tau}\int_{T}^{T+\tau}\xi_{T}^{u}du
\rightarrow\frac{1}{\tau}\int_{T}^{T+\tau}\xi_{0}^{u}du=F_{0}^{2}(T),
\end{equation}
a.s. as $\omega\rightarrow 0$.
Therefore, in the small vol-of-vol regime ($\omega\rightarrow 0$), 
the VIX call option is OTM if $F_{0}(T)<K$, 
ATM if $F_{0}(T)=K$
and ITM if $F_{0}(T)>K$; 
the VIX put option is OTM if $F_{0}(T)>K$, 
ATM if $F_{0}(T)=K$
and ITM if $F_{0}(T)<K$.

We have the following result that provides the leading-order asymptotics
for OTM VIX options in the small vol-of-vol regime.

\begin{theorem}\label{thm:small:OTM:N:factor}
(i) If $F_{0}(T)<K$, then
\begin{equation}
\lim_{\omega\rightarrow 0}\omega^{2}\log C(T,\omega)=-\inf_{x=(x_{1},\ldots,x_{N}):\frac{1}{\tau}\int_{T}^{T+\tau}\xi_{0}^{u}e^{\alpha_{\theta}\sum_{i=1}^{N}\theta_{i}e^{-k_{i}(u-T)}x_{i}}du=K^{2}}\frac{1}{2}x^{\top}\Sigma_{T}^{-1}x,
\end{equation}
where $\Sigma_{T}:=(v_{T}^{ij})_{1\leq i,j\leq N}$ with
\begin{align}\label{v:T:i:j}
v_{T}^{ii}:=\frac{1-e^{-2k_{i}T}}{2k_{i}},\quad i=1,2,\ldots,N,
\qquad
v_{T}^{ij}:=\rho\frac{1-e^{-(k_{i}+k_{j})T}}{k_{i}+k_{j}},\quad i\neq j.
\end{align}

(ii) If $F_{0}(T)>K$, then
\begin{equation}
\lim_{\omega\rightarrow 0}\omega^{2}\log P(T,\omega)=-\inf_{x=(x_{1},\ldots,x_{N}):\frac{1}{\tau}\int_{T}^{T+\tau}\xi_{0}^{u}e^{\alpha_{\theta}\sum_{i=1}^{N}\theta_{i}e^{-k_{i}(u-T)}x_{i}}du=K^{2}}\frac{1}{2}x^{\top}\Sigma_{T}^{-1}x.
\end{equation}
\end{theorem}

We have the following result that provides the leading-order asymptotics
for ATM VIX options in the small vol-of-vol regime.

\begin{theorem}\label{thm:small:ATM:N:factor}
If $F_{0}(T)=K$, then
\begin{equation}
\lim_{\omega\rightarrow 0}\frac{C(T,\omega)}{\omega}
=\lim_{\omega\rightarrow 0}\frac{P(T,\omega)}{\omega}
=\frac{(\hat{v}(\tau))^{1/2}}{2\left(\frac{1}{\tau}\int_{T}^{T+\tau}\xi_{0}^{u}du\right)^{1/2}}\frac{1}{\sqrt{2\pi}},
\end{equation}
where
\begin{align}
\hat{v}(\tau)&:=\sum_{i=1}^{N}\left(\frac{1}{\tau}\int_{T}^{T+\tau}\xi_{0}^{u}\alpha_{\theta}\theta_{i}e^{-k_{i}(u-T)}du\right)^{2}v_{T}^{ii}
\nonumber
\\
&\qquad+2\sum_{1\leq i<j\leq N}\left(\frac{1}{\tau}\int_{T}^{T+\tau}\xi_{0}^{u}\alpha_{\theta}\theta_{i}e^{-k_{i}(u-T)}du\right)\left(\frac{1}{\tau}\int_{T}^{T+\tau}\xi_{0}^{u}\alpha_{\theta}\theta_{j}e^{-k_{j}(u-T)}du\right)v_{T}^{ij},
\end{align}
where $v_{T}^{ii},v_{T}^{ij}$ are given in \eqref{v:T:i:j}.
\end{theorem}

\section{VIX Implied Volatility in the Bergomi Model}
\label{sec:VIXsmile}

In this section, we study more closely the asymptotics of VIX options in the Bergomi model in the two limits considered above: the short maturity limit $T\to 0$ and the small vol-of-vol limit $\omega\to 0$. We obtain more explicit expressions for the rate functions of VIX options in these two limits and derive asymptotic predictions for the implied volatility of the VIX options. 

More precisely we give predictions for the ATM level and the skew of the VIX implied volatility, which are of interest for practical applications and have been studied in the literature \cite{Alos2022VIX}.

We consider separately the predictions from the short-maturity regime discussed in Section~\ref{sec:short} and from the small volatility-of-volatility regime discussed in Section~\ref{sec:small}. They lead to different predictions for the asymptotic VIX implied volatility. 
They are related to the exact VIX implied volatility $\sigma_{\mathrm{VIX}}(K,T;\omega)$ as
\begin{equation}\label{3.1}
\lim_{T\to 0} \sigma_{\mathrm{VIX}}(K,T;\omega) = \sigma_{\mathrm{VIX}}^{(t)}(K;\omega)\,,
\end{equation}
and 
\begin{equation}\label{3.2}
\lim_{\omega\to 0} \frac{1}{\omega}\sigma_{\mathrm{VIX}}(K,T;\omega) = \sigma_{\mathrm{VIX}}^{(v)}(K,T)\,,
\end{equation}
respectively, where the functions $\sigma_{\mathrm{VIX}}^{(t)}(K;\omega)$ and
$\sigma_{\mathrm{VIX}}^{(v)}(K,T)$ are calculable and can be expressed in terms of the relevant rate functions. The superscripts $t$ and $v$ are a reminder of the short maturity and small vol-of-vol limits considered. 

We will show that $\sigma_{\mathrm{VIX}}^{(t)}(K;\omega)$ is linear in $\omega$, which is consistent with \eqref{3.2}. Furthermore, taking the $T\to 0$ limit in \eqref{3.2} reproduces \eqref{3.1}: $\lim_{T\to 0} \sigma_{\mathrm{VIX}}^{(v)}(K,T) = \frac{1}{\omega} \sigma_{\mathrm{VIX}}^{(t)}(K;\omega)$.
Thus, the result \eqref{3.2} embeds both predictions, and is the most predictive for practical applications as it keeps the full maturity dependence. 

For both cases, we restrict ourselves to the two-factor Bergomi model, which is widely used in practice. The results can be extended straightforwardly to the $N$-factor model. We will give explicit expressions for $\sigma_{\mathrm{VIX}}^{(v)}(K,T)$ and $\sigma_{\mathrm{VIX}}^{(t)}(K;\omega)$ in an expansion in log-moneyness $x_{\mathrm{VIX}} = \log(K/F_V(T))$.

\subsection{The short-maturity regime}

Denote $J_V(K)$ the rate function for the short-maturity regime for OTM VIX options. For the two-factor Bergomi model this is given in Theorem~\ref{thm:short:OTM:two:factor}.
Recall that in the short-maturity limit a VIX option is ATM when $K=F_0$ with
\begin{equation}
F_0^2 = \frac{1}{\tau} \int_0^\tau \xi_0^u du\,.
\end{equation}
Denote the similar integrals
\begin{align}\label{Kabdef}
F_{a,i}^2 &:= \frac{1}{\tau} \int_0^\tau e^{-k_i u} \xi_0^u du\,,\quad i = 1,2\,,\\
F_{b,ij}^2 &:= \frac{1}{\tau} \int_0^\tau e^{-(k_i+k_j) u} \xi_0^u du\,,\quad i,j=1,2\,.
\end{align}

The rate function for VIX options can be expanded around the ATM point in powers of log-moneyness
$x_{\mathrm{VIX}}=\log(K/F_0)$ as
\begin{equation}\label{JVexp}
J_V(K) = j_0 x_{\mathrm{VIX}}^2 + j_1 x_{\mathrm{VIX}}^3 + O\left(x_{\mathrm{VIX}}^4\right)\,.
\end{equation}
The next result gives explicit results for the first two coefficients in this expansion.

\begin{proposition}\label{prop:JV2F}
The first two coefficients in the expansion of the rate function for VIX options in the two-factor Bergomi model in the short-maturity regime are
\begin{align}
j_0 &= 2 \left( \frac{F_0^2}{\omega \alpha_\theta} \right)^2 \cdot \frac{1}{D}\,, \\
\label{j1sol}
j_1 &= \frac{4 F_0^4}{(\omega\alpha_\theta)^2 D} \left( 1- \frac{F_0^2 G}{D^2} \right) \,,
\end{align}
where
\begin{align}\label{Ddef}
D &:= \left(\theta_1 F^2_{a1}\right)^2 + 2\rho \left(\theta_1 F^2_{a1}\right)\left(\theta_2 F^2_{a2}\right) + 
\left(\theta_2 F^2_{a2}\right)^2\,, \\
\label{Gdef}
G &:= \theta_1^2 F_{b11}^2 \left(\theta_1 F_{a1}^2 + \rho \theta_2 F_{a2}^2\right)^2 \nonumber  \\
&\qquad\qquad+ 2 \theta_1 \theta_2 F_{b12}^2 \left(\theta_1 F_{a1}^2 + \rho \theta_2 F_{a2}^2\right)
\left(\rho \theta_1 F_{a1}^2 +  \theta_2 F_{a2}^2\right)+ \theta_2^2 F_{b22}^2 \left(\rho\theta_1 F_{a1}^2 +  \theta_2 F_{a2}^2\right)^2\,.
\end{align}
\end{proposition}

From these results, we can obtain the corresponding expressions in the one-factor Bergomi model. 
This is obtained by taking $\theta_1=1,\theta_2=0$, which gives $\alpha_\theta=1$ and
\begin{equation}
D = F_a^4\,,\quad G = F_b^2 F_a^4\,,
\end{equation}
where $F_a^2, F_b^2$ are defined as in \eqref{Kabdef}.

\begin{corollary}
The expansion coefficients of the VIX rate function for the one-factor Bergomi model in the short-maturity regime are
\begin{align}
j_0 &= 2 \frac{F_0^4}{\omega^2 F_a^4}\,, \\
j_1 &= \frac{4F_0^4}{\omega^2 F_a^4} \left(1 - \frac{F_0^2 F_b^2}{F_a^4}\right)\,.
\end{align}
\end{corollary}

\subsubsection{Asymptotic VIX implied volatility}

The short-maturity asymptotics of VIX options can be expressed as a limiting result for the 
VIX implied volatility, see for example \cite{VIXpaper}
\begin{equation}\label{sigVIX}
\lim_{T\to 0} \sigma_{VIX}(K,T) := \sigma_{\mathrm{VIX}}^{(t)2}(K) = \frac{x_{\mathrm{VIX}}^2}{2J_V(K)}\,.
\end{equation}
We denote $\sigma_{\mathrm{VIX}}^{(t)2}(K)$ the  short-maturity asymptotic VIX implied volatility.
Using Proposition~\ref{prop:JV2F} 
we can compute the asymptotic VIX implied volatility in an expansion around the ATM point. 
The first few terms in this expansion in log-moneyness are
\begin{equation}
\sigma_{\mathrm{VIX}}^{(t)}(K) = \sigma_{\mathrm{VIX}}^{(t)}(0) + s_{\mathrm{VIX}}^{(t)} x_{\mathrm{VIX}} + O(x_{\mathrm{VIX}}^2)\,,
\end{equation}
where $\sigma_{\mathrm{VIX}}^{(t)}(0)$ is the ATM VIX volatility level and
$s_{\mathrm{VIX}}^{(t)}$ is the ATM VIX skew. They can be found in closed form and are given in the following result.

\begin{proposition}\label{prop:sigVIX}
The short-maturity asymptotics for the 
ATM VIX level and skew in the two-factor Bergomi model in the short-maturity regime are
\begin{equation}\label{sigVIXatm2F}
\sigma_{\mathrm{VIX}}^{(t)}(0) = \frac12 \frac{\omega \alpha_\theta}{F_0^2} \sqrt{D}\,,
\end{equation}
and
\begin{equation}\label{VIXskew2F}
s_{\mathrm{VIX}}^{(t)}= - \sigma_{\mathrm{VIX}}(0) \left(1 - F_0^2 \frac{G}{D^2} \right)\,,
\end{equation}
where the functions $D,G$ are given in \eqref{Ddef} and \eqref{Gdef}, respectively.
\end{proposition}

The result \eqref{sigVIXatm2F} agrees with the prediction of 
Theorem~\ref{thm:short:ATM:two:factor}.
We can study also the limiting form of these results in the one-factor Bergomi model. 

\begin{corollary}\label{corr:VIX1F}
In the one-factor Bergomi model we have the following predictions in the short-maturity regime.
The asymptotic  ATM VIX implied volatility is
\begin{equation}\label{sigVIXATM1F}
\sigma_{\mathrm{VIX}}^{(t)}(0)  = \frac12\omega \frac{F_a^2}{F_0^2}\,,
\end{equation}
and the ATM VIX skew is
\begin{equation}\label{sVIX1F}
s_{\mathrm{VIX}}^{(t)} = - \sigma_{\mathrm{VIX}}(0) \left( 1 - \frac{F_0^2 F_b^2}{(F_a^2)^2} \right)\,. 
\end{equation}
\end{corollary}

The result \eqref{sigVIXATM1F} agrees with the prediction of Theorem~\ref{thm:short:ATM:one:factor}, since we have
\begin{equation}
\lim_{T\to 0} \frac{C_V(K)}{\sqrt{T}} = \frac{1}{\sqrt{2\pi}} F_0 \sigma_{\mathrm{VIX}}(0) =
\frac{1}{2\sqrt{2\pi}} \frac{\omega F_a^2}{F_0}\,.
\end{equation}


\subsubsection{The sign of the ATM VIX skew in the Bergomi model}
\label{sec:3.1.2}

We will prove that for $\rho \geq 0$ the ATM VIX skew in the 
two-factor Bergomi model in the short-maturity regime is non-negative. 
This result follows from the following inequalities.

\begin{lemma}\label{lemma:2}
The integrals $F_{ai}^2, F_{b,ij}^2$ defined in \eqref{Kabdef} satisfy the inequalities
\begin{equation}
F_{a1}^2 F_{a2}^2 \leq F_0^2 F_{b,12}^2\,,
\end{equation}
and
\begin{equation}
F_{a,i}^4 \leq F_0^2 F_{b,ii}^2\,,\quad i= 1,2\,.
\end{equation}
\end{lemma}

Using this result we can prove the following result.

\begin{corollary}\label{corr:D2G}
Assume $\rho \geq 0$. The asymptotic ATM VIX skew in the two-factor Bergomi model
is non-negative $s_{\mathrm{VIX}}\geq 0$.
\end{corollary}

A similar result holds in the one-factor Bergomi model.
For this case the VIX skew is given in \eqref{sVIX1F}. 
From Lemma~\ref{lemma:2} we have the inequality
\begin{equation}
F_a^4 \leq F_0^2 F_b^2\,.
\end{equation}
This implies $1 - \frac{F_0^2 F_b^2}{(F_a^2)^2}\leq 0$  and thus the ATM VIX skew is positive also in the one-factor Bergomi model. 

\subsection{The small volatility of volatility regime}
\label{sec:3.2}

The small volatility-of-volatility limit for OTM VIX option prices in the Bergomi model 
is given by Theorem~\ref{thm:small:OTM:one:factor} for the one-factor model, by Theorem~\ref{thm:small:OTM:two:factor} for the two-factor model and by Theorem~\ref{thm:small:OTM:N:factor} for the $N$-factor model. These results can be written in common form as 
\begin{equation}\label{thm:vov}
\lim_{\omega\to 0}\omega^2 \log C(K,T;\omega) = - J_{\rm vov} (K,T)\,,
\end{equation}
where $J_{\rm vov}(K,T)$ is a rate function which depends only on the 
log-moneyness $x = \log(K/F_0(T))$ of the VIX option and maturity $T$. The explicit form of the rate function is different for the one-factor, two-factor, $N$-factor models. This rate function was given above as the solution of an extremal problem.

Recall that in the small vol-of-vol limit a VIX option is ATM when $K=F_0(T)$ with
\begin{equation}
F_0^2(T) = \frac{1}{\tau} \int_T^{T+\tau} \xi_0^u du\,.
\end{equation}
Denote the integrals
\begin{align}\label{KabTdef}
F_{a,i}^2(T) &:= \frac{1}{\tau} \int_T^{T+\tau} e^{-k_i (u-T)} \xi_0^u du\,,\quad i = 1,2\,,\\
F_{b,ij}^2(T) &:= \frac{1}{\tau} \int_T^{T+\tau} e^{-(k_i+k_j) (u-T)} \xi_0^u du\,,\quad i,j=1,2\,.
\end{align}

The rate function for VIX options can be expanded around the ATM point in powers of log-moneyness
$x_{\mathrm{VIX}}=\log(K/F_0(T))$ as
\begin{equation}\label{Jomegaexp}
J_{\rm vov}(K) = j_0^{\rm vov} x_{\mathrm{VIX}}^2 + j_1^{\rm vov} x_{\mathrm{VIX}}^3 + O\left(x_{\mathrm{VIX}}^4\right)\,.
\end{equation}
We give next explicit results for the first two coefficients in this expansion.

\begin{proposition}\label{prop:JV2Fvov}
The first two coefficients in the expansion of the rate function $J_{\rm vov}(K,T)$ for VIX options in the two-factor Bergomi model in the small vol-of-vol regime are
\begin{align}
j_0^{\rm vov} &= 2 \left( \frac{F_0^2(T)}{ \alpha_\theta} \right)^2 \cdot \frac{1}{D_{\rm vov}(T)}\,, \\
\label{j1omegasol}
j_1^{\rm vov} &= \frac{4 F_0^4(T)}{(\alpha_\theta)^2 D_{\rm vov}(T)} \left( 1- \frac{F_0^2(T) G_{\rm vov}(T)}{D_{\rm vov}^2(T)} \right) \,,
\end{align}
where 
\begin{align}\label{DTdef}
D_{\rm vov}(T) &:= \theta_1^2 F^4_{a1}(T)v^1_T + 2\theta_1 \theta_2 F^2_{a1}(T) F^2_{a2}(T) v^{1,2}_T + 
\theta_2^2 F^4_{a2}(T) v^2_T\,, \\
\label{GTdef}
G_{\rm vov}(T) &:= \theta_1^2 F_{b11}^2(T) \left(\theta_1 F_{a1}^2(T) v^1_T + \theta_2 F_{a2}^2(T) v^{1,2}_T\right)^2 \nonumber  \\
&\quad+ 2 \theta_1 \theta_2 F_{b12}^2(T) \left(\theta_1 F_{a1}^2(T) v^1_T + \theta_2 F_{a2}^2(T) v^{1,2}_T\right)
\left(\theta_1 F_{a1}^2(T) v^{1,2}_T +  \theta_2 F_{a2}^2(T) v^2_T\right)\nonumber
\\
&\quad\quad+ \theta_2^2 F_{b22}^2(T)\left(\theta_1 F_{a1}^2(T) v^{1,2}_T +  \theta_2 F_{a2}^2(T) v^2_T\right)^2\,.
\end{align}
\end{proposition}

As before, we recover the corresponding rate function for the one-factor Bergomi model by taking $\theta_1=1,\theta_2=0$. This gives
\begin{equation}
D_{\rm vov}(T)=F_a^2(T) v_T\,,\quad G_{\rm vov}(T)=F_b^2(T) F_a^4(T) (v_T)^2 \,,
\end{equation}
where $F_a^2(T),F_b^2(T)$ are defined in analogy with \eqref{KabTdef}:
\begin{align}\label{FabT}
F_{a}^2(T) &:= \frac{1}{\tau} \int_T^{T+\tau} e^{-k (u-T)} \xi_0^u du\,,\qquad
F_{b}^2(T) := \frac{1}{\tau} \int_T^{T+\tau} e^{-2k (u-T)} \xi_0^u du\,.
\end{align}

\begin{corollary}
The expansion coefficients of the VIX rate function for the one-factor Bergomi model in the small vol-of-vol limit are
\begin{align}
j_0^{\rm vov} &= 2 \frac{F_0^4(T)}{F_a^4(T)}\,, \\
j_1^{\rm vov} &= \frac{4F_0^4(T)}{F_a^4(T)} \left(1 - \frac{F_0^2(T) F_b^2(T)}{F_a^4(T)}\right)\,.
\end{align}
\end{corollary}

\subsubsection{VIX implied volatility in the small vol-of-vol limit}

We derive here the implications of this result for the implied volatility of the VIX options $\sigma_{\mathrm{VIX}}(K,T)$ in the small vol-of-vol limit $\omega\to 0$. They are expressed as the following result, in terms of the rate function $J_{\rm vov}(K,T)$ studied in the previous section.

\begin{proposition}\label{prop:limvov}
In the small vol-of-vol limit $\omega \to 0$ the implied volatility of an OTM
VIX option with log-moneyness $x=\log(K/F_0(T))$ and maturity $T$ satisfies
\begin{equation}\label{limvov}
\lim_{\omega \to 0} \frac{\omega^2}{2\sigma^2_{\mathrm{VIX}}(K,T) T} x^2 = J_{\rm vov}(K,T)\,,
\end{equation}
where $J_{\rm vov}(K,T)$ is the rate function for the VIX options in the small vov-of-vol limit.
\end{proposition}

Heuristically speaking, this implies that the VIX implied volatility is approximately linear in $\omega$ for sufficiently small vol-of-vol:
\begin{equation}\label{3.38}
\sigma_{\mathrm{VIX}}(K,T) =  \omega \frac{|\log(K/F_0(T))|}{\sqrt{2J_{\rm vov}(K,T) T}} (1 + o(1)).
\end{equation}

In practice, the vol-of-vol $\omega$ is of the order $O(1)$ so it is not necessarily very small. However, such an expansion is similar to the Bergomi-Guyon expansion for the European options implied volatility in the Bergomi models. We will test it in Section~\ref{sec:numerical}
by numerical experiments for realistic values of $\omega$.

Denote $\sigma_{\mathrm{VIX}}^{v)}(K,T)$ the leading order term in \eqref{3.38}. 
We denote it as the asymptotic VIX implied volatility in the small vol-of-vol limit. It can be approximated as a linear function in log-moneyness $x_{\mathrm{VIX}}$:
\begin{equation}
\sigma_{\mathrm{VIX}}^{(v)}(K,T) \simeq \sigma_{\mathrm{VIX}}^{(v)}(0) + s_{\mathrm{VIX}}^{(v)} x_{\mathrm{VIX}}\,.
\end{equation}

\begin{proposition}\label{prop:VIXpredictionvov}
The asymptotic prediction for the asymptotic ATM VIX implied volatility predicted from the small vol-of-vol regime is
\begin{equation}\label{sigVIXATMvov}
\sigma_{\mathrm{VIX},\mathrm{ATM}}^{(v)}(T) = \frac12 \omega \sqrt{\frac{D_{\rm vov}(T)}{T}} \frac{\alpha_\theta}{F_0^2(T)}\,,
\end{equation}
where $D_{\rm vov}(T)$ is given in \eqref{DTdef}.

The asymptotic ATM skew is
\begin{equation}
s_{\mathrm{VIX}}^{(v)}(T) = - \sigma_{\mathrm{VIX},\mathrm{ATM}}(T) \left( 1 - \frac{F_0^2(T) G_{\rm vov}(T)}{D_{\rm vov}^2(T)} \right)\,,
\end{equation}
with $G_{\rm vov}(T)$ given in \eqref{GTdef}.
\end{proposition}

The result \eqref{sigVIXATMvov} for the ATM VIX implied volatility 
agrees with the prediction of 
Theorem~\ref{thm:small:ATM:two:factor}.

We give also the limiting form of these results in the one-factor Bergomi model. 
They are obtained by the substitutions $\alpha_\theta\to 1, D_{vov}(T) \to F_a^2(T) v_T$ and $G_{\rm vov}(T) \to F_{b}^2(T) F_a^4(T) (v_T)^2$. See \eqref{FabT} for the definitions of $F_a^2(T).F_b^2(T)$. We have
\begin{equation}
\sigma_{\mathrm{VIX,ATM}}^{(v)}(T) = \frac12 \omega \frac{F_a(T)}{F_0^2(T)} \sqrt{\frac{v_T}{T}}\,,
\end{equation}
and 
\begin{equation}
s_{\mathrm{VIX}}^{(v)}(T)= - \sigma_{\mathrm{VIX}}(0,T) \left(1 - F_0^2(T)\frac{F_b^2(T)}{F_a^4(T)}\right)\,.
\end{equation}

\begin{remark}
    All the conclusions of Section~\ref{sec:3.1.2} concerning the positive sign of the ATM VIX skew in the short-maturity regime hold also in the small vol-of-vol regime. In particular, the ATM VIX skew is non-negative for all correlations $\rho \geq 0$. The proofs are similar and are omitted for simplicity.
\end{remark}


\section{Numerical Tests}\label{sec:numerical}

In this section, we compare the asymptotic results obtained in Section~\ref{sec:VIXsmile} with numerical simulations of the Bergomi model. 
As mentioned, the asymptotic results in the small vol-of-vol limit 
are more appropriate for practical applications since they retain dependence on maturity $T$. They also reduce to the short-maturity limit asymptotic results, in the limit $T\to 0$. 

Therefore, we will restrict ourselves to the small vol-of-vol asymptotic results and consider the two-factor Bergomi model which has a reasonably small number of parameters and is widely used in financial practice for modeling volatility products.

We present first the numerical predictions for the VIX implied volatility from the small vol-of-vol asymptotics, and then we will test them by a Monte Carlo simulation.

\subsection{Numerical asymptotic predictions}

We assume a flat variance curve $\xi_0^u=\xi_0$. Under this assumption, the functions $D_{\rm vov}(T)$ and $G_{\rm vov}(T)$ defined in \eqref{DTdef} and \eqref{GTdef} depend on maturity $T$ only through $v_T^i$ and $v_{T}^{1,2}$.
The VIX forward price $F_0(T) = \xi_0$ and the functions $F_{ai}^2(T), F_{b,ij}^2(T)$ are independent of maturity $T$. Their expressions are given in \eqref{KabTdef}.

\textbf{Parameters.}
Assume a constant initial variance $\xi_0^u=0.1$ and the mean reversion parameters
$k_1= 7.54$, $k_2=0.24$. 
We consider three scenarios for the weights $\theta_i$:
\begin{align}
\mbox{Case 1:  } & \theta_1=\theta_2=0.5\,, \\
\mbox{Case 2:  } & \theta_1=0.9\,,\quad \theta_2=0.1\,, \\
\mbox{Case 3:  } & \theta_1= 0.1\,,\quad \theta_2 = 0.9 \,.
\end{align}
The correlation parameter $\rho$ will be varied in the range $0\leq \rho \leq 1$.

The parameters $F_{ai}^2(T)$ and $F_{b,ij}^2(T)$ defined in \eqref{FabT} are independent of maturity in the limit of a constant $\xi_0^u$ and are evaluated as
\begin{align}
F_{ai}^2(T) = \frac{1}{2k_i} \left(1 - e^{-2k_i\tau}\right)\,,\quad
F_{b,ij}^2(T) = \frac{1}{k_i+k_j} \left(1 - e^{-(k_i+k_j)\tau}\right)\,.
\end{align}
The VIX forward is also independent of maturity $F_0(T) = \xi_0$.

The VIX implied volatility can be approximated as a linear function in log-moneyness $x_{\mathrm{VIX}}$
\begin{equation}\label{VIXlinear}
\sigma_{\mathrm{VIX,lin}}(x_{\mathrm{VIX}},T) := \sigma_{\mathrm{VIX,ATM}}^{(v)}(T) + s_{\mathrm{VIX}}^{(v)}(T) x_{\mathrm{VIX}}\,.
\end{equation}
The asymptotic ATM VIX volatility and skew in the small vol-of-vol limit are computed using Proposition~\ref{prop:VIXpredictionvov}. The numerical results are shown in Table~\ref{tab:predictions} for VIX option maturity $T=1/12$ (1 month) and 
$T=1$ (1 year), for the three scenarios of $\theta_i$ and several correlations.

\begin{table}[h!]
\centering
\caption{Small vol-of-vol asymptotic predictions for the ATM VIX volatility and skew in the two-factor Bergomi model for a VIX option with several maturities $T$.
All predictions are proportional to $\omega$, which is written explicitly.}
    \begin{tabular}{|cc|cc|cc|cc|}
    \hline
    &    & \multicolumn{2}{c|}{Case 1}
        &  \multicolumn{2}{c|}{Case 2} 
        &  \multicolumn{2}{c|}{Case 3} \\
        \hline
$T$ & $\rho$ & $\sigma_{\mathrm{VIX,ATM}}^{(v)}(T)$ 
    & $10^3 s_{\mathrm{VIX}}^{(v)}(T)$
                 & $\sigma_{\mathrm{VIX,ATM}}^{(v)}(T)$ & $10^3 s_{\mathrm{VIX}}^{(v)}(T)$ 
                 & $\sigma_{\mathrm{VIX,ATM}}^{(v)}(T)$ & $10^5 s_{\mathrm{VIX}}^{(v)}(T)$ \\
    \hline\hline
1/12 & 0.0 &  0.399$\omega$ & 0.953$\omega$ & 0.284$\omega$ &  8.627$\omega$
      &  0.488$\omega$ &  2.046$\omega$ \\
1/12 & 0.3 & 0.392$\omega$ & 1.313$\omega$ & 0.290$\omega$ & 8.009$\omega$
      & 0.482$\omega$ &  4.435$\omega$ \\
1/12 & 0.5 &  0.389$\omega$ &  1.499$\omega$ & 0.293$\omega$ &  7.668$\omega$ 
      &  0.478$\omega$ &  6.380$\omega$ \\
1/12 & 0.7 &  0.387$\omega$ & 1.654$\omega$ & 0.296$\omega$ &  7.371$\omega$
      &  0.474$\omega$ &  8.547$\omega$  \\
\hline
1 & 0.0 &  0.319$\omega$ & 0.059$\omega$ & 0.107$\omega$ &  2.225$\omega$
      &  0.439$\omega$ & 1.513$\omega$  \\
1 & 0.3 & 0.290$\omega$ & 0.102$\omega$ & 0.110$\omega$ & 2.106$\omega$ 
      & 0.427$\omega$ &  1.843$\omega$ \\
1 & 0.5 &  0.275$\omega$ & 0.132$\omega$ & 0.112$\omega$ &  2.043$\omega$
      &  0.419$\omega$ & 2.073$\omega$  \\
1 & 0.7 &  0.264$\omega$ & 0.161$\omega$ & 0.114$\omega$ &  1.989$\omega$
      &  0.412$\omega$ & 2.311$\omega$  \\
    \hline
    \end{tabular}%
  \label{tab:predictions}%
\end{table}%

\subsection{Monte Carlo simulation}

We will test the asymptotic results obtained in the previous sections against benchmarks obtained using Monte Carlo simulation of the two-factor Bergomi model.

The timeline is discretized as $t_i=i\tau$, and $i=0,1,2,\ldots, N$ with time step $\tau=1/12$ (one month). The simulation horizon is taken $N=12$, corresponding to a maximum maturity of one year.


Recall the solution of the two-factor Bergomi model given in Section~\ref{sec:2FBergomi}:
\begin{align}\label{sol2F}
\xi_t^u = \xi_0^u f^u(t,x_t^u),
\end{align}
where $f^u(t,x)$ is given in \eqref{f2Fdef} and 
$x_t^u$ depends on two correlated Ornstein-Uhlenbeck processes $X_t^1, X_t^2$ started at zero.


\begin{figure}[h!]
\centering
\includegraphics[width=6.00in]{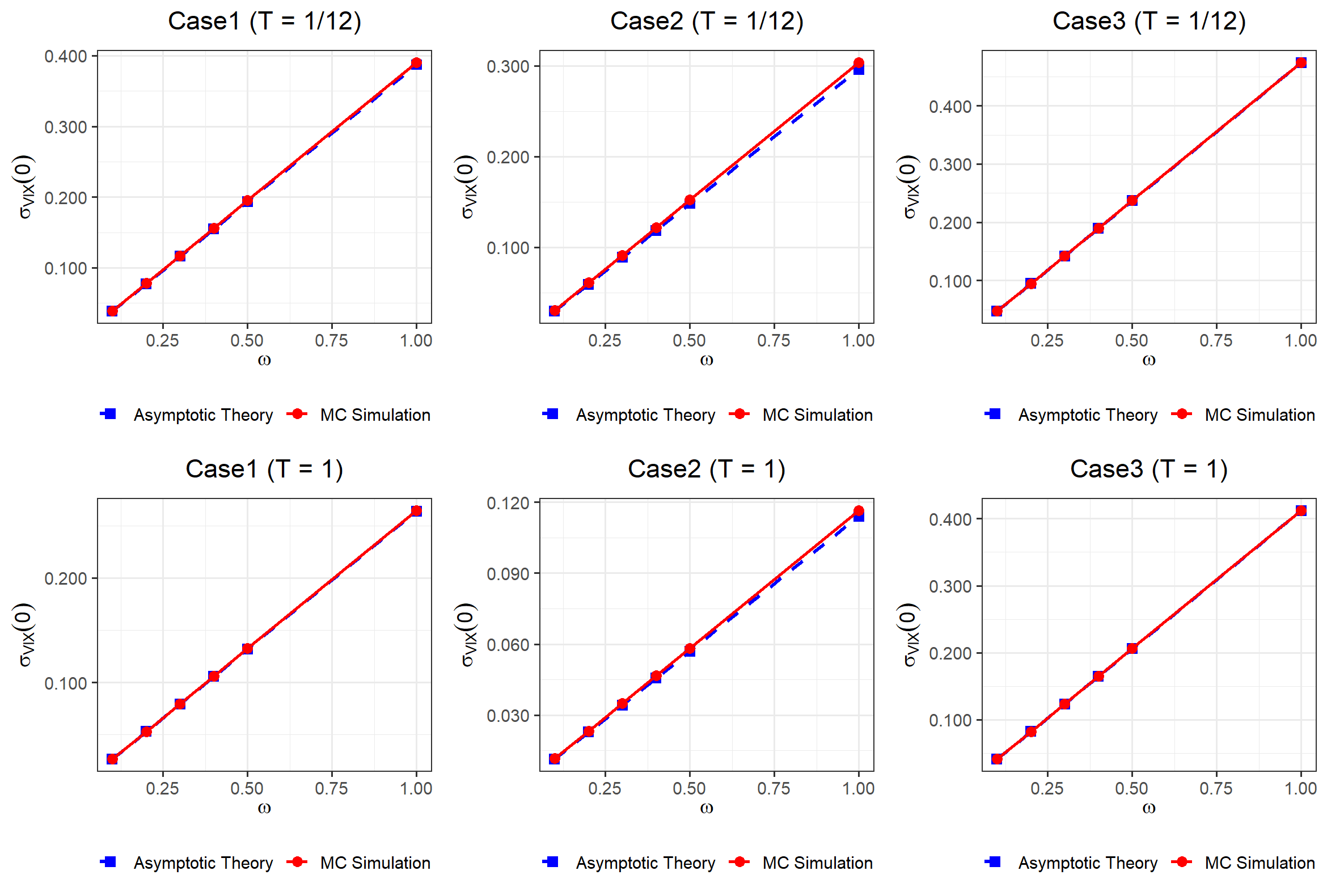}
\caption{Test for the ATM VIX implied volatility in the two-factor Bergomi model vs $\omega$. Correlation $\rho=0.7$.
The dashed blue line shows the asymptotic prediction from the small vol-of-vol limit $\sigma_{\rm VIX,ATM}^{(v)}(T)$ and the red dots show the results of an MC simulation. }
\label{Fig:1}
\end{figure}

We simulate the Ornstein-Uhlenbeck processes exactly using the recursion
\begin{align}
X^1_{n+1} &= e^{-k_1 \tau} X^1_n + \sqrt{v^1_\tau} z_1\,,\\
X^2_{n+1} &= e^{-k_2 \tau} X^2_n + \sqrt{v^2_\tau} z_2
\end{align}
with $v^i_\tau = \frac{1}{2k_i} \left(1 - e^{-2k_i\tau}\right)$
and $(z_1,z_2)$ are standard normals correlated with correlation $\rho$.

Consider the VIX index at maturity $T=n \tau$ with 
$n < N$. We approximate $\mathrm{VIX}_T$ by the application of the trapezoidal rule as
$$\mathrm{VIX}_T^2 = \frac{1}{\tau} \int_T^{T+\tau} \xi_T^u du
\simeq\frac12 \left(\xi_n^n + \xi_n^{n+1}\right)\,,$$
where for simplicity, we denote
$\xi_{i\tau}^{j\tau}$ with integers $i,j$ simply as $\xi_i^j$ and
\begin{align}
\xi_n^n= \xi_0^n \exp\left( \omega x_{n\tau}^{n\tau} - \frac12 \omega^2 v_{n\tau}(n\tau) \right)\,, \qquad
\xi_n^{n+1} = \xi_0^{n+1} \exp\left( \omega x_{n\tau}^{(n+1)\tau} - \frac12 \omega^2 v_{n\tau}((n+1)\tau) \right) \,.
\end{align}

The prices of VIX futures and options are
\begin{align}
F_V(T) = \mathbb{E}[\mathrm{VIX}_T]\,, \quad
C_V(K,T) = \mathbb{E}[(\mathrm{VIX}_T - K)^+]\,,
\quad
P_V(K,T) = \mathbb{E}[(K - \mathrm{VIX}_T )^+] \,.
\end{align}
They will be computed by MC simulation.

\begin{figure}[h!]
\centering
\includegraphics[width=3.0in]{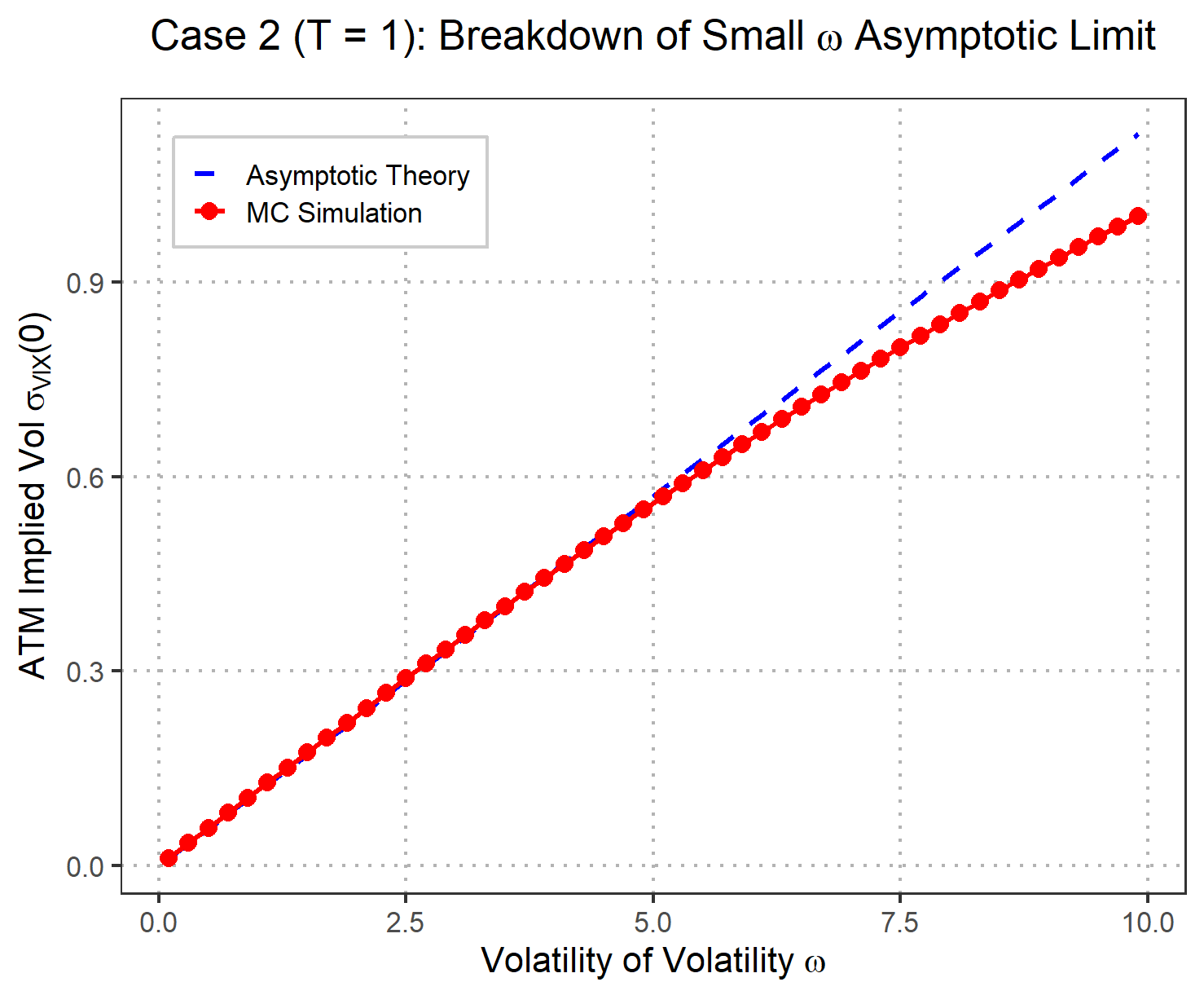}
\caption{The ATM VIX implied volatility becomes non-linear in $\omega$ for sufficiently large values of this parameter. The VIX option maturity $T=1$,
case 2, correlation $\rho=0.7$.}
\label{Fig:3}
\end{figure}

Figure~\ref{Fig:1} shows a test for the ATM VIX implied volatility $\sigma_{\mathrm{VIX,ATM}}(T)$ in the two-factor Bergomi model for VIX options with maturity $T=1/12$ (1 month) and $T=1$ (1 year), for the three cases of the mixing weights $\theta_i$. The correlation is fixed to $\rho=0.7$. 
The figure compares the result of a MC simulation with the asymptotic predictions in Table~\ref{tab:predictions}. The agreement is very good even for realistic values $\omega=O(1)$.

\begin{figure}[h!]
\centering
\includegraphics[width=2.0in]{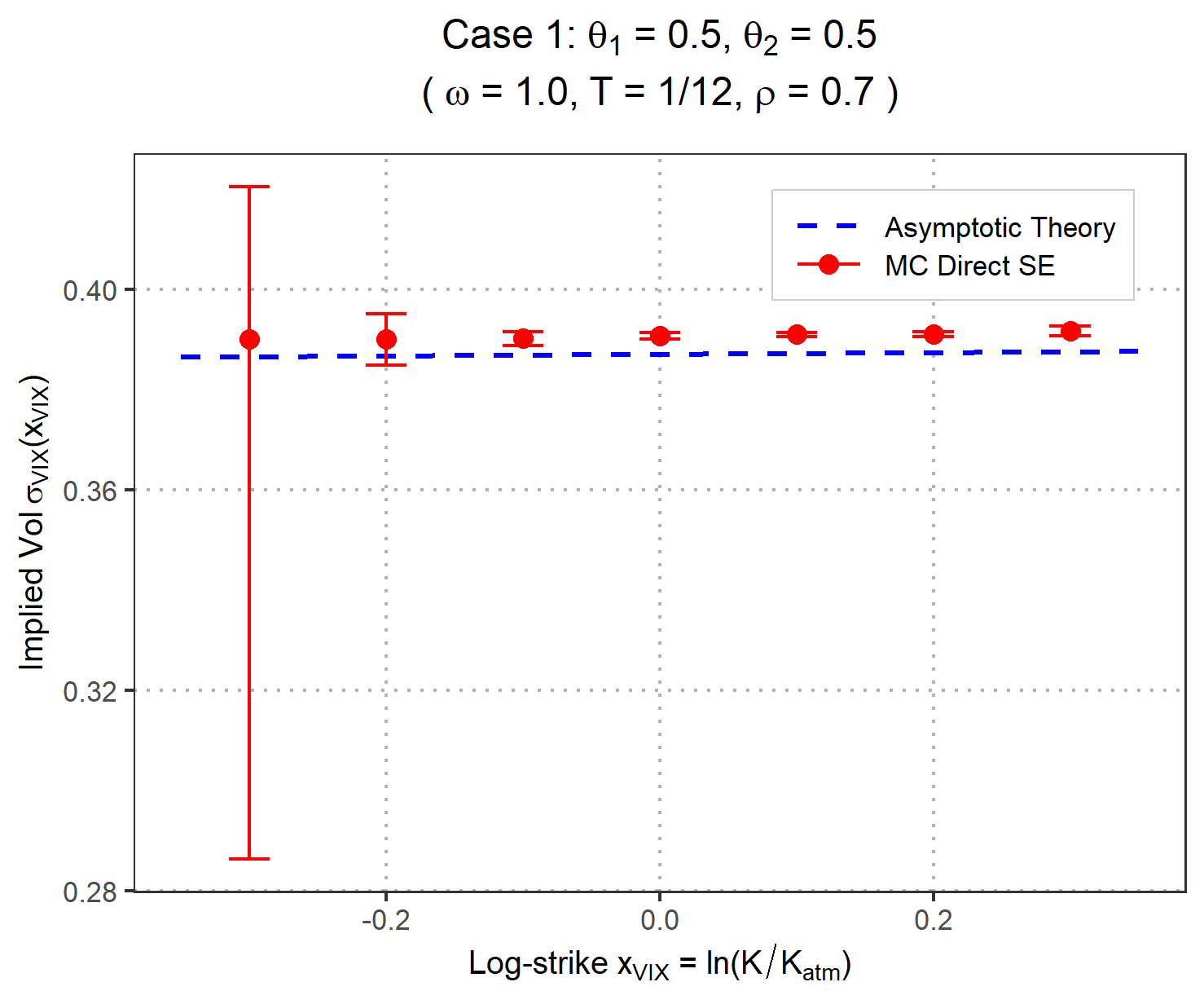}
\includegraphics[width=2.0in]{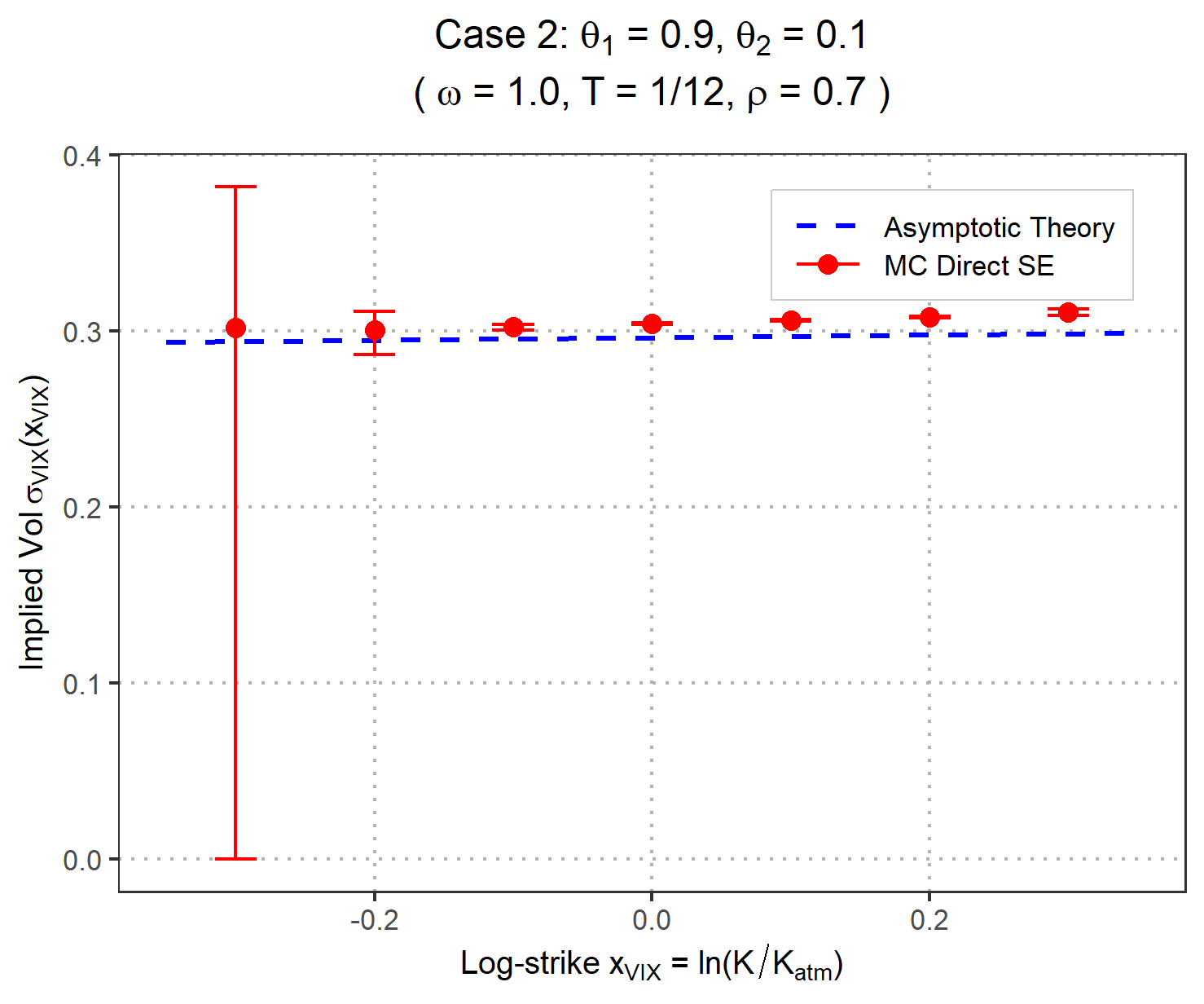}
\includegraphics[width=2.0in]{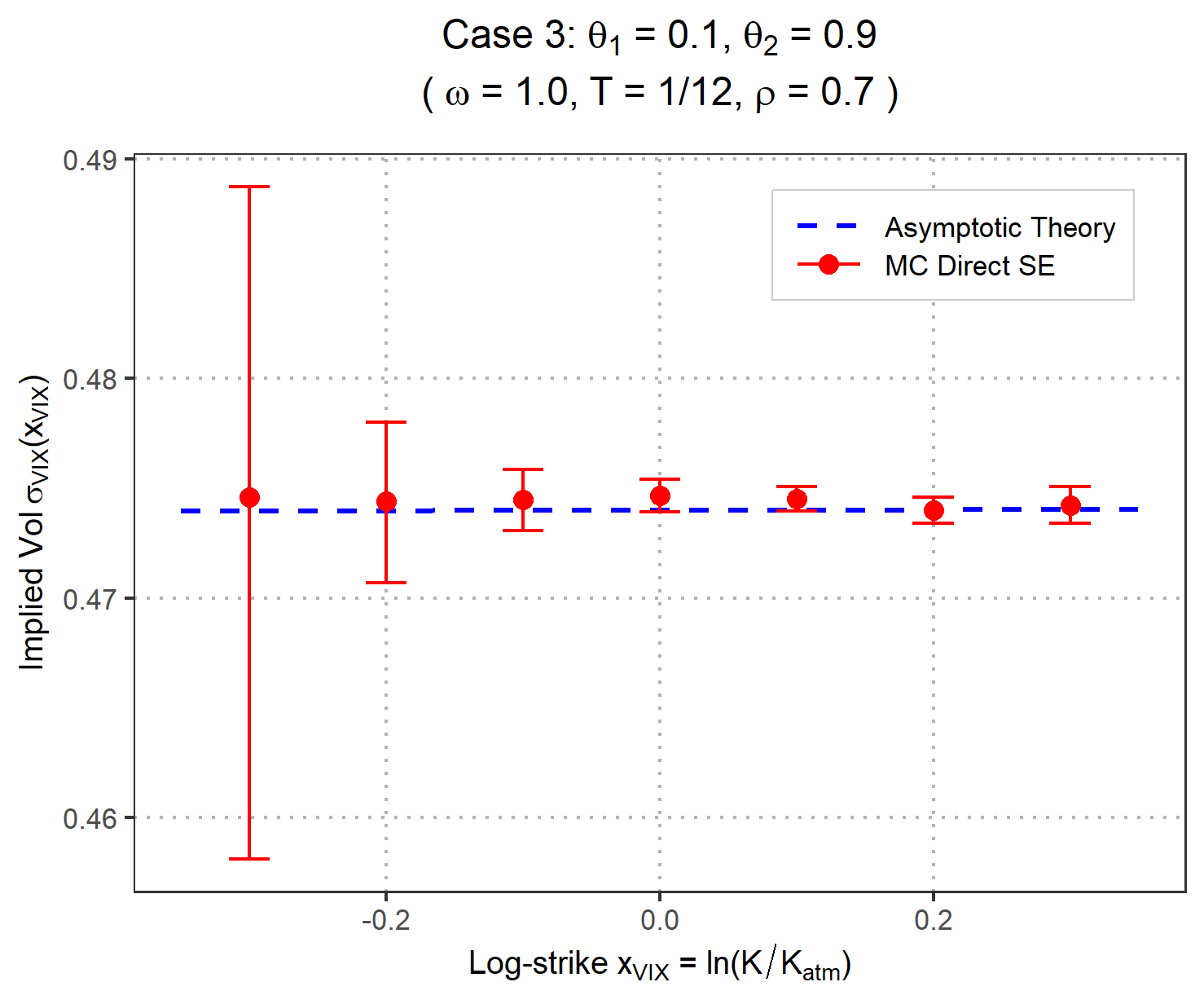}
\caption{Tests for the VIX implied volatility comparing the MC simulation (red dots with error bars) and the asymptotic prediction from the small vol-of-vol limit (dashed blue). The VIX options have maturity $T=1/12$.
Correlation $\rho=0.7$.}
\label{Fig:2}
\end{figure}

We note very good agreement of the linear dependence of the ATM VIX implied volatility on $\omega$ even for realistic values of the vol-of-vol parameter $\omega =O(1)$. For larger values of $\omega$ the dependence becomes non-linear, as shown in Figure~\ref{Fig:3} which shows the ATM VIX implied volatility for a VIX option with maturity $T=1$ and parameters as shown, over a wider range of $\omega$. The deviations from linear dependence appear at $\omega \sim 6$ which is larger than typical values appearing in practice. This suggests that the asymptotic predictions from the small vol-of-vol hold well for cases of practical interest.

Figure~\ref{Fig:2} shows a test for the VIX implied volatility smile $\sigma_{\mathrm{VIX}}(K,T)$ in the two-factor Bergomi model for VIX options with maturity $T=1/12$ (1 month). The asymptotic prediction from the linear approximation \eqref{VIXlinear} with coefficients from Table~\ref{tab:predictions} are shown as the dashed blue line, and the VIX implied volatility obtained by MC simulation is shown as the red dots with error bars. We note very good agreement within the MC errors. 

We conclude that the analytical results for the VIX implied volatility following from the small vol-of-vol limit are an useful approximation for pricing VIX options in the Bergomi model with realistic parameters relevant for practical applications. They can be expected to be useful for fast and efficient calibration and pricing VIX options in this model. 
 
\section*{Acknowledgements}
Lingjiong Zhu is partially supported by the grants NSF DMS-2053454, NSF DMS-2208303.

\bibliographystyle{alpha}
\bibliography{VIX}

@article{Bergomi,
	Author = {Bergomi, L.},
	Journal = {Risk},
	Title = {Smile dynamics},
	volume={9},
	pages={117-123},
	Year = {2004}}

@article{Bergomi-II,
	Author = {Bergomi, L.},
	Journal = {Risk},
	Title = {Smile dynamics {II}},
	volume={10},
	pages={67-73},
	Year = {2005}}

@article{Bergomi-III,
	Author = {Bergomi, L.},
	Journal = {Risk},
	Title = {Smile dynamics {III}},
	volume={10},
	pages={90-96},
	Year = {2008}}

@book{BergomiBook,
	Author = {Bergomi, L.},
	Title = {Stochastic Volatility Modeling},
	publisher={CRC Press},
	Year = {2015}}

@article{Bergomi-Guyon,
	Author = {Bergomi, L. and Guyon, J.},
	Journal = {SSRN:1967470},
	Title = {The smile in stochastic volatility models},
	Year = {2011}}

@article{Bourgey,
	Author = {Bourgey, F. and De Marco, S. and Gobet, E.},
	Journal = {Quantitative Finance},
	Title = {Weak approximations and {VIX} option price expansions in forward variance curve models},
    volume = {23},
    pages = {1259--1283},
	Year = {2023}}

@article{Liao,
	Author = {Liao, Y. and Agarwal, A. and Bourgey, F.},
	Journal = {arXiv:2604.25123[q-fin.CP]},
	Title = {Implied volatility expansions for {VIX} options in forward variance models},
	Year = {2026}}

@article{OuldAly,
	Author = {Ould Aly, S.},
	Journal = {Applied Mathematical Finance},
	Title = {Forward variance dynamics: Bergomi's model revisited},
    Volume = {21},
    Pages = {84--107},
	Year = {2011}}

@article{Guyon2021,
	Author = {Guyon, J.},
	Journal = {SSRN 3956786},
	Title = {The smile of stochastic volatility: Revisiting the {B}ergomi-{G}uyon expansion},
	Year = {2021}}

@article{Guyon,
	Author = {Guyon, J.},
	Journal = {SIAM Journal on Financial Mathematics},
	Title = {The {VIX} future in {B}ergomi models: Fast approximation formulas and joint calibration with {S\&P500} skew},
    Volume = {13},
    Pages = {1418--1485},
	Year = {2022}}

@article{Lacombe,
	Author = {Lacombe, C. and Muguruza, A. and Stone, H.},
	Journal = {Mathematics and Financial Economics},
	Title = {Asymptotics for volatility derivatives in multi-factor rough volatility models},
    Volume = {15},
    Pages = {545-577},
	Year = {2021}}

@article{Kyakutwika,
	Author = {Kyakutwika, N. and Alfeus, M. and Schlogl, E.},
	Journal = {arXiv:2506.23409[q-fin.PR]},
	Title = {Pricing and calibration of {VIX} derivatives in mixed {B}ergomi models via quantisation},
	Year = {2025}}

@article{VIXpaper,
	Author = {Pirjol, D. and Wang, X. and Zhu, L.},
	Journal = {arxiv:2407.16813},
	Title = {Short-maturity asymptotics for {VIX} and {E}uropean options in local-stochastic volatility models},
	Year = {2024}}

@article{VIX-jumps-paper,
	Author = {Guo, D. and Pirjol, D. and Wang, X. and Zhu, L.},
	Journal = {arXiv:2601.17248},
	Title = {{VIX} and {E}uropean options with jumps in the short-maturity regime},
	Year = {2026}}

@article{Kokholm2015,
    author = {Kokholm, T. and Stisen, M.},
    journal = {The Journal of Risk Finance},
    title = {Joint pricing of {VIX} and {SPX} options with stochastic volatility and jump models},
    year = {2015},
    volume = {16},
    number = {1},
    pages = {27–48}
}

@article{AbiJaber2022,
  title={{The quintic Ornstein-Uhlenbeck volatility model that
jointly calibrates SPX and VIX smiles}},
  author={Abi Jaber, E. and Illand, C. and Li, S.},
  journal={Risk},
  year={2023},
  volume={June}
}

@article{Baldeaux2014,
  title={Consistent modelling of {VIX} and equity derivatives using a
3/2 plus jumps model},
  author={Baldeaux, J. and Badran, A.},
  journal={Applied Mathematical Finance},
  volume={21},
  pages={299-312},
  year={2014}
}

@article{Cao2020,
title = {Valuation of {VIX} and target volatility options with affine {GARCH} models},
author = {Cao, H. and Badescu, A. and Cui, Z. and Jayaraman, S.K.},
journal = {Journal of Futures Markets},
year = {2020},
pages = {1880-1917},
volume={40}
}

@article{Carr2005,
  title={Pricing options on realized variance},
  author={Carr, P. and Geman, H. and Madan, D.B. and Yor, M.},
  journal={Finance and Stochastics},
  year={2005},
  pages={453-475},
  volume={9}
 }

@Unpublished{VIXwp,
  title={{Volatility Index Methodology: {CBOE} Volatility Index}},
  author={{CBOE Global Indices}},
  note={White Paper},
  year={2023}
  }

@Unpublished{VIX1Dwp,
  title={{Volatility Index Methodology: {CBOE} 1-Day Volatility Index}},
  author={{CBOE Global Indices}},
  note={White Paper},
  year={2023}
  }

@article{Cuchiero2023,
  title={Joint calibration to {SPX} and {VIX} options with signature-based models},
  author={Cuchiero, C. and Gazzani, G. and M\"{o}ller, J. and Svaluto-Ferro, S.},
  journal={Mathematical Finance},
  volume={35},
  number={1},
  pages={161-213},
  year={2025}
 }

@book{Dembo1998,
  title={Large Deviations Techniques and Applications},
  edition={2nd},
  author={Dembo, A. and Zeitouni, O.},
  year={1998},
  publisher={Springer Verlag}
}

@article{Detemple2000,
  title={The valuation of volatility options},
  author={Detemple, J. and Osakwe, C.},
  journal={European Finance Review},
  year={2000},
  pages={21-50},
  volume={4}
 }

@article{Goard2013,
  title={Stochastic volatility models and the pricing of {VIX} options},
  author={Goard, J. and Mazur, M.},
  journal={Mathematical Finance},
  year={2013},
  pages={439-458},
  volume={23}
  }

@article{Guyon2020,
  title={The joint {SP500/VIX} smile calibration puzzle solved},
  author={Guyon, J.},
  journal={Risk},
  year={2020},
  volume={June}
  }

@article{Horvath2020,
author = {Horvath, B. and Jacquier, A. and Tankov, P.},
title = {Volatility options in rough volatility models},
journal={SIAM Journal on Financial Mathematics},
volume={11}, 
pages={437-469},
year={2020}
}

@article{Jacquier2021,
author = {Jacquier, A. and Muguruza, A. and Pannier, A.},
title = {Rough multi-factor volatility for {SPX} and {VIX} options},
journal={Advances in Applied Probability},
year={2025},
volume={57},
number={2},
pages={524-565}
}

@article{PZ2025,
  title={{VIX} options in the {SABR} model},
  author={Pirjol, D. and Zhu, L. },
  journal={Operations Research Letters},
  year={2025},
  volume={63},
  pages={107347}
}

@article{Romer2022,
title={Empirical analysis of rough and classical stochastic volatility models to the {SPX} and {VIX} markets},
author={R\o{}mer, S.E.},
journal={Quantitative Finance},
volume={22},
pages={1805-1838},
year={2022}
}

@article{Sepp2008,
  title={{VIX} option pricing in a jump-diffusion model},
  author={Sepp, A.},
  journal={Risk},
  year={2008},
  pages={84-89},
  volume={May}
  }

@article{Sepp2008a,
  title={Pricing options on realized variance in {H}eston model with jumps in returns},
  author={Sepp, A.},
  journal={Journal of Computational Finance},
  year={2008},
  pages={33-70},
  volume={11}
  }

@book{VaradhanLD,
  title={Large Deviations and Applications},
  author={Varadhan, S.R.S.},
  year={1984},
  publisher={SIAM},
  address={Philadelphia}
  }

@article{Guyon2024,
  title={Dispersion-constrained martingale {S}chr\"{o}dinger problems and the exact join {S\&P500/VIX} smile calibration puzzle},
  author={Guyon, J.},
  journal={Finance and Stochastics},
  volume={28},
  pages={27-79},
  year={2024}
}

@article{Tong2021,
  title={Pricing {VIX} options with realized volatility},
  author={Tong, C. and Huang, Z.},
  journal={Journal of Futures Markets},
  year={2021},
  pages={118-1200},
  volume={41},
  number={8}
}

@article{Alos2022VIX,
  title={On smile properties of volatility derivatives: Understanding the {VIX} skew},
  author={Al\'os, E. and Garc\'{i}a-Lorite, D. and Gonzalez, A.M.},
  journal={SIAM Journal of Financial Mathematics},
  volume={13},
  number={13},
  pages={32-69},
  year={2022}
}

@article{Yuan2022,
  title={Time-Varying Skew in {VIX} Derivatives Pricing},
  author={Yuan, P.},
  journal={Management Science},
  volume={68},
  number={10},
  pages={7065-7791},
  year={2022}
}


\appendix

\section{Background on Large Deviations Theory}\label{sec:LDP}

We present here a few basic concepts of large deviations theory from probability theory
which are in the proofs; see e.g. Dembo and Zeitouni \cite{Dembo1998} and Varadhan \cite{VaradhanLD} for more details on large deviations and its applications.
First, we present the definition of the large deviation principle.

\begin{definition}[Large Deviation Principle, see e.g. Section 1.2. in \cite{Dembo1998}]\label{defn:LDP}
A sequence $(P_\epsilon)_{\epsilon \in \mathbb{R}^+}$ of probability measures
on a topological space $X$ satisfies the large deviation principle with rate function $I: X \to \mathbb{R}$
if $I$ is non-negative, lower semicontinuous and for any measurable set $A$, we have
\begin{equation}
- \inf_{x\in A^o} I(x) \leq \liminf_{\epsilon\to 0} \epsilon \log P_\epsilon(A) \leq
\limsup_{\epsilon\to 0} \epsilon \log P_\epsilon(A) \leq - \inf_{x\in \bar A} I(x) \,,
\end{equation}
where $A^o$ denotes the interior of $A$ and $\bar A$ its closure.
\end{definition}

Next, we present the contraction principle.

\begin{theorem}[Contraction Principle, see e.g. Theorem 4.2.1. in \cite{Dembo1998}]\label{Contraction:Thm}
If $F:X\rightarrow Y$ is a continuous map and 
$P_{\epsilon}$ satisfies a large deviation principle on $X$ with the rate 
function $I(x)$,
then the probability measures $Q_{\epsilon}:=P_{\epsilon}F^{-1}$ satisfies
a large deviation principle on $Y$ with the rate function
$J(y)=\inf_{x: F(x)=y}I(x)$.
\end{theorem}

Finally, we present the G\"{a}rtner-Ellis theorem on $\mathbb{R}^{d}$.

\begin{theorem}[G\"{a}rtner-Ellis Theorem, see e.g. Theorem 2.3.6. in \cite{Dembo1998}]\label{GE:Thm}
Let $X_{\epsilon}$ be a sequence of random vectors in $\mathbb{R}^{d}$.
For any $\psi\in\mathbb{R}^{d}$, assume that the limit $\Lambda(\psi):=\lim_{\epsilon\rightarrow 0}\epsilon\log\mathbb{E}\left[e^{\langle\psi,\epsilon^{-1}X_{\epsilon}\rangle}\right]$
exists as an extended real number and  the origin belongs to the interior of $\mathcal{D}_{\Lambda}:=\{\psi\in\mathbb{R}^{d}:\Lambda(\psi)<\infty\}$.
Further assume that $\Lambda$ is essentially smooth (that is, $\mathcal{D}^{\circ}_{\Lambda}$ is non-empty, $\Lambda(\cdot)$ is differentiable throughout $\mathcal{D}_{\Lambda}^{\circ}$
and $\Lambda(\cdot)$ is steep, i.e. $\lim_{n\rightarrow\infty}\Vert\nabla\Lambda(\psi_{n})\Vert=\infty$ whenever $\psi_{n}$ is a sequence in $\mathcal{D}_{\Lambda}^{\circ}$ converging to a boundary point of $\mathcal{D}_{\Lambda}^{\circ}$) and lower semicontinuous. Then $\mathbb{P}(X_{\epsilon}\in\cdot)$ satisfies a large deviation principle
with the rate function
\begin{equation}
\Lambda^{\ast}(x):=\sup_{\psi\in\mathbb{R}^{d}}\left\{\langle\psi,x\rangle-\Lambda(\psi)\right\}.
\end{equation}
\end{theorem}

\section{Technical Proofs}\label{sec:proofs}

\subsection{Proof of Theorem~\ref{thm:short:OTM:one:factor}}

\begin{proof}
Let us consider the case $\sqrt{\frac{1}{\tau}\int_{0}^{\tau}\xi_{0}^{u}du}<K$
and prove the result for OTM VIX call option. 

First, we can compute that
\begin{equation}\label{short:to:show:1}
\lim_{T\rightarrow 0}T\log C(T,\omega)=\lim_{T\rightarrow 0}T\log\mathbb{Q}\left(\mathrm{VIX}_{T}^{2}\geq K^{2}\right),
\end{equation}
provided that the limit on the right hand side exists and is continuous in $K$.

Let us prove \eqref{short:to:show:1}.
For any $p,q>1$ with $\frac{1}{p}+\frac{1}{q}=1$, by H\"{o}lder's inequality,
\begin{align}
\limsup_{T\rightarrow 0}T\log C(T,\omega)&=\limsup_{T\rightarrow 0}T\log\mathbb{E}\left[(\mathrm{VIX}_{T}-K)1_{\mathrm{VIX}_{T}\geq K}\right]
\nonumber
\\
&\leq
\limsup_{T\rightarrow 0}T\log\left(\mathbb{E}\left[|\mathrm{VIX}_{T}-K|^{p}\right]\right)^{1/p}
\left(\mathbb{E}\left[\left(1_{\mathrm{VIX}_{T}\geq K}\right)^{q}\right]\right)^{1/q}
\nonumber
\\
&\leq
\limsup_{T\rightarrow 0}T\log\left(\mathbb{E}\left[|\mathrm{VIX}_{T}-K|^{p}\right]\right)^{1/p}
\left(\mathbb{Q}\left(\mathrm{VIX}_{T}\geq K\right)\right)^{1/q}.
\end{align}
By Jensen's inequality, 
\begin{align}
\mathbb{E}\left[|\mathrm{VIX}_{T}-K|^{p}\right]
\leq\mathbb{E}\left[(\mathrm{VIX}_{T}+K)^{p}\right]
\leq 2^{p-1}\mathbb{E}\left[\mathrm{VIX}_{T}^{p}+K^{p}\right],
\end{align}
and for any $p>2$, by Jensen's inequality
\begin{align}
\mathbb{E}\left[\mathrm{VIX}_{T}^{p}\right]
&=\mathbb{E}\left[\left(\frac{1}{\tau}\int_{T}^{T+\tau}\xi_{0}^{u}e^{\omega e^{-k(u-T)}X_{T}-\frac{\omega^{2}}{2}e^{-2k(u-T)}v_{T}}du\right)^{\frac{p}{2}}\right]
\nonumber
\\
&\leq\mathbb{E}\left[\frac{1}{\tau}\int_{T}^{T+\tau}\left(\xi_{0}^{u}\right)^{\frac{p}{2}}e^{\frac{p\omega}{2}e^{-k(u-T)}X_{T}-\frac{p\omega^{4}}{2}e^{-2k(u-T)}v_{T}}du\right]
\nonumber
\\
&\leq\mathbb{E}\left[\frac{1}{\tau}\int_{T}^{T+\tau}\left(\xi_{0}^{u}\right)^{\frac{p}{2}}e^{\frac{p\omega}{2}X_{T}}du\right]
\nonumber
\\
&=\mathbb{E}\left[e^{\frac{p\omega}{2}X_{T}}\right]\frac{1}{\tau}\int_{T}^{T+\tau}\left(\xi_{0}^{u}\right)^{\frac{p}{2}}du
=e^{\frac{1}{2}\frac{p^{2}\omega^{2}}{4}\frac{1-e^{-2kT}}{2k}}\frac{1}{\tau}\int_{T}^{T+\tau}\left(\xi_{0}^{u}\right)^{\frac{p}{2}}du,
\end{align}
where we used the fact that $X_{T}$  is Gaussian with mean $0$ and variance $\frac{1-e^{-2kT}}{2k}$.
Therefore, we conclude that for any $p>2$ and $1<q<2$,
\begin{align}
\limsup_{T\rightarrow 0}T\log C(T,\omega)
\leq
\frac{1}{q}\limsup_{T\rightarrow 0}T\log\mathbb{Q}\left(\mathrm{VIX}_{T}\geq K\right).
\end{align}
Since it holds for any $2>q>1$, we have
\begin{align}
\limsup_{T\rightarrow 0}T\log C(T,\omega)
\leq
\limsup_{T\rightarrow 0}T\log\mathbb{Q}\left(\mathrm{VIX}_{T}\geq K\right).
\end{align}

On the other hand, for any $\epsilon>0$, 
\begin{align}
\liminf_{T\rightarrow 0}T\log C(T,\omega)&=\liminf_{T\rightarrow 0}T\log\mathbb{E}\left[(\mathrm{VIX}_{T}-K)1_{\mathrm{VIX}_{T}\geq K}\right]
\nonumber
\\
&\geq\liminf_{T\rightarrow 0}T\log\mathbb{E}\left[(\mathrm{VIX}_{T}-K)1_{\mathrm{VIX}_{T}\geq K+\epsilon}\right]
\nonumber
\\
&\geq
\liminf_{T\rightarrow 0}T\log\epsilon\mathbb{Q}\left(\mathrm{VIX}_{T}\geq K+\epsilon\right)
\nonumber
\\
&=\liminf_{T\rightarrow 0}T\log\mathbb{Q}\left(\mathrm{VIX}_{T}^{2}\geq (K+\epsilon)^{2}\right).
\end{align}
Since it holds for any $\epsilon>0$, we conclude that
\begin{equation}
\lim_{T\rightarrow 0}T\log C(T,\omega)=\lim_{T\rightarrow 0}T\log\mathbb{Q}\left(\mathrm{VIX}_{T}^{2}\geq K^{2}\right),
\end{equation}
provided that the limit on the right hand sides exists and is continuous in $K$.
This proves \eqref{short:to:show:1}.

Next, we can show that
\begin{equation}\label{short:to:show:2}
\lim_{T\rightarrow 0}T\log\mathbb{Q}\left(\mathrm{VIX}_{T}^{2}\geq K^{2}\right)
=\lim_{T\rightarrow 0}T\log\mathbb{Q}\left(\frac{1}{\tau}\int_{0}^{\tau}\xi_{0}^{u}e^{\omega e^{-ku}X_{T}}du\geq K^{2}\right),
\end{equation}
provided that the limit on the right hand side exists and is continuous in $K$.

Let us prove \eqref{short:to:show:2}. 
Note that
\begin{align}
\lim_{T\rightarrow 0}T\log\mathbb{Q}\left(\mathrm{VIX}_{T}^{2}\geq K^{2}\right)
&=\lim_{T\rightarrow 0}T\log\mathbb{Q}\left(\frac{1}{\tau}\int_{T}^{T+\tau}\xi_{0}^{u}e^{\omega e^{-k(u-T)}X_{T}-\frac{\omega^{2}}{2}e^{-2k(u-T)}v_{T}}du\geq K^{2}\right)
\nonumber
\\
&=\lim_{T\rightarrow 0}T\log\mathbb{Q}\left(\frac{1}{\tau}\int_{0}^{\tau}\xi_{0}^{u+T}e^{\omega e^{-ku}X_{T}-\frac{\omega^{2}}{2}e^{-2ku}v_{T}}du\geq K^{2}\right).
\end{align}
Since $\xi_{0}^{u}>0$ and it is continuous in $u$ and $v_{T}=\frac{1-e^{-2kT}}{2k}\rightarrow 0$
as $T\rightarrow 0$, for any $\epsilon>0$, 
uniformly in $0\leq u\leq T$, we have
\begin{equation}
\frac{1}{1+\epsilon}\xi_{0}^{u}e^{\omega e^{-ku}X_{T}}
\leq
\xi_{0}^{u+T}e^{\omega e^{-ku}X_{T}-\frac{\omega^{2}}{2}e^{-2ku}v_{T}}
\leq
(1+\epsilon)\xi_{0}^{u}e^{\omega e^{-ku}X_{T}},
\end{equation}
for any sufficiently small $T>0$. 
Therefore, 
\begin{align}
&\liminf_{T\rightarrow 0}T\log\mathbb{Q}\left(\frac{1}{\tau}\int_{0}^{\tau}\xi_{0}^{u}e^{\omega e^{-ku}X_{T}}du\geq K^{2}(1+\epsilon)\right)
\nonumber
\\
&\leq
\liminf_{T\rightarrow 0}T\log\mathbb{Q}\left(\mathrm{VIX}_{T}^{2}\geq K^{2}\right)
\leq
\limsup_{T\rightarrow 0}T\log\mathbb{Q}\left(\mathrm{VIX}_{T}^{2}\geq K^{2}\right)
\nonumber
\\
&\leq
\limsup_{T\rightarrow 0}T\log\mathbb{Q}\left(\frac{1}{\tau}\int_{0}^{\tau}\xi_{0}^{u}e^{\omega e^{-ku}X_{T}}du\geq\frac{K^{2}}{1+\epsilon}\right).
\end{align}
Since it holds for any $\epsilon>0$, we conclude that
\begin{equation}
\lim_{T\rightarrow 0}T\log\mathbb{Q}\left(\mathrm{VIX}_{T}^{2}\geq K^{2}\right)
=\lim_{T\rightarrow 0}T\log\mathbb{Q}\left(\frac{1}{\tau}\int_{0}^{\tau}\xi_{0}^{u}e^{\omega e^{-ku}X_{T}}du\geq K^{2}\right),
\end{equation}
provided that the limit on the right hand side exists and is continuous in $K$.
This proves \eqref{short:to:show:2}.

Since $X_{T}$ is Gaussian with mean $0$ and variance $\frac{1-e^{-2kT}}{2k}$, 
one can compute that for any $\theta\in\mathbb{R}$, 
\begin{equation}
\lim_{T\rightarrow 0}T\log\mathbb{E}\left[e^{\frac{\theta}{T}X_{T}}\right]=
\lim_{T\rightarrow 0}T\log\exp\left(\frac{\theta^{2}}{2T^{2}}\frac{1-e^{-2kT}}{2k}\right)
=\frac{\theta^{2}}{2},
\end{equation}
by G\"{a}rtner-Ellis theorem (see Theorem~\ref{GE:Thm}), $\mathbb{Q}(X_{T}\in\cdot)$ satisfies a large deviation principle (see Definition~\ref{defn:LDP})
with the rate function 
\begin{equation}
\sup_{\theta\in\mathbb{R}}\left\{\theta x-\frac{\theta^{2}}{2}\right\}=\frac{1}{2}x^{2}.
\end{equation}
Therefore, by contraction principle (see Theorem~\ref{Contraction:Thm}), we obtain
\begin{equation}
\lim_{T\rightarrow 0}T\log\mathbb{Q}\left(\frac{1}{\tau}\int_{0}^{\tau}\xi_{0}^{u}e^{\omega e^{-ku}X_{T}}du\geq K^{2}\right)
=-\inf_{x:\frac{1}{\tau}\int_{0}^{\tau}\xi_{0}^{u}e^{\omega e^{-ku}x}du=K^{2}}\frac{1}{2}x^{2}.
\end{equation}
Hence, we conclude that
\begin{equation}
\lim_{T\rightarrow 0}T\log C(T,\omega)=\frac{1}{2}x_{+}^{2},
\end{equation}
where $x_{+}>0$ is the unique positive value such that
\begin{equation}
\frac{1}{\tau}\int_{0}^{\tau}\xi_{0}^{u}e^{\omega e^{-ku}x_{+}}du=K^{2}.
\end{equation}

The argument for OTM VIX put option
is similar and is hence omitted here. 
The proof is complete.
\end{proof}


\subsection{Proof of Theorem~\ref{thm:short:ATM:one:factor}}

\begin{proof}
First of all, we have
\begin{align}
\mathrm{VIX}_{T}&=\left(\frac{1}{\tau}\int_{T}^{T+\tau}\xi_{T}^{u}du\right)^{1/2}
\nonumber
\\
&=\left(\frac{1}{\tau}\int_{T}^{T+\tau}\xi_{0}^{u}e^{\omega e^{-k(u-T)}X_{T}-\frac{\omega^{2}}{2}e^{-2k(u-T)}v_{T}}du\right)^{1/2}
\nonumber
\\
&=\left(\frac{1}{\tau}\int_{0}^{\tau}\xi_{0}^{u}du\right)^{1/2}
+\frac{\frac{1}{\tau}\int_{0}^{\tau}\xi_{0}^{u}\omega e^{-ku}X_{T}du}{2\left(\frac{1}{\tau}\int_{0}^{\tau}\xi_{0}^{u}du\right)^{1/2}}+\mathcal{R}(T),
\end{align}
where 
\begin{align}
\mathcal{R}(T):= \left(\frac{1}{\tau}\int_{T}^{T+\tau}\xi_{T}^{u}du\right)^{1/2}-\left(\frac{1}{\tau}\int_{0}^{\tau}\xi_{0}^{u}du\right)^{1/2}
-\frac{\frac{1}{\tau}\int_{0}^{\tau}\xi_{0}^{u}\omega e^{-ku}X_{T}du}{2\left(\frac{1}{\tau}\int_{0}^{\tau}\xi_{0}^{u}du\right)^{1/2}}
\end{align}
is a remainder term that can be shown to be negligible.
Indeed, since the map $x\mapsto x^{+}$ is $1$-Lipschitz and $\left(\frac{1}{\tau}\int_{0}^{\tau}\xi_{0}^{u}du\right)^{1/2}=K$, 
\begin{align}
\left|\mathbb{E}\left[\left(\mathrm{VIX}_{T}-K\right)^{+}\right]-\mathbb{E}\left[\left(\frac{\frac{1}{\tau}\int_{0}^{\tau}\xi_{0}^{u}\omega e^{-ku}X_{T}du}{2\left(\frac{1}{\tau}\int_{0}^{\tau}\xi_{0}^{u}du\right)^{1/2}}\right)^{+}\right]
\right|
\leq\mathbb{E}|\mathcal{R}(T)|.
\end{align}
We can compute that
\begin{align}
\mathcal{R}(T)
&=\left(\frac{1}{\tau}\int_{T}^{T+\tau}\xi_{0}^{u}e^{\omega e^{-k(u-T)}X_{T}-\frac{\omega^{2}}{2}e^{-2k(u-T)}v_{T}}du\right)^{1/2}
-\left(\frac{1}{\tau}\int_{0}^{\tau}\xi_{0}^{u}du\right)^{1/2}
-\frac{\frac{1}{\tau}\int_{0}^{\tau}\xi_{0}^{u}\omega e^{-ku}X_{T}du}{2\left(\frac{1}{\tau}\int_{0}^{\tau}\xi_{0}^{u}du\right)^{1/2}}
\nonumber
\\
&=\left(\frac{1}{\tau}\int_{0}^{\tau}\xi_{0}^{T+u}e^{\omega e^{-ku}X_{T}-\frac{\omega^{2}}{2}e^{-2ku}v_{T}}du\right)^{1/2}
-\left(\frac{1}{\tau}\int_{0}^{\tau}\xi_{0}^{u}du\right)^{1/2}
-\frac{\frac{1}{\tau}\int_{0}^{\tau}\xi_{0}^{u}\omega e^{-ku}X_{T}du}{2\left(\frac{1}{\tau}\int_{0}^{\tau}\xi_{0}^{u}du\right)^{1/2}}.\label{R:T}
\end{align}
First, note that
\begin{align}
\mathbb{E}|\mathcal{R}(T)|
\leq
\mathbb{E}\left[|\mathcal{R}(T)|1_{|X_{T}|\leq 1}\right]
+\mathbb{E}\left[|\mathcal{R}(T)|1_{|X_{T}|>1}\right].
\end{align}
By Cauchy-Schwarz inequality, 
\begin{align}
\mathbb{E}\left[|\mathcal{R}(T)|1_{|X_{T}|>1}\right]
\leq\left(\mathbb{E}|\mathcal{R}(T)|^{2}\right)^{1/2}\left(\mathbb{Q}(|X_{T}|>1)\right)^{1/2}.
\end{align}
By the large deviations estimate in the proof of Theorem~\ref{thm:short:OTM:one:factor}, $\mathbb{Q}(|X_{T}|>1)=e^{-O(1/T)}$ as $T\rightarrow 0$.
Moreover, it follows from \eqref{R:T} and the inequality $(a+b+c)^{2}\leq 3a^{2}+3b^{2}+3c^{2}$ for any $a,b,c\in\mathbb{R}$ that
\begin{align}
\mathbb{E}|\mathcal{R}(T)|^{2}
&\leq
3\mathbb{E}\left[\frac{1}{\tau}\int_{0}^{\tau}\xi_{0}^{T+u}e^{\omega e^{-ku}X_{T}-\frac{\omega^{2}}{2}e^{-2ku}v_{T}}du\right]
+\frac{3}{\tau}\int_{0}^{\tau}\xi_{0}^{u}du
+3\mathbb{E}\left[\left(\frac{\frac{1}{\tau}\int_{0}^{\tau}\xi_{0}^{u}\omega e^{-ku}X_{T}du}{2\left(\frac{1}{\tau}\int_{0}^{\tau}\xi_{0}^{u}du\right)^{1/2}}\right)^{2}\right]
\nonumber
\\
&=3\frac{1}{\tau}\int_{0}^{\tau}\xi_{0}^{T+u}du+\frac{3}{\tau}\int_{0}^{\tau}\xi_{0}^{u}du
+3\left(\frac{\frac{1}{\tau}\int_{0}^{\tau}\xi_{0}^{u}\omega e^{-ku}du}{2\left(\frac{1}{\tau}\int_{0}^{\tau}\xi_{0}^{u}du\right)^{1/2}}\right)^{2}v_{T},
\end{align}
where we used the fact that $X_{T}$ is Gaussian with mean $0$ and variance $v_{T}$.
Thus, $\mathbb{E}|\mathcal{R}(T)|^{2}\leq O(1)$ as $T\rightarrow 0$. 
Hence, we conclude that
\begin{align}
\mathbb{E}\left[|\mathcal{R}(T)|1_{|X_{T}|>1}\right]=e^{-O(1/T)},
\end{align}
as $T\rightarrow 0$, which is negligible. 

Next, we focus on the term $\mathbb{E}\left[|\mathcal{R}(T)|1_{|X_{T}|\leq 1}\right]$.
First, we have
\begin{align}
&\mathbb{E}\left[\left|\left(\frac{1}{\tau}\int_{0}^{\tau}\xi_{0}^{T+u}e^{\omega e^{-ku}X_{T}-\frac{\omega^{2}}{2}e^{-2ku}v_{T}}du\right)^{1/2}
-\left(\frac{1}{\tau}\int_{0}^{\tau}\xi_{0}^{u}e^{\omega e^{-ku}X_{T}}du\right)^{1/2}\right|1_{|X_{T}|\leq 1}\right]
\nonumber
\\
&=\mathbb{E}\left[\frac{\left|\frac{1}{\tau}\int_{0}^{\tau}\xi_{0}^{T+u}e^{\omega e^{-ku}X_{T}-\frac{\omega^{2}}{2}e^{-2ku}v_{T}}du
-\frac{1}{\tau}\int_{0}^{\tau}\xi_{0}^{u}e^{\omega e^{-ku}X_{T}}du\right|1_{|X_{T}|\leq 1}}
{\left(\frac{1}{\tau}\int_{0}^{\tau}\xi_{0}^{T+u}e^{\omega e^{-ku}X_{T}-\frac{\omega^{2}}{2}e^{-2ku}v_{T}}du\right)^{1/2}
+\left(\frac{1}{\tau}\int_{0}^{\tau}\xi_{0}^{u}e^{\omega e^{-ku}X_{T}}du\right)^{1/2}}\right]
\nonumber
\\
&\leq\mathbb{E}\left[\frac{\left|\frac{1}{\tau}\int_{0}^{\tau}\xi_{0}^{T+u}e^{\omega e^{-ku}X_{T}-\frac{\omega^{2}}{2}e^{-2ku}v_{T}}du
-\frac{1}{\tau}\int_{0}^{\tau}\xi_{0}^{u}e^{\omega e^{-ku}X_{T}}du\right|1_{|X_{T}|\leq 1}}
{\left(\frac{1}{\tau}\int_{0}^{\tau}\xi_{0}^{u}e^{\omega e^{-ku}X_{T}}du\right)^{1/2}}\right]
\nonumber
\\
&\leq\frac{\mathbb{E}\left[\left|\frac{1}{\tau}\int_{0}^{\tau}\xi_{0}^{T+u}e^{\omega e^{-ku}X_{T}-\frac{\omega^{2}}{2}e^{-2ku}v_{T}}du
-\frac{1}{\tau}\int_{0}^{\tau}\xi_{0}^{u}e^{\omega e^{-ku}X_{T}}du\right|\right]}
{\left(\frac{1}{\tau}\int_{0}^{\tau}\xi_{0}^{u}e^{-\omega e^{-ku}}du\right)^{1/2}}.
\end{align}
Since there exists some $C>0$ such that
$\sup_{0\leq u\leq\tau}|\xi_{0}^{T+u}-\xi_{0}^{u}|\leq CT$ for any sufficiently small $T$,
we can further compute that
\begin{align}
&\mathbb{E}\left[\left|\frac{1}{\tau}\int_{0}^{\tau}\xi_{0}^{T+u}e^{\omega e^{-ku}X_{T}-\frac{\omega^{2}}{2}e^{-2ku}v_{T}}du
-\frac{1}{\tau}\int_{0}^{\tau}\xi_{0}^{u}e^{\omega e^{-ku}X_{T}}du\right|\right]
\nonumber
\\
&\leq
\frac{1}{\tau}\int_{0}^{\tau}\left|\xi_{0}^{T+u}-\xi_{0}^{u}\right|\mathbb{E}\left[e^{\omega e^{-ku}X_{T}-\frac{\omega^{2}}{2}e^{-2ku}v_{T}}\right]du
+\frac{1}{\tau}\int_{0}^{\tau}\xi_{0}^{u}\mathbb{E}\left[e^{\omega e^{-ku}X_{T}}-e^{\omega e^{-ku}X_{T}-\frac{\omega^{2}}{2}e^{-2ku}v_{T}}\right]du
\nonumber
\\
&=\frac{1}{\tau}\int_{0}^{\tau}\left|\xi_{0}^{T+u}-\xi_{0}^{u}\right|du
+\frac{1}{\tau}\int_{0}^{\tau}\xi_{0}^{u}\left(e^{\frac{\omega^{2}}{2}e^{-2ku}v_{T}}-1\right)du
=O(T),
\end{align}
as $T\rightarrow 0$.

Next, we upper bound
\begin{align}
\mathbb{E}\left[\left|\left(\frac{1}{\tau}\int_{0}^{\tau}\xi_{0}^{u}e^{\omega e^{-ku}X_{T}}du\right)^{1/2}
-\left(\frac{1}{\tau}\int_{0}^{\tau}\xi_{0}^{u}du\right)^{1/2}
-\frac{\frac{1}{\tau}\int_{0}^{\tau}\xi_{0}^{u}\omega e^{-ku}X_{T}du}{2\left(\frac{1}{\tau}\int_{0}^{\tau}\xi_{0}^{u}du\right)^{1/2}}\right|1_{|X_{T}|\leq 1}\right].
\end{align}
Let us define the function:
\begin{equation}
F(x):=\left(\frac{1}{\tau}\int_{0}^{\tau}\xi_{0}^{u}e^{\omega e^{-ku}x}du\right)^{1/2},\qquad x\in\mathbb{R}.
\end{equation}
We can compute that
\begin{align}
F(0)=\left(\frac{1}{\tau}\int_{0}^{\tau}\xi_{0}^{u}du\right)^{1/2},
\qquad
F'(0)=\frac{\frac{1}{\tau}\int_{0}^{\tau}\xi_{0}^{u}\omega e^{-ku}du}{2\left(\frac{1}{\tau}\int_{0}^{\tau}\xi_{0}^{u}du\right)^{1/2}}.
\end{align}
By Taylor's theorem, for any $|x|\leq 1$, 
\begin{equation}
\left|F(x)-F(0)-F'(0)x\right|
\leq\frac{1}{2}\max_{|y|\leq 1}|F''(y)|x^{2}.
\end{equation}
Therefore, we have
\begin{align}
&\mathbb{E}\left[\left|\left(\frac{1}{\tau}\int_{0}^{\tau}\xi_{0}^{u}e^{\omega e^{-ku}X_{T}}du\right)^{1/2}
-\left(\frac{1}{\tau}\int_{0}^{\tau}\xi_{0}^{u}du\right)^{1/2}
-\frac{\frac{1}{\tau}\int_{0}^{\tau}\xi_{0}^{u}\omega e^{-ku}X_{T}du}{2\left(\frac{1}{\tau}\int_{0}^{\tau}\xi_{0}^{u}du\right)^{1/2}}\right|1_{|X_{T}|\leq 1}\right]
\nonumber
\\
&\leq
\mathbb{E}\left[\frac{1}{2}\max_{|y|\leq 1}|F''(y)|X_{T}^{2}1_{|X_{T}|\leq 1}\right]
\\
&\leq
\frac{1}{2}\max_{|y|\leq 1}|F''(y)|\mathbb{E}[X_{T}^{2}]
=\frac{1}{2}\max_{|y|\leq 1}|F''(y)|v_{T}=O(T),
\end{align}
as $T\rightarrow 0$.
Hence, we conclude that
\begin{equation}
\mathbb{E}|\mathcal{R}(T)|
\leq
\mathbb{E}\left[|\mathcal{R}(T)|1_{|X_{T}|\leq 1}\right]
+\mathbb{E}\left[|\mathcal{R}(T)|1_{|X_{T}|>1}\right]
\leq O(T),
\end{equation}
as $T\rightarrow 0$.
Therefore,
\begin{equation}
\lim_{T\rightarrow 0}\frac{C(T,\omega)}{\sqrt{T}}
=\lim_{T\rightarrow 0}\frac{1}{\sqrt{T}}\mathbb{E}\left[\left(\frac{\frac{1}{\tau}\int_{0}^{\tau}\xi_{0}^{u}\omega e^{-ku}X_{T}du}{2\left(\frac{1}{\tau}\int_{0}^{\tau}\xi_{0}^{u}du\right)^{1/2}}\right)^{+}\right]
=\frac{\frac{1}{\tau}\int_{0}^{\tau}\xi_{0}^{u}\omega e^{-ku}du}{2\left(\frac{1}{\tau}\int_{0}^{\tau}\xi_{0}^{u}du\right)^{1/2}}\lim_{T\rightarrow 0}\frac{\mathbb{E}\left[\left(X_{T}\right)^{+}\right]}{\sqrt{T}}.
\end{equation}
Since $X_{T}$ is Gaussian with mean $0$ and variance $\frac{1-e^{-2kT}}{2k}$, 
\begin{equation}
\mathbb{E}\left[\left(X_{T}\right)^{+}\right]
=\left(\frac{1-e^{-2kT}}{2k}\right)^{1/2}\mathbb{E}[Z^{+}],
\end{equation}
where $Z\sim\mathcal{N}(0,1)$ such that $\mathbb{E}[Z^{+}]=\frac{1}{\sqrt{2\pi}}\int_{0}^{\infty}ze^{-\frac{z^{2}}{2}}dz=\frac{1}{\sqrt{2\pi}}$.
Hence, we conclude that
\begin{equation}
\lim_{T\rightarrow 0}\frac{C(T,\omega)}{\sqrt{T}}
=\frac{\frac{1}{\tau}\int_{0}^{\tau}\xi_{0}^{u}\omega e^{-ku}du}{2\left(\frac{1}{\tau}\int_{0}^{\tau}\xi_{0}^{u}du\right)^{1/2}}
\lim_{T\rightarrow 0}\frac{1}{\sqrt{T}}\left(\frac{1-e^{-2kT}}{2k}\right)^{1/2}\frac{1}{\sqrt{2\pi}}
=\frac{\frac{1}{\tau}\int_{0}^{\tau}\xi_{0}^{u}\omega e^{-ku}du}{2\sqrt{2\pi}\left(\frac{1}{\tau}\int_{0}^{\tau}\xi_{0}^{u}du\right)^{1/2}}.
\end{equation}
The proof for the VIX put option is similar and hence omitted here.
This completes the proof.
\end{proof}


\subsection{Proof of Theorem~\ref{thm:short:OTM:two:factor}}

\begin{proof}
Let us consider the case $\sqrt{\frac{1}{\tau}\int_{0}^{\tau}\xi_{0}^{u}du}<K$
and prove the result for OTM VIX call option. 

Similarly as in the proof of \eqref{short:to:show:1} in the proof of Theorem~\ref{thm:short:OTM:one:factor}, 
one can show that
 \begin{equation}
\lim_{T\rightarrow 0}T\log C(T,\omega)=\lim_{T\rightarrow 0}T\log\mathbb{Q}\left(\mathrm{VIX}_{T}^{2}\geq K^{2}\right),
\end{equation}
and similarly as in the proof of \eqref{short:to:show:2} in the proof of Theorem~\ref{thm:short:OTM:one:factor}, 
one can show that
\begin{equation}
\lim_{T\rightarrow 0}T\log\mathbb{Q}\left(\mathrm{VIX}_{T}^{2}\geq K^{2}\right)
=\lim_{T\rightarrow 0}T\log\mathbb{Q}\left(\frac{1}{\tau}\int_{0}^{\tau}\xi_{0}^{u}e^{\omega\alpha_{\theta}(\theta_{1}e^{-k_{1}u}X_{T}^{1}+\theta_{2}e^{-k_{2}u}X_{T}^{2})}du\geq K^{2}\right).
\end{equation}

Since $X_{T}^{1},X_{T}^{2}$ are Gaussian distributed with 
\begin{align}
\left(X_{T}^{1},X_{T}^{2}\right)\sim\mathcal{N}\left(0,
\left(\begin{array}{cc}
v_{T}^{1} & v_{T}^{1,2}
\\
v_{T}^{1,2} & v_{T}^{2}
\end{array}\right)\right),
\end{align}
where
\begin{align}
v_{T}^{i}:=\frac{1-e^{-2k_{i}T}}{2k_{i}},\quad i=1,2,
\qquad
v_{T}^{1,2}:=\rho\frac{1-e^{-(k_{1}+k_{2})T}}{k_{1}+k_{2}},
\end{align}
one can compute that for any $\psi_{1},\psi_{2}\in\mathbb{R}$, 
\begin{align}
&\lim_{T\rightarrow 0}T\log\mathbb{E}\left[e^{\frac{\psi_{1}}{T}X_{T}^{1}+\frac{\psi_{2}}{T}X_{T}^{2}}\right]
\nonumber
\\
&=\lim_{T\rightarrow 0}T\log\exp\left(\frac{\psi_{1}^{2}}{2T^{2}}\frac{1-e^{-2k_{1}T}}{2k_{1}}
+\frac{\psi_{2}^{2}}{2T^{2}}\frac{1-e^{-2k_{2}T}}{2k_{2}}+\frac{\psi_{1}\psi_{2}}{T^{2}}\rho\frac{1-e^{-(k_{1}+k_{2})T}}{k_{1}+k_{2}}\right)
\nonumber
\\
&=\frac{\psi_{1}^{2}}{2}+\frac{\psi_{2}^{2}}{2}+\psi_{1}\psi_{2}\rho,
\end{align}
by G\"{a}rtner-Ellis theorem (see Theorem~\ref{GE:Thm}), $\mathbb{Q}(X_{T}\in\cdot)$ satisfies a large deviation principle (see Definition~\ref{defn:LDP})
with the rate function 
\begin{equation}
\sup_{\psi_{1},\psi_{2}\in\mathbb{R}}\left\{\psi_{1}x_{1}+\psi_{2}x_{2}-\frac{\psi_{1}^{2}}{2}-\frac{\psi_{2}^{2}}{2}-\psi_{1}\psi_{2}\rho\right\}=\frac{x_{1}^{2}+x_{2}^{2}-2x_{1}x_{2}\rho}{2(1-\rho^{2})}.
\end{equation}
Therefore, by contraction principle (see Theorem~\ref{Contraction:Thm}), we obtain
\begin{align}
&\lim_{T\rightarrow 0}T\log\mathbb{Q}\left(\frac{1}{\tau}\int_{0}^{\tau}\xi_{0}^{u}e^{\omega\alpha_{\theta}(\theta_{1}e^{-k_{1}u}X_{T}^{1}+\theta_{2}e^{-k_{2}u}X_{T}^{2})}du\geq K^{2}\right)
\nonumber
\\
&=-\inf_{x_{1},x_{2}:\frac{1}{\tau}\int_{0}^{\tau}\xi_{0}^{u}e^{\omega\alpha_{\theta}(\theta_{1}e^{-k_{1}u}x_{1}+\theta_{2}e^{-k_{2}u}x_{2})}du=K^{2}}\frac{x_{1}^{2}+x_{2}^{2}-2x_{1}x_{2}\rho}{2(1-\rho^{2})}.
\end{align}

The proof for the VIX put option is similar and hence omitted here.
This completes the proof.
\end{proof}


\subsection{Proof of Theorem~\ref{thm:short:ATM:two:factor}}

\begin{proof}
First of all, we have
\begin{align}
\mathrm{VIX}_{T}&=\left(\frac{1}{\tau}\int_{T}^{T+\tau}\xi_{T}^{u}du\right)^{1/2}
\nonumber
\\
&=\left(\frac{1}{\tau}\int_{T}^{T+\tau}\xi_{0}^{u}e^{\omega\alpha_{\theta}\left(\theta_{1}e^{-k_{1}(u-T)}X_{T}^{1}+\theta_{2}e^{-k_{2}(u-T)}X_{T}^{2}\right)-\frac{\omega^{2}}{2}v_{T}(u)}du\right)^{1/2}
\nonumber
\\
&=\left(\frac{1}{\tau}\int_{0}^{\tau}\xi_{0}^{u}du\right)^{1/2}
+\frac{\frac{1}{\tau}\int_{0}^{\tau}\xi_{0}^{u}\omega\alpha_{\theta}\left(\theta_{1}e^{-k_{1}u}X_{T}^{1}+\theta_{2}e^{-k_{2}u}X_{T}^{2}\right)du}{2\left(\frac{1}{\tau}\int_{0}^{\tau}\xi_{0}^{u}du\right)^{1/2}}+\mathcal{R}(T),
\end{align}
where 
\begin{align}
\mathcal{R}(T):= \left(\frac{1}{\tau}\int_{T}^{T+\tau}\xi_{T}^{u}du\right)^{1/2}-\left(\frac{1}{\tau}\int_{0}^{\tau}\xi_{0}^{u}du\right)^{1/2}
-\frac{\frac{1}{\tau}\int_{0}^{\tau}\xi_{0}^{u}\omega\alpha_{\theta}\left(\theta_{1}e^{-k_{1}u}X_{T}^{1}+\theta_{2}e^{-k_{2}u}X_{T}^{2}\right)du}{2\left(\frac{1}{\tau}\int_{0}^{\tau}\xi_{0}^{u}du\right)^{1/2}}
\end{align}
is a remainder term that can be shown to be negligible.
Indeed, since the map $x\mapsto x^{+}$ is $1$-Lipschitz and $\left(\frac{1}{\tau}\int_{0}^{\tau}\xi_{0}^{u}du\right)^{1/2}=K$, 
\begin{align}
\left|\mathbb{E}\left[\left(\mathrm{VIX}_{T}-K\right)^{+}\right]-\mathbb{E}\left[\left(\frac{\frac{1}{\tau}\int_{0}^{\tau}\xi_{0}^{u}\omega\alpha_{\theta}\left(\theta_{1}e^{-k_{1}u}X_{T}^{1}+\theta_{2}e^{-k_{2}u}X_{T}^{2}\right)du}{2\left(\frac{1}{\tau}\int_{0}^{\tau}\xi_{0}^{u}du\right)^{1/2}}\right)^{+}\right]
\right|
\leq\mathbb{E}|\mathcal{R}(T)|.
\end{align}
Similarly as in the proof of Theorem~\ref{thm:short:ATM:one:factor}, we can show that
\begin{equation}
\mathbb{E}|\mathcal{R}(T)|\leq O(T),
\end{equation}
as $T\rightarrow 0$.
Therefore,
\begin{align}
\lim_{T\rightarrow 0}\frac{C(T,\omega)}{\sqrt{T}}
&=\lim_{T\rightarrow 0}\frac{1}{\sqrt{T}}\mathbb{E}\left[\left(\frac{\frac{1}{\tau}\int_{0}^{\tau}\xi_{0}^{u}\omega\alpha_{\theta}\left(\theta_{1}e^{-k_{1}u}X_{T}^{1}+\theta_{2}e^{-k_{2}u}X_{T}^{2}\right)du}{2\left(\frac{1}{\tau}\int_{0}^{\tau}\xi_{0}^{u}du\right)^{1/2}}\right)^{+}\right]
\\
&=\lim_{T\rightarrow 0}\frac{1}{\sqrt{T}}\mathbb{E}\left[\left(\frac{\frac{1}{\tau}\int_{0}^{\tau}\xi_{0}^{u}\omega\alpha_{\theta}\theta_{1}e^{-k_{1}u}du\cdot X_{T}^{1}+\frac{1}{\tau}\int_{0}^{\tau}\xi_{0}^{u}\omega\alpha_{\theta}\theta_{2}e^{-k_{2}u}du\cdot X_{T}^{2}}{2\left(\frac{1}{\tau}\int_{0}^{\tau}\xi_{0}^{u}du\right)^{1/2}}\right)^{+}\right].
\end{align}
Since $X_{T}^{1},X_{T}^{2}$ are Gaussian with 
\begin{align}
\left(X_{T}^{1},X_{T}^{2}\right)\sim\mathcal{N}\left(0,
\left(\begin{array}{cc}
v_{T}^{1} & v_{T}^{1,2}
\\
v_{T}^{1,2} & v_{T}^{2}
\end{array}\right)\right),
\end{align}
where
\begin{align}
v_{T}^{i}:=\frac{1-e^{-2k_{i}T}}{2k_{i}},\quad i=1,2,
\qquad
v_{T}^{1,2}:=\rho\frac{1-e^{-(k_{1}+k_{2})T}}{k_{1}+k_{2}},
\end{align}
one can compute that $\frac{1}{\tau}\int_{0}^{\tau}\xi_{0}^{u}\omega\alpha_{\theta}\theta_{1}e^{-k_{1}u}du\cdot X_{T}^{1}+\frac{1}{\tau}\int_{0}^{\tau}\xi_{0}^{u}\omega\alpha_{\theta}\theta_{2}e^{-k_{2}u}du\cdot X_{T}^{2}$
is Gaussian with mean $0$ and variance 
\begin{align}
&\mathrm{Var}\left(\frac{1}{\tau}\int_{0}^{\tau}\xi_{0}^{u}\omega\alpha_{\theta}\theta_{1}e^{-k_{1}u}du\cdot X_{T}^{1}+\frac{1}{\tau}\int_{0}^{\tau}\xi_{0}^{u}\omega\alpha_{\theta}\theta_{2}e^{-k_{2}u}du\cdot X_{T}^{2}\right)
\\
&=\left(\frac{1}{\tau}\int_{0}^{\tau}\xi_{0}^{u}\omega\alpha_{\theta}\theta_{1}e^{-k_{1}u}du\right)^{2}v_{T}^{1}
+\left(\frac{1}{\tau}\int_{0}^{\tau}\xi_{0}^{u}\omega\alpha_{\theta}\theta_{2}e^{-k_{2}u}du\right)^{2}v_{T}^{2}
\nonumber
\\
&\qquad+2\left(\frac{1}{\tau}\int_{0}^{\tau}\xi_{0}^{u}\omega\alpha_{\theta}\theta_{1}e^{-k_{1}u}du\right)\left(\frac{1}{\tau}\int_{0}^{\tau}\xi_{0}^{u}\omega\alpha_{\theta}\theta_{2}e^{-k_{2}u}du\right)v_{T}^{1,2},
\end{align}
such that
\begin{align}
&\frac{1}{T}\mathrm{Var}\left(\frac{1}{\tau}\int_{0}^{\tau}\xi_{0}^{u}\omega\alpha_{\theta}\theta_{1}e^{-k_{1}u}du\cdot X_{T}^{1}+\frac{1}{\tau}\int_{0}^{\tau}\xi_{0}^{u}\omega\alpha_{\theta}\theta_{2}e^{-k_{2}u}du\cdot X_{T}^{2}\right)
\\
&\rightarrow
v(\tau):=\left(\frac{1}{\tau}\int_{0}^{\tau}\xi_{0}^{u}\omega\alpha_{\theta}\theta_{1}e^{-k_{1}u}du\right)^{2}
+\left(\frac{1}{\tau}\int_{0}^{\tau}\xi_{0}^{u}\omega\alpha_{\theta}\theta_{2}e^{-k_{2}u}du\right)^{2}
\nonumber
\\
&\qquad+2\left(\frac{1}{\tau}\int_{0}^{\tau}\xi_{0}^{u}\omega\alpha_{\theta}\theta_{1}e^{-k_{1}u}du\right)\left(\frac{1}{\tau}\int_{0}^{\tau}\xi_{0}^{u}\omega\alpha_{\theta}\theta_{2}e^{-k_{2}u}du\right)\rho,
\end{align}
as $T\rightarrow 0$.
Therefore, we have
\begin{align}
\lim_{T\rightarrow 0}\frac{C(T,\omega)}{\sqrt{T}}
=\frac{(v(\tau))^{1/2}}{2\left(\frac{1}{\tau}\int_{0}^{\tau}\xi_{0}^{u}du\right)^{1/2}}\mathbb{E}[Z^{+}],
\end{align}
where $Z\sim\mathcal{N}(0,1)$ such that $\mathbb{E}[Z^{+}]=\frac{1}{\sqrt{2\pi}}\int_{0}^{\infty}ze^{-\frac{z^{2}}{2}}dz=\frac{1}{\sqrt{2\pi}}$.
Hence, we conclude that
\begin{equation}
\lim_{T\rightarrow 0}\frac{C(T,\omega)}{\sqrt{T}}
=\frac{(v(\tau))^{1/2}}{2\left(\frac{1}{\tau}\int_{0}^{\tau}\xi_{0}^{u}du\right)^{1/2}}\frac{1}{\sqrt{2\pi}}.
\end{equation}
The proof for the VIX put option is similar and hence omitted here.
This completes the proof.
\end{proof}


\subsection{Proof of Theorem~\ref{thm:short:OTM:N:factor}}

\begin{proof}
Let us consider the case $\sqrt{\frac{1}{\tau}\int_{0}^{\tau}\xi_{0}^{u}du}<K$
and prove the result for OTM VIX call option. 

Similarly as in the proof of \eqref{short:to:show:1} in the proof of Theorem~\ref{thm:short:OTM:one:factor}, 
one can show that
 \begin{equation}
\lim_{T\rightarrow 0}T\log C(T,\omega)=\lim_{T\rightarrow 0}T\log\mathbb{Q}\left(\mathrm{VIX}_{T}^{2}\geq K^{2}\right),
\end{equation}
and similarly as in the proof of \eqref{short:to:show:2} in the proof of Theorem~\ref{thm:short:OTM:one:factor}, 
one can show that
\begin{equation}
\lim_{T\rightarrow 0}T\log\mathbb{Q}\left(\mathrm{VIX}_{T}^{2}\geq K^{2}\right)
=\lim_{T\rightarrow 0}T\log\mathbb{Q}\left(\frac{1}{\tau}\int_{0}^{\tau}\xi_{0}^{u}e^{\omega\alpha_{\theta}\sum_{i=1}^{N}\theta_{i}e^{-k_{i}u}X_{T}^{i}}du\geq K^{2}\right).
\end{equation}

Since $X_{T}:=(X_{T}^{1},\ldots,X_{T}^{N})$ is an $N$-dimensional Gaussian random vector
with mean $0$ and covariance matrix $\Sigma_{T}:=(v_{T}^{ij})_{1\leq i,j\leq N}$, 
with
\begin{align}
v_{T}^{ii}:=\frac{1-e^{-2k_{i}T}}{2k_{i}},\quad i=1,2,\ldots,N,
\qquad
v_{T}^{ij}:=\rho_{ij}\frac{1-e^{-(k_{i}+k_{j})T}}{k_{i}+k_{j}},\qquad i\neq j,
\end{align}
one can compute that for any $\psi\in\mathbb{R}^{N}$, 
\begin{align}
\lim_{T\rightarrow 0}T\log\mathbb{E}\left[e^{\langle\psi,\frac{1}{T}X_{T}\rangle}\right]
=\lim_{T\rightarrow 0}T\log\exp\left(\frac{1}{2T^{2}}\psi^{\top}\Sigma_{T}\psi\right)
=\frac{1}{2}\psi^{\top}\hat{\Sigma}_{0}\psi,
\end{align}
where $\hat{\Sigma}_{0}$ is an $N\times N$ symmetric matrix with $(i,j)$-th entry $\rho_{ij}1_{i=j}$,
by G\"{a}rtner-Ellis theorem (see Theorem~\ref{GE:Thm}), $\mathbb{Q}(X_{T}\in\cdot)$ satisfies a large deviation principle (see Definition~\ref{defn:LDP})
with the rate function 
\begin{equation}
\sup_{\psi\in\mathbb{R}^{N}}\left\{\langle\psi,x\rangle-\frac{1}{2}\psi^{\top}\hat{\Sigma}_{0}\psi\right\}=\frac{1}{2}x^{\top}\hat{\Sigma}_{0}^{-1}x.
\end{equation}
Therefore, by contraction principle (see Theorem~\ref{Contraction:Thm}), we obtain
\begin{align}
&\lim_{T\rightarrow 0}T\log\mathbb{Q}\left(\frac{1}{\tau}\int_{0}^{\tau}\xi_{0}^{u}e^{\omega\alpha_{\theta}\sum_{i=1}^{N}\theta_{i}e^{-k_{i}u}X_{T}^{i}}du\geq K^{2}\right)
\nonumber
\\
&=-\inf_{x=(x_{1},\ldots,x_{N}):\frac{1}{\tau}\int_{0}^{\tau}\xi_{0}^{u}e^{\omega\alpha_{\theta}\sum_{i=1}^{N}\theta_{i}e^{-k_{i}u}x_{i}}du=K^{2}}\frac{1}{2}x^{\top}\hat{\Sigma}_{0}^{-1}x.
\end{align}

The proof for the VIX put option is similar and hence omitted here.
This completes the proof.
\end{proof}


\subsection{Proof of Theorem~\ref{thm:short:ATM:N:factor}}

\begin{proof}
The proof of Theorem~\ref{thm:short:ATM:N:factor} follows
the same arguments as in the proof of Theorem~\ref{thm:short:ATM:one:factor}
and Theorem~\ref{thm:short:ATM:two:factor}
and is hence omitted here.
\end{proof}

\subsection{Proof of Theorem~\ref{thm:small:OTM:one:factor}}

\begin{proof}
Let us consider the case $\sqrt{\frac{1}{\tau}\int_{T}^{T+\tau}\xi_{0}^{u}du}<K$
and prove the result for OTM VIX call option. 

First, we can compute that
\begin{equation}\label{small:to:show:1}
\lim_{\omega\rightarrow 0}\omega^{2}\log C(T,\omega)=\lim_{\omega\rightarrow 0}\omega^{2}\log\mathbb{Q}\left(\mathrm{VIX}_{T}^{2}\geq K^{2}\right),
\end{equation}
provided that the limit on the right hand side exists.
The proof of \eqref{small:to:show:1} is similar to that of \eqref{short:to:show:1}
and is hence omitted here.

Next, we can show that
\begin{equation}\label{small:to:show:2}
\lim_{\omega\rightarrow 0}\omega^{2}\log\mathbb{Q}\left(\mathrm{VIX}_{T}^{2}\geq K^{2}\right)
=\lim_{\omega\rightarrow 0}\omega^{2}\log\mathbb{Q}\left(\frac{1}{\tau}\int_{T}^{T+\tau}\xi_{0}^{u}e^{\omega e^{-k(u-T)}X_{T}}du\geq K^{2}\right),
\end{equation}
provided that the limit on the right hand side exists which is continuous in $K$.

Let us prove \eqref{small:to:show:2}.
One can compute that
\begin{align}
&\limsup_{\omega\rightarrow 0}\omega^{2}\log\mathbb{Q}\left(\mathrm{VIX}_{T}^{2}\geq K^{2}\right)\nonumber
\\
&=\limsup_{\omega\rightarrow 0}\omega^{2}\log\mathbb{Q}\left(\frac{1}{\tau}\int_{T}^{T+\tau}\xi_{0}^{u}e^{\omega e^{-k(u-T)}X_{T}-\frac{\omega^{2}}{2}e^{-2k(u-T)}v_{T}}du\geq K^{2}\right)
\nonumber
\\
&\leq
\limsup_{\omega\rightarrow 0}\omega^{2}\log\mathbb{Q}\left(\frac{1}{\tau}\int_{T}^{T+\tau}\xi_{0}^{u}e^{\omega e^{-k(u-T)}X_{T}}du\geq K^{2}\right),
\end{align}
and for any $\epsilon>0$, we have $e^{-\frac{\omega^{2}}{2}e^{-2k(u-T)}v_{T}}\geq e^{-\frac{\omega^{2}}{2}e^{-2k\tau}v_{T}}\geq\frac{1}{1+\epsilon}$
for any sufficiently small $\omega$, where $T\leq u\leq T+\tau$, such that
\begin{align}
\liminf_{\omega\rightarrow 0}\omega^{2}\log\mathbb{Q}\left(\mathrm{VIX}_{T}^{2}\geq K^{2}\right)
&=\limsup_{\omega\rightarrow 0}\omega^{2}\log\mathbb{Q}\left(\frac{1}{\tau}\int_{T}^{T+\tau}\xi_{0}^{u}e^{\omega e^{-k(u-T)}X_{T}-\frac{\omega^{2}}{2}e^{-2k(u-T)}v_{T}}du\geq K^{2}\right)
\nonumber
\\
&\geq
\liminf_{\omega\rightarrow 0}\omega^{2}\log\mathbb{Q}\left(\frac{1}{\tau}\int_{T}^{T+\tau}\xi_{0}^{u}e^{\omega e^{-k(u-T)}X_{T}}du\geq(1+\epsilon)K^{2}\right).
\end{align}
Since $\epsilon>0$ is arbitrary, we conclude that
\begin{equation}
\lim_{\omega\rightarrow 0}\omega^{2}\log\mathbb{Q}\left(\mathrm{VIX}_{T}^{2}\geq K^{2}\right)
=\lim_{\omega\rightarrow 0}\omega^{2}\log\mathbb{Q}\left(\frac{1}{\tau}\int_{T}^{T+\tau}\xi_{0}^{u}e^{\omega e^{-k(u-T)}X_{T}}du\geq K^{2}\right),
\end{equation}
provided that the limit on the right hand side exists which is continuous in $K$.
This proves \eqref{small:to:show:2}.

Since $X_{T}$ is Gaussian with mean $0$ and variance $\frac{1-e^{-2kT}}{2k}$,
one can compute that for any $\theta\in\mathbb{R}$, 
\begin{equation}
\lim_{\omega\rightarrow 0}\omega^{2}\log\mathbb{E}\left[e^{\frac{\theta}{\omega^{2}}\omega X_{T}}\right]=
\lim_{\omega\rightarrow 0}\omega^{2}\log\exp\left(\frac{\theta^{2}}{2\omega^{2}}\frac{1-e^{-2kT}}{2k}\right)
=\frac{\theta^{2}}{2}\frac{1-e^{-2kT}}{2k},
\end{equation}
by G\"{a}rtner-Ellis theorem (see Theorem~\ref{GE:Thm}),
we conclude that 
$\mathbb{Q}(\omega X_{T}\in\cdot)$ satisfies a large deviation principle (see Definition~\ref{defn:LDP})
with the rate function 
\begin{align}
\sup_{\theta\in\mathbb{R}}\left\{\theta x-\frac{\theta^{2}}{2}\frac{1-e^{-2kT}}{2k}\right\}=\frac{kx^{2}}{1-e^{-2kT}}. 
\end{align}
Therefore, by contraction principle (see Theorem~\ref{Contraction:Thm}), we obtain
\begin{equation}
\lim_{\omega\rightarrow 0}\omega^{2}\log\mathbb{Q}\left(\frac{1}{\tau}\int_{T}^{T+\tau}\xi_{0}^{u}e^{\omega e^{-k(u-T)}X_{T}}du\geq K^{2}\right)
=-\inf_{x:\frac{1}{\tau}\int_{T}^{T+\tau}\xi_{0}^{u}e^{e^{-k(u-T)}x}du=K^{2}}\frac{kx^{2}}{1-e^{-2kT}}.
\end{equation}
Hence, we conclude that
\begin{equation}
\lim_{\omega\rightarrow 0}\omega^{2}\log C(T,\omega)=-\frac{ky_{+}^{2}}{1-e^{-2kT}},
\end{equation}
where $y_{+}>0$ is the unique positive value such that
\begin{equation}
\frac{1}{\tau}\int_{T}^{T+\tau}\xi_{0}^{u}\exp\left(e^{-k(u-T)}y_{+}\right)du=K^{2}.
\end{equation}

The argument for OTM VIX put option
is similar and is hence omitted here.
The proof is complete.
\end{proof}

\subsection{Proof of Theorem~\ref{thm:small:ATM:one:factor}}

\begin{proof}
First of all, we have
\begin{align}
\mathrm{VIX}_{T}&=\left(\frac{1}{\tau}\int_{T}^{T+\tau}\xi_{T}^{u}du\right)^{1/2}
\nonumber
\\
&=\left(\frac{1}{\tau}\int_{T}^{T+\tau}\xi_{0}^{u}e^{\omega e^{-k(u-T)}X_{T}-\frac{\omega^{2}}{2}e^{-2k(u-T)}v_{T}}du\right)^{1/2}
\nonumber
\\
&=\left(\frac{1}{\tau}\int_{T}^{T+\tau}\xi_{0}^{u}du\right)^{1/2}
+\frac{\frac{1}{\tau}\int_{T}^{T+\tau}\xi_{0}^{u}\omega e^{-k(u-T)}X_{T}du}{2\left(\frac{1}{\tau}\int_{T}^{T+\tau}\xi_{0}^{u}du\right)^{1/2}}+\mathcal{R}(\omega),
\end{align}
where
\begin{align}
\mathcal{R}(\omega):= \left(\frac{1}{\tau}\int_{T}^{T+\tau}\xi_{T}^{u}du\right)^{1/2}-\left(\frac{1}{\tau}\int_{T}^{T+\tau}\xi_{0}^{u}du\right)^{1/2}
-\frac{\frac{1}{\tau}\int_{T}^{T+\tau}\xi_{0}^{u}\omega e^{-k(u-T)}X_{T}du}{2\left(\frac{1}{\tau}\int_{T}^{T+\tau}\xi_{0}^{u}du\right)^{1/2}}
\end{align}
is a remainder term that can be shown to be negligible.
Indeed, since the map $x\mapsto x^{+}$ is $1$-Lipschitz and $\left(\frac{1}{\tau}\int_{T}^{T+\tau}\xi_{0}^{u}du\right)^{1/2}=K$, 
\begin{align}
\left|\mathbb{E}\left[\left(\mathrm{VIX}_{T}-K\right)^{+}\right]-\mathbb{E}\left[\left(\frac{\frac{1}{\tau}\int_{T}^{T+\tau}\xi_{0}^{u}\omega e^{-k(u-T)}X_{T}du}{2\left(\frac{1}{\tau}\int_{T}^{T+\tau}\xi_{0}^{u}du\right)^{1/2}}\right)^{+}\right]
\right|
\leq\mathbb{E}|\mathcal{R}(\omega)|.
\end{align}
Similarly as in the proof of Theorem~\ref{thm:short:ATM:one:factor}, we can show that
\begin{equation}
\mathbb{E}|\mathcal{R}(\omega)|\leq O(\omega^{2}),
\end{equation}
as $\omega\rightarrow 0$.
Therefore,
\begin{equation}
\lim_{\omega\rightarrow 0}\frac{C(T,\omega)}{\omega}
=\mathbb{E}\left[\left(\frac{\frac{1}{\tau}\int_{T}^{T+\tau}\xi_{0}^{u}e^{-k(u-T)}X_{T}du}{2\left(\frac{1}{\tau}\int_{T}^{T+\tau}\xi_{0}^{u}du\right)^{1/2}}\right)^{+}\right]
=\frac{\frac{1}{\tau}\int_{T}^{T+\tau}\xi_{0}^{u}e^{-k(u-T)}du}{2\left(\frac{1}{\tau}\int_{T}^{T+\tau}\xi_{0}^{u}du\right)^{1/2}}\mathbb{E}\left[\left(X_{T}\right)^{+}\right].
\end{equation}
Since $X_{T}$ is Gaussian with mean $0$ and variance $\frac{1-e^{-2kT}}{2k}$, 
\begin{equation}
\mathbb{E}\left[\left(X_{T}\right)^{+}\right]
=\left(\frac{1-e^{-2kT}}{2k}\right)^{1/2}\mathbb{E}[Z^{+}],
\end{equation}
where $Z\sim\mathcal{N}(0,1)$ such that $\mathbb{E}[Z^{+}]=\frac{1}{\sqrt{2\pi}}\int_{0}^{\infty}ze^{-\frac{z^{2}}{2}}dz=\frac{1}{\sqrt{2\pi}}$.
Hence, we conclude that
\begin{equation}
\lim_{\omega\rightarrow 0}\frac{C(T,\omega)}{\omega}
=\frac{\frac{1}{\tau}\int_{T}^{T+\tau}\xi_{0}^{u}e^{-k(u-T)}du}{2\left(\frac{1}{\tau}\int_{T}^{T+\tau}\xi_{0}^{u}du\right)^{1/2}}
\left(\frac{1-e^{-2kT}}{2k}\right)^{1/2}\frac{1}{\sqrt{2\pi}}.
\end{equation}
The proof for the VIX put option is similar and hence omitted here.
\end{proof}


\subsection{Proof of Theorem~\ref{thm:small:OTM:two:factor}}

\begin{proof}
Let us consider the case $\sqrt{\frac{1}{\tau}\int_{0}^{\tau}\xi_{0}^{u}du}<K$
and prove the result for OTM VIX call option. 

Similarly as in the proof of \eqref{small:to:show:1} in the proof of Theorem~\ref{thm:small:OTM:one:factor}, 
one can show that
 \begin{equation}
\lim_{\omega\rightarrow 0}\omega^{2}\log C(T,\omega)=\lim_{\omega\rightarrow 0}\omega^{2}\log\mathbb{Q}\left(\mathrm{VIX}_{T}^{2}\geq K^{2}\right),
\end{equation}
and similarly as in the proof of \eqref{small:to:show:2} in the proof of Theorem~\ref{thm:small:OTM:one:factor}, 
one can show that
\begin{align}
&\lim_{\omega\rightarrow 0}\omega^{2}\log\mathbb{Q}\left(\mathrm{VIX}_{T}^{2}\geq K^{2}\right)
\nonumber
\\
&=\lim_{\omega\rightarrow 0}\omega^{2}\log\mathbb{Q}\left(\frac{1}{\tau}\int_{T}^{T+\tau}\xi_{0}^{u}e^{\omega\alpha_{\theta}(\theta_{1}e^{-k_{1}(u-T)}X_{T}^{1}+\theta_{2}e^{-k_{2}(u-T)}X_{T}^{2})}du\geq K^{2}\right).
\end{align}

Since $X_{T}^{1},X_{T}^{2}$ are Gaussian with 
\begin{align}
\left(X_{T}^{1},X_{T}^{2}\right)\sim\mathcal{N}\left(0,
\left(\begin{array}{cc}
v_{T}^{1} & v_{T}^{1,2}
\\
v_{T}^{1,2} & v_{T}^{2}
\end{array}\right)\right),
\end{align}
where
\begin{align}
v_{T}^{i}:=\frac{1-e^{-2k_{i}T}}{2k_{i}},\quad i=1,2,
\qquad
v_{T}^{1,2}:=\rho\frac{1-e^{-(k_{1}+k_{2})T}}{k_{1}+k_{2}},
\end{align}
one can compute that for any $\psi_{1},\psi_{2}\in\mathbb{R}$, 
\begin{align}
\lim_{\omega\rightarrow 0}\omega^{2}\log\mathbb{E}\left[e^{\frac{\psi_{1}}{\omega}X_{T}^{1}+\frac{\psi_{2}}{\omega}X_{T}^{2}}\right]
&=\lim_{\omega\rightarrow 0}\omega^{2}\log\exp\left(\frac{\psi_{1}^{2}}{2\omega^{2}}v_{T}^{1}
+\frac{\psi_{2}^{2}}{2\omega^{2}}v_{T}^{2}+\frac{\psi_{1}\psi_{2}}{\omega^{2}}v_{T}^{1,2}\right)
\nonumber
\\
&=\frac{\psi_{1}^{2}}{2}v_{T}^{1}+\frac{\psi_{2}^{2}}{2}v_{T}^{2}+\psi_{1}\psi_{2}v_{T}^{1,2},
\end{align}
by G\"{a}rtner-Ellis theorem (see Theorem~\ref{GE:Thm}), $\mathbb{Q}(X_{T}\in\cdot)$ satisfies a large deviation principle (see Definition~\ref{defn:LDP})
with the rate function 
\begin{align}
&\sup_{\psi_{1},\psi_{2}\in\mathbb{R}}\left\{\psi_{1}x_{1}+\psi_{2}x_{2}-\frac{\psi_{1}^{2}}{2}v_{T}^{1}-\frac{\psi_{2}^{2}}{2}v_{T}^{2}-\psi_{1}\psi_{2}v_{T}^{1,2}\right\}
\nonumber
\\
&=\sup_{\psi_{1},\psi_{2}\in\mathbb{R}}\left\{\frac{\psi_{1}}{\sqrt{v_{T}^{1}}}x_{1}+\frac{\psi_{2}}{\sqrt{v_{T}^{2}}}x_{2}-\frac{\psi_{1}^{2}}{2}-\frac{\psi_{2}^{2}}{2}-\frac{\psi_{1}}{\sqrt{v_{T}^{1}}}\frac{\psi_{2}}{\sqrt{v_{T}^{2}}} v_{T}^{1,2}\right\}
\nonumber
\\
&=\frac{\frac{x_{1}^{2}}{v_{T}^{1}}+\frac{x_{2}^{2}}{v_{T}^{2}}-2\frac{x_{1}}{\sqrt{v_{T}^{1}}}\frac{x_{2}}{\sqrt{v_{T}^{2}}}\frac{v_{T}^{1,2}}{\sqrt{v_{T}^{1}}\sqrt{v_{T}^{2}}}}{2\left(1-\frac{\left(v_{T}^{1,2}\right)^{2}}{v_{T}^{1}v_{T}^{2}}\right)}=\frac{v_{T}^{2}x_{1}^{2}+v_{T}^{1}x_{2}^{2}-2x_{1}x_{2}v_{T}^{1,2}}{2\left(v_{T}^{1}v_{T}^{2}-\left(v_{T}^{1,2}\right)^{2}\right)}.
\end{align}
Therefore, by contraction principle (see Theorem~\ref{Contraction:Thm}), we obtain
\begin{align}
&\lim_{T\rightarrow 0}T\log\mathbb{Q}\left(\frac{1}{\tau}\int_{T}^{T+\tau}\xi_{0}^{u}e^{\omega\alpha_{\theta}(\theta_{1}e^{-k_{1}(u-T)}X_{T}^{1}+\theta_{2}e^{-k_{2}(u-T)}X_{T}^{2})}du\geq K^{2}\right)
\nonumber
\\
&=-\inf_{x_{1},x_{2}:\frac{1}{\tau}\int_{T}^{T+\tau}\xi_{0}^{u}e^{\alpha_{\theta}(\theta_{1}e^{-k_{1}(u-T)}x_{1}+\theta_{2}e^{-k_{2}(u-T)}x_{2})}du=K^{2}}\frac{v_{T}^{2}x_{1}^{2}+v_{T}^{1}x_{2}^{2}-2x_{1}x_{2}v_{T}^{1,2}}{2\left(v_{T}^{1}v_{T}^{2}-\left(v_{T}^{1,2}\right)^{2}\right)}.
\end{align}

The proof for the VIX put option is similar and hence omitted here.
This completes the proof.
\end{proof}


\subsection{Proof of Theorem~\ref{thm:small:ATM:two:factor}}

\begin{proof}
First of all, we have
\begin{align}
\mathrm{VIX}_{T}&=\left(\frac{1}{\tau}\int_{T}^{T+\tau}\xi_{T}^{u}du\right)^{1/2}
\nonumber
\\
&=\left(\frac{1}{\tau}\int_{T}^{T+\tau}\xi_{0}^{u}e^{\omega\alpha_{\theta}\left(\theta_{1}e^{-k_{1}(u-T)}X_{T}^{1}+\theta_{2}e^{-k_{2}(u-T)}X_{T}^{2}\right)-\frac{\omega^{2}}{2}v_{T}(u)}du\right)^{1/2}
\nonumber
\\
&=\left(\frac{1}{\tau}\int_{T}^{T+\tau}\xi_{0}^{u}du\right)^{1/2}
+\frac{\frac{1}{\tau}\int_{T}^{T+\tau}\xi_{0}^{u}\omega\alpha_{\theta}\left(\theta_{1}e^{-k_{1}(u-T)}X_{T}^{1}+\theta_{2}e^{-k_{2}(u-T)}X_{T}^{2}\right)du}{2\left(\frac{1}{\tau}\int_{T}^{T+\tau}\xi_{0}^{u}du\right)^{1/2}}+\mathcal{R}(\omega),
\end{align}
where
\begin{align}
\mathcal{R}(\omega)&:= \left(\frac{1}{\tau}\int_{T}^{T+\tau}\xi_{T}^{u}du\right)^{1/2}-\left(\frac{1}{\tau}\int_{T}^{T+\tau}\xi_{0}^{u}du\right)^{1/2}
\nonumber
\\
&\qquad\qquad
-\frac{\frac{1}{\tau}\int_{T}^{T+\tau}\xi_{0}^{u}\omega\alpha_{\theta}\left(\theta_{1}e^{-k_{1}(u-T)}X_{T}^{1}+\theta_{2}e^{-k_{2}(u-T)}X_{T}^{2}\right)du}{2\left(\frac{1}{\tau}\int_{T}^{T+\tau}\xi_{0}^{u}du\right)^{1/2}}
\end{align}
is a remainder term that can be shown to be negligible.
Indeed, since the map $x\mapsto x^{+}$ is $1$-Lipschitz and $\left(\frac{1}{\tau}\int_{T}^{T+\tau}\xi_{0}^{u}du\right)^{1/2}=K$, 
\begin{align}
\left|\mathbb{E}\left[\left(\mathrm{VIX}_{T}-K\right)^{+}\right]-\mathbb{E}\left[\left(\frac{\frac{1}{\tau}\int_{T}^{T+\tau}\xi_{0}^{u}\omega\alpha_{\theta}\left(\theta_{1}e^{-k_{1}(u-T)}X_{T}^{1}+\theta_{2}e^{-k_{2}(u-T)}X_{T}^{2}\right)du}{2\left(\frac{1}{\tau}\int_{T}^{T+\tau}\xi_{0}^{u}du\right)^{1/2}}\right)^{+}\right]
\right|
\leq\mathbb{E}|\mathcal{R}(\omega)|.
\end{align}
Similarly as in the proof of Theorem~\ref{thm:short:ATM:one:factor}, we can show that
\begin{equation}
\mathbb{E}|\mathcal{R}(\omega)|\leq O(\omega^{2}),
\end{equation}
as $\omega\rightarrow 0$.
Therefore,
\begin{align}
\lim_{\omega\rightarrow 0}\frac{C(T,\omega)}{\omega}
&=\lim_{\omega\rightarrow 0}\frac{1}{\omega}\mathbb{E}\left[\left(\frac{\frac{1}{\tau}\int_{T}^{T+\tau}\xi_{0}^{u}\omega\alpha_{\theta}\left(\theta_{1}e^{-k_{1}(u-T)}X_{T}^{1}+\theta_{2}e^{-k_{2}(u-T)}X_{T}^{2}\right)du}{2\left(\frac{1}{\tau}\int_{T}^{T+\tau}\xi_{0}^{u}du\right)^{1/2}}\right)^{+}\right]
\\
&=\mathbb{E}\left[\left(\frac{\frac{1}{\tau}\int_{T}^{T+\tau}\xi_{0}^{u}\alpha_{\theta}\theta_{1}e^{-k_{1}(u-T)}du\cdot X_{T}^{1}+\frac{1}{\tau}\int_{T}^{T+\tau}\xi_{0}^{u}\alpha_{\theta}\theta_{2}e^{-k_{2}(u-T)}du\cdot X_{T}^{2}}{2\left(\frac{1}{\tau}\int_{T}^{T+\tau}\xi_{0}^{u}du\right)^{1/2}}\right)^{+}\right].
\end{align}
Since $X_{T}^{1},X_{T}^{2}$ are Gaussian with 
\begin{align}
\left(X_{T}^{1},X_{T}^{2}\right)\sim\mathcal{N}\left(0,
\left(\begin{array}{cc}
v_{T}^{1} & v_{T}^{1,2}
\\
v_{T}^{1,2} & v_{T}^{2}
\end{array}\right)\right),
\end{align}
where
\begin{align}
v_{T}^{i}:=\frac{1-e^{-2k_{i}T}}{2k_{i}},\quad i=1,2,
\qquad
v_{T}^{1,2}:=\rho\frac{1-e^{-(k_{1}+k_{2})T}}{k_{1}+k_{2}},
\end{align}
one can compute that $\frac{1}{\tau}\int_{T}^{T+\tau}\xi_{0}^{u}\alpha_{\theta}\theta_{1}e^{-k_{1}(u-T)}du\cdot X_{T}^{1}+\frac{1}{\tau}\int_{T}^{T+\tau}\xi_{0}^{u}\alpha_{\theta}\theta_{2}e^{-k_{2}(u-T)}du\cdot X_{T}^{2}$
is Gaussian with mean $0$ and variance 
\begin{align}
\hat{v}(\tau)&:=\mathrm{Var}\left(\frac{1}{\tau}\int_{T}^{T+\tau}\xi_{0}^{u}\alpha_{\theta}\theta_{1}e^{-k_{1}(u-T)}du\cdot X_{T}^{1}+\frac{1}{\tau}\int_{T}^{T+\tau}\xi_{0}^{u}\alpha_{\theta}\theta_{2}e^{-k_{2}(u-T)}du\cdot X_{T}^{2}\right)
\\
&=\left(\frac{1}{\tau}\int_{T}^{T+\tau}\xi_{0}^{u}\alpha_{\theta}\theta_{1}e^{-k_{1}(u-T)}du\right)^{2}v_{T}^{1}
+\left(\frac{1}{\tau}\int_{T}^{T+\tau}\xi_{0}^{u}\alpha_{\theta}\theta_{2}e^{-k_{2}(u-T)}du\right)^{2}v_{T}^{2}
\nonumber
\\
&\qquad+2\left(\frac{1}{\tau}\int_{T}^{T+\tau}\xi_{0}^{u}\alpha_{\theta}\theta_{1}e^{-k_{1}(u-T)}du\right)\left(\frac{1}{\tau}\int_{T}^{T+\tau}\xi_{0}^{u}\alpha_{\theta}\theta_{2}e^{-k_{2}(u-T)}du\right)v_{T}^{1,2}.
\end{align}
Therefore, we have
\begin{align}
\lim_{\omega\rightarrow 0}\frac{C(T,\omega)}{\omega}
=\frac{(\hat{v}(\tau))^{1/2}}{2\left(\frac{1}{\tau}\int_{T}^{T+\tau}\xi_{0}^{u}du\right)^{1/2}}\mathbb{E}[Z^{+}],
\end{align}
where $Z\sim\mathcal{N}(0,1)$ such that $\mathbb{E}[Z^{+}]=\frac{1}{\sqrt{2\pi}}\int_{0}^{\infty}ze^{-\frac{z^{2}}{2}}dz=\frac{1}{\sqrt{2\pi}}$.
Hence, we conclude that
\begin{equation}
\lim_{\omega\rightarrow 0}\frac{C(T,\omega)}{\omega}
=\frac{(\hat{v}(\tau))^{1/2}}{2\left(\frac{1}{\tau}\int_{T}^{T+\tau}\xi_{0}^{u}du\right)^{1/2}}\frac{1}{\sqrt{2\pi}}.
\end{equation}
The proof for the VIX put option is similar and hence omitted here.
This completes the proof.
\end{proof}


\subsection{Proof of Theorem~\ref{thm:small:OTM:N:factor}}

\begin{proof}
Let us consider the case $\sqrt{\frac{1}{\tau}\int_{0}^{\tau}\xi_{0}^{u}du}<K$
and prove the result for OTM VIX call option. 

Similarly as in the proof of \eqref{short:to:show:1} in the proof of Theorem~\ref{thm:short:OTM:one:factor}, 
one can show that
 \begin{equation}
\lim_{\omega\rightarrow 0}\omega^{2}\log C(T,\omega)=\lim_{\omega\rightarrow 0}\omega^{2}\log\mathbb{Q}\left(\mathrm{VIX}_{T}^{2}\geq K^{2}\right),
\end{equation}
and similarly as in the proof of \eqref{short:to:show:2} in the proof of Theorem~\ref{thm:short:OTM:one:factor}, 
one can show that
\begin{equation}
\lim_{\omega\rightarrow 0}\omega^{2}\log\mathbb{Q}\left(\mathrm{VIX}_{T}^{2}\geq K^{2}\right)
=\lim_{\omega\rightarrow 0}\omega^{2}\log\mathbb{Q}\left(\frac{1}{\tau}\int_{0}^{\tau}\xi_{0}^{u}e^{\omega\alpha_{\theta}\sum_{i=1}^{N}\theta_{i}e^{-k_{i}u}X_{T}^{i}}du\geq K^{2}\right).
\end{equation}

Since $X_{T}:=(X_{T}^{1},\ldots,X_{T}^{N})$ is an $N$-dimensional Gaussian random vector
with mean $0$ and covariance matrix $\Sigma_{T}:=(v_{T}^{ij})_{1\leq i,j\leq N}$, 
with
\begin{align}
v_{T}^{ii}:=\frac{1-e^{-2k_{i}T}}{2k_{i}},\quad i=1,2,\ldots,N,
\qquad
v_{T}^{ij}:=\rho_{ij}\frac{1-e^{-(k_{i}+k_{j})T}}{k_{i}+k_{j}},\qquad i\neq j,
\end{align}
one can compute that for any $\psi\in\mathbb{R}^{N}$, 
\begin{align}
\lim_{\omega\rightarrow 0}\omega^{2}\log\mathbb{E}\left[e^{\langle\psi,\frac{1}{\omega}X_{T}\rangle}\right]
=\lim_{\omega\rightarrow 0}\omega^{2}\log\exp\left(\frac{1}{2\omega^{2}}\psi^{\top}\Sigma_{T}\psi\right)
=\frac{1}{2}\psi^{\top}\Sigma_{T}\psi,
\end{align}
by G\"{a}rtner-Ellis theorem (see Theorem~\ref{GE:Thm}), $\mathbb{Q}(X_{T}\in\cdot)$ satisfies a large deviation principle (see Definition~\ref{defn:LDP})
with the rate function 
\begin{equation}
\sup_{\psi\in\mathbb{R}^{N}}\left\{\langle\psi,x\rangle-\frac{1}{2}\psi^{\top}\Sigma_{T}\psi\right\}=\frac{1}{2}x^{\top}\Sigma_{T}^{-1}x.
\end{equation}
Therefore, by contraction principle (see Theorem~\ref{Contraction:Thm}), we obtain
\begin{align}
&\lim_{\omega\rightarrow 0}\omega^{2}\log\mathbb{Q}\left(\frac{1}{\tau}\int_{T}^{T+\tau}\xi_{0}^{u}e^{\omega\alpha_{\theta}\sum_{i=1}^{N}\theta_{i}e^{-k_{i}(u-T)}X_{T}^{i}}du\geq K^{2}\right)
\nonumber
\\
&=-\inf_{x=(x_{1},\ldots,x_{N}):\frac{1}{\tau}\int_{T}^{T+\tau}\xi_{0}^{u}e^{\alpha_{\theta}\sum_{i=1}^{N}\theta_{i}e^{-k_{i}(u-T)}x_{i}}du=K^{2}}\frac{1}{2}x^{\top}\Sigma_{T}^{-1}x.
\end{align}

The proof for the VIX put option is similar and hence omitted here.
This completes the proof.
\end{proof}


\subsection{Proof of Theorem~\ref{thm:small:ATM:N:factor}}

\begin{proof}
The proof of Theorem~\ref{thm:small:ATM:N:factor} follows
the same arguments as in the proof of Theorem~\ref{thm:small:ATM:one:factor}
and Theorem~\ref{thm:small:ATM:two:factor}
and is hence omitted here.
\end{proof}


\subsection{Proof of Proposition~\ref{prop:JV2F}}

\begin{proof}
By Theorem~\ref{thm:short:OTM:two:factor}, the rate function for VIX options in the two-factor Bergomi model is expressed as a constrained optimization problem
\begin{equation}
I_V(K) = \inf_{x_1,x_2} \frac{1}{2(1-\rho^2)}
\left( x_1^2 + x_2^2 - 2\rho x_1 x_2 \right)\,,
\end{equation}
where $x_{1,2}$ are constrained as
\begin{equation}\label{x12const}
\frac{1}{\tau} \int_0^\tau \xi_0^u e^{\omega \alpha_\theta [\theta_1 e^{-k_1 u} x_1 + \theta_2 e^{-k_2 u}
x_2]} du = K^2 = F_0^2 e^{2x_{\mathrm{VIX}}}\,,
\end{equation}
where $x_{\mathrm{VIX}}=\log(K/F_0)$ is the log-moneyness of the VIX option.

The constraint is included in the usual way by introducing a Lagrange multiplier $\lambda$ by considering the auxiliary function
\begin{align}
\Lambda(x_1,x_2,\lambda) &= 
\frac{1}{2(1-\rho^2)}
\left( x_1^2 + x_2^2 - 2\rho x_1 x_2 \right) \\
&\qquad\qquad+ \lambda
\left( \frac{1}{\tau} \int_0^\tau \xi_0^u e^{\omega \alpha_\theta 
[\theta_1 e^{-k_1 u} x_1 + \theta_2 e^{-k_2 u} x_2]} du - K^2 \right) \,.\nonumber
\end{align} 

The extremal conditions for this function $\partial_{x_i} \Lambda(x_i,\lambda) = 0$ can be put into the form
\begin{align}\label{x1eq}
&x_1 + \lambda(\omega \alpha_\theta) \frac{1}{\tau} \int_0^\tau \xi_0^u 
\left(\theta_1 e^{-k_1 u} + \rho \theta_2 e^{-k_2 u}\right) e^{T(u)} du = 0\,, \\
\label{x2eq}
&x_2 + \lambda(\omega \alpha_\theta) \frac{1}{\tau} \int_0^\tau \xi_0^u 
\left(\rho \theta_1 e^{-k_1 u} +  \theta_2 e^{-k_2 u}\right) e^{T(u)} du = 0 \,,
\end{align}
where we defined $T(u) := \omega \alpha_\theta 
[\theta_1 e^{-k_1 u} x_1 + \theta_2 e^{-k_2 u} x_2]$.

These are two non-linear equations for $(x_1,x_2)$.  Together with the constraint \eqref{x12const} they form a system of 3 non-linear equations for $(x_i,\lambda)$. We seek a solution of these equations in an expansion of the log-moneyness $x_{\mathrm{VIX}}$. For an ATM VIX option, these equations are satisfied by $x_1=x_2=0$. It is reasonable to expect that, for sufficiently small $x_{\mathrm{VIX}}$, the solutions $x_i$ will be also small. We expand them as
\begin{align}
x_1 &= c_{11} x_{\mathrm{VIX}} + c_{12} x^2_{\mathrm{VIX}} + O\left(x^3_{\mathrm{VIX}}\right)\,,\\
x_2 &= c_{21} x_{\mathrm{VIX}} + c_{22} x^2_{\mathrm{VIX}} + O\left(x^3_{\mathrm{VIX}}\right)\,,\\
\lambda &= \lambda_0 x_{\mathrm{VIX}} + \lambda_1 x^2_{\mathrm{VIX}} + O\left(x^3_{\mathrm{VIX}}\right)\,.
\end{align}
The coefficients in these expansions can be determined order by order in $x_{\mathrm{VIX}}$. We start by considering the $O(x_{\mathrm{VIX}})$ terms in the expansion of the master equations \eqref{x1eq}, \eqref{x2eq}.

\textit{$O(x_{\mathrm{VIX}})$ terms.} Expanding the equations \eqref{x1eq}, \eqref{x2eq} and the constraint \eqref{x12const} in $x_{\mathrm{VIX}}$ and keeping only the linear terms gives
\begin{align}
& c_{11} + \lambda_0 (\omega \alpha_\theta) \frac{1}{\tau} \int_0^\tau \xi_0^u \left(\theta_1 e^{-k_1 u} +
\rho \theta_2 e^{-k_2 u}\right) du = 0\,, \\
& c_{21} + \lambda_0 (\omega \alpha_\theta) \frac{1}{\tau} \int_0^\tau \xi_0^u \left(\rho\theta_1 e^{-k_1 u} +
\theta_2 e^{-k_2 u}\right) du = 0\,.
\end{align}
The constraint \eqref{x12const} gives the additional equation
\begin{equation}
\omega \alpha_\theta \left( \theta_1 \left(\frac{1}{\tau} \int_0^\tau \xi_0^u e^{-k_1 u}du\right) c_{11} +
\theta_2 \left(\frac{1}{\tau} \int_0^\tau \xi_0^u e^{-k_2 u}du\right) c_{21}  \right) =2F_0^2\,.
\end{equation}
The integrals can be written in terms of the parameters $K_{ai}^2$ introduced in \eqref{Kabdef},
This is a system of three linear equations for $(c_{11},c_{21},\lambda_0)$, 
which has a unique solution, provided that the determinant 
$D = (\theta_1 F_{a1}^2)^2 + (\theta_2 F_{a2}^2)^2 + 2 \rho (\theta_1 F_{a1}^2) (\theta_2 F_{a2}^2)$ is non-zero.
The solution is 
\begin{align}
c_{11} &= \frac{2F_0^2}{\omega \alpha_\theta} \cdot\frac{1}{D} 
\left(\theta_1 F_{a1}^2 + \rho \theta_2 F_{a2}^2\right)\,, \\
c_{21} &= \frac{2F_0^2}{\omega \alpha_\theta} \cdot\frac{1}{D} 
\left(\rho \theta_1 F_{a1}^2 + \theta_2 F_{a2}^2\right)\,, \\
\lambda_0 &= -\frac{2F_0^2}{\omega \alpha_\theta} \cdot\frac{1}{D}\,.
\end{align}

The rate function for the VIX options in the short-maturity regime
can be expanded in $x_{\mathrm{VIX}}$ as
\begin{equation}
J_V(K) = j_0 x^2_{\mathrm{VIX}} + j_1 x^3_{\mathrm{VIX}} + O\left(x^4_{\mathrm{VIX}}\right)\,.
\end{equation}
The coefficient of the quadratic term is
\begin{align}
j_0 = \frac{1}{2(1-\rho^2)} 
\left( c_{11}^2 + c_{21}^2 - 2\rho c_{11} c_{21} \right) 
=  2 \left( \frac{F_0^2}{\omega \alpha_\theta} \right)^2 \cdot \frac{1}{D}\,, 
\end{align}
which reproduces the quoted result for  $j_0$ in Proposition~\ref{prop:JV2F}.

\textit{$O(x_{\mathrm{VIX}}^2)$ terms.} The computation can be extended to the quadratic terms in log-moneyness. The terms of quadratic order in the expansion of the equations \eqref{x1eq}, \eqref{x2eq} give two equations for $c_{12},c_{22}$:
\begin{align}
c_{12} &+ \lambda_1 \left(\omega \alpha_\theta\right) \left(\theta_1 F_{a1}^2 + \rho \theta_2 F_{a2}^2\right)
+\lambda_0 \left(\omega \alpha_\theta\right)^2 \Bigg\{
\theta_1 c_{11} \frac{1}{\tau} \int_0^\tau \xi_0^u \left(\theta_1 e^{-k_1 u} + \rho \theta_2 e^{-k_2 u}\right) e^{-k_1 u} du  \nonumber \\
&\qquad\qquad + \theta_2 c_{21} \frac{1}{\tau} \int_0^\tau \xi_0^u \left(\theta_1 e^{-k_1 u} + \rho \theta_2 e^{-k_2 u}\right) e^{-k_2 u} du  \Bigg\}  = 0\,,\label{46} \\
c_{22} &+ \lambda_1 \left(\omega \alpha_\theta\right) \left(\rho \theta_1 F_{a1}^2 + \theta_2 F_{a2}^2\right)+ \lambda_0 (\omega \alpha_\theta)^2 \Bigg\{
\theta_1 c_{11} \frac{1}{\tau} \int_0^\tau \xi_0^u \left(\rho \theta_1 e^{-k_1 u} +  \theta_2 e^{-k_2 u}\right) e^{-k_1 u} du  \nonumber \\
&\qquad \qquad+ \theta_2 c_{21} \frac{1}{\tau} \int_0^\tau \xi_0^u \left(\rho \theta_1 e^{-k_1 u} + \theta_2 e^{-k_2 u}\right) e^{-k_2 u} du  \Bigg\}  = 0\label{47}\,.
\end{align}

Expanding also the constraint \eqref{x12const} to quadratic order gives another linear equation 
\begin{align}\label{lam1eq}
\omega \alpha_\theta \left(\theta_1 F_{a1}^2 c_{12} + \theta_2 F_{a2}^2 c_{22} \right) 
+ \frac12(\omega\alpha_\theta)^2 
\Big\{ \theta_1^2 F_{b,11}^2 c_{11}^2 + 2 \theta_1\theta_2 F_{b,12}^2 c_{11}c_{21}
+ \theta_2^2 K_{b,22}^2 c_{21}^2\Big\}  = 2 F_0^2\,.
\end{align}
Combined with \eqref{46} and \eqref{47} we have three linear equations in the three unknowns $c_{12},c_{22}, \lambda_1$.

The coefficient of the cubic order term $O(x^3_{\mathrm{VIX}})$ in the rate function is 
\begin{align}\label{j1exp}
j_1 &= \frac{1}{(1-\rho^2)} 
\left( c_{12} (c_{11} - \rho c_{21}) +  c_{22} (c_{21}- \rho c_{11} )\right) \\
&=  2 \left( \frac{K_{\mathrm{ATM}}^2}{\omega \alpha_\theta} \right) \cdot \frac{1}{D} 
\left(\theta_1 F_{a1}^2 c_{12} + \theta_2 F_{a2}^2 c_{22} \right)\,, \nonumber
\end{align}
where we used the solutions for $c_{11},c_{21}$ obtained in the previous step.

We note that the required combination of $c_{21},c_{22}$ is precisely the same as that which 
appears in the first line of the constraint \eqref{lam1eq}. 
Substituting the known values of $c_{11},c_{21}$ in the second term of \eqref{lam1eq} gives
\begin{align}\label{50}
\theta_1 F_{a1}^2 c_{12} + \theta_2 F_{a2}^2 c_{22} &=
-\frac{2 F_0^4}{(\omega\alpha_\theta)D^2} G +  \frac{2F_0^2}{\omega\alpha_\theta}\,, 
\end{align}
with $D:=(\theta_1 F_{a1}^2)^2 + (\theta_2 F_{a2}^2)^2 + 2\rho (\theta_1 F_{a1}^2)(\theta_2 F_{a2}^2)$ and
\begin{align}
G &:= \theta_1^2 F_{b11}^2 \left(\theta_1 F_{a1}^2 + \rho \theta_2 F_{a2}^2\right)^2 \nonumber \\
&\qquad\qquad+ 2 \theta_1 \theta_2 F_{b12}^2 \left(\theta_1 F_{a1}^2 + \rho \theta_2 F_{a2}^2\right)
\left(\rho \theta_1 F_{a1}^2 +  \theta_2 F_{a2}^2\right) 
+ \theta_2^2 F_{b22}^2 \left(\rho\theta_1 F_{a1}^2 +  \theta_2 F_{a2}^2\right)^2\,.
\end{align}

Substituting \eqref{50} into \eqref{j1exp} yields the stated result \eqref{j1sol} for $j_1$.
This completes the proof.
\end{proof}


\subsection{Proof of Lemma~\ref{lemma:2}}

\begin{proof}
We have
\begin{align}
F_{ai}^2 F_{aj}^2 - F_0^2 F_{b,ij}^2 &= \frac{1}{\tau^2}
\int_0^\tau \xi_0^u \xi_0^v \left(e^{-k_i u - k_j v} - e^{-(k_i+k_j) v} \right) du dv \\
&= - \frac{1}{2\tau^2} 
\int_0^\tau  \xi_0^u \xi_0^v \left(e^{-k_i u} - e^{-k_i v} \right) \left(e^{-k_j u} - e^{-k_j v} \right) du dv\,.\nonumber
\end{align}
Using $\xi_0^u \geq 0$ and 
noting that $\left(e^{-k_i u} - e^{-k_i v} \right) \left(e^{-k_j u} - e^{-k_j v} \right)\geq 0$ for all $u,v$, 
the stated inequality follows. 
\end{proof}

\subsection{Proof of Corollary~\ref{corr:D2G}}

\begin{proof}
We start by proving the inequality $D^2 \leq F_0^2 G$ for $\rho \geq 0$.
Using the definitions \eqref{Ddef} and \eqref{Gdef}, the difference $D^2-F_0^2 G$ is expressed as a sum of powers of $\theta_i$ as
\begin{equation}
D^2 - F_0^2 G = \theta_1^4 c_1 + \theta_1^3 \theta_2 c_2 + \theta_1^2 \theta_2^2 c_3+\theta_1 \theta_2^3 c_4 + \theta_2^4 c_5\,.
\end{equation}

The coefficients $c_i$ are given by
\begin{align}
c_1 &= F_{a1}^4\left(F_{a1}^4-F_0^2 F_{b11}^2 \right)\,, \\
c_2 &= 2\rho F_{a1}^2 \left( F_{a2}^2 \left(F_{a1}^2 - F_0^2 F_{b11}^2\right) +
F_{a1}^2 \left(F_{a1}^2 F_{a2}^2 - F_0^2 F_{b12}^2\right)\right)\,, \\
c_3 &= 2(1+\rho^2) F_{a1}^2 F_{a2}^2 \left(F_{a1}^2 F_{a2}^2 - F_0^2 F_{b12}^2\right) \nonumber \\
&\qquad\qquad+ \rho^2 F_{a2}^4 \left(F_{a1}^4 - F_0^2 F_{b11}^2 \right) +
 \rho^2 F_{a1}^4 \left(F_{a2}^4 - F_0^2 F_{b22}^2 \right)\,,\\
 c_4 &= 2\rho F_{a2}^2 \left( F_{a2}^2 \left(F_{a1}^2 F_{a2}^2 - F_0^2 F_{b12}^2\right) +
F_{a1}^2 \left(F_{a2}^4 - F_0^2 K_{b22}^2\right)\right)\,,\\
 c_5 &= F_{a2}^4(F_{a2}^4 - F_0^2 F_{b22}^2)\,.
\end{align}

Using the inequalities in Lemma~\ref{lemma:2}, 
the coefficients $c_1, c_3, c_5$ are positive for any $\rho$, but $c_2,c_4$ are proportional to $\rho$ with a positive coefficient,
so they are non-negative only for $\rho\geq 0$. 

Using the inequality of Corollary~\ref{corr:D2G} into the result \eqref{VIXskew2F} gives that the asymptotic ATM VIX skew is non-negative, for any $\rho\geq 0$. 
This completes the proof.
\end{proof}

\subsection{Proof of Proposition~\ref{prop:JV2Fvov}}

\begin{proof}

By Theorem~\ref{thm:small:OTM:two:factor}, the rate function for VIX options in the two-factor Bergomi model in the small vol-of-vol limit 
is expressed as a constrained optimization problem
\begin{equation}
J_{\rm vov}(K,T) = \inf_{x_1,x_2} \frac{1}{2\Delta_v}
\left( x_1^2 v_T^1 + x_2^2 v_T^1 - 2x_1 x_2 v_T^{1,2} \right)\,,
\end{equation}
with $\Delta_v := v_T^1 v_T^2 - (v_T^{1,2})^2$ and
$x_{1,2}$ are constrained as
\begin{equation}\label{x12constvov}
\frac{1}{\tau} \int_T^{T+\tau} \xi_0^u 
\exp\left( \alpha_\theta \left[\theta_1 e^{-k_1 (u-T)} x_1 + \theta_2 e^{-k_2 (u-T)}
x_2\right]\right) du = K^2 = F_0^2(T) e^{2x_{\mathrm{VIX}}}\,,
\end{equation}
where $x_{\mathrm{VIX}}=\log(K/F_0(T))$ is the log-moneyness of the VIX option.

The constraint is taken into account by introducing a Lagrange multiplier $\lambda$. Define the auxiliary function
\begin{align}
\Lambda(x_1,x_2,\lambda) &= 
\frac{1}{2\Delta_v}
\left( x_1^2 v_T^2 + x_2^2 v_T^1 - 2v_T^{1,2} x_1 x_2 \right) \\
&\qquad\qquad+ \lambda
\left( \frac{1}{\tau} \int_T^{T+\tau} \xi_0^u 
\exp \left[ \alpha_\theta 
\left(\theta_1 e^{-k_1 (u-T)} x_1 + \theta_2 e^{-k_2 (u-T)} x_2\right)\right] du - K^2 \right) \,.\nonumber
\end{align} 

The extremal conditions for this function $\partial_{x_i} \Lambda(x_i,\lambda) = 0$ can be put into the form
\begin{align}\label{x1Teq}
&x_1 + \lambda \alpha_\theta \frac{1}{\tau} \int_T^{T+\tau} \xi_0^u 
\left(\theta_1 e^{-k_1 (u-T)} v_T^1 + \theta_2 e^{-k_2 (u-T)} v_T^{1,2}\right) e^{Z(u)} du = 0\,, \\
\label{x2Teq}
&x_2 + \lambda \alpha_\theta \frac{1}{\tau} \int_0^\tau \xi_0^u 
\left(\theta_1 e^{-k_1 (u-T)} v_T^{1,2}+  \theta_2 e^{-k_2 (u-T)} v_T^2\right) e^{Z(u)} du = 0 \,,
\end{align}
where we defined $Z(u) := \alpha_\theta 
\left(\theta_1 e^{-k_1 (u-T)} x_1 + \theta_2 e^{-k_2 (u-T)} x_2\right)$.

These are two non-linear equations for $(x_1,x_2)$.  Together with the constraint \eqref{x12constvov} they form a system of 3 non-linear equations for $(x_i,\lambda)$. We seek a solution of these equations in an expansion of the log-moneyness $x_{\mathrm{VIX}}$. For an ATM VIX option, these equations are satisfied by $x_1=x_2=0$. Therefore, we seek the solution as an expansion in 
$x_{\mathrm{VIX}}$ as
\begin{align}
x_1 &= c_{11} x_{\mathrm{VIX}} + c_{12} x^2_{\mathrm{VIX}} + O\left(x^3_{\mathrm{VIX}}\right)\,,\\
x_2 &= c_{21} x_{\mathrm{VIX}} + c_{22} x^2_{\mathrm{VIX}} + O\left(x^3_{\mathrm{VIX}}\right)\,,\\
\lambda &= \lambda_0 x_{\mathrm{VIX}} + \lambda_1 x^2_{\mathrm{VIX}} + O\left(x^3_{\mathrm{VIX}}\right)\,.
\end{align}
The coefficients $c_{ij}$ can be determined order by order in $x_{\mathrm{VIX}}$. We start by considering the $O(x_{\mathrm{VIX}})$ terms in the expansion of the master equations \eqref{x1Teq}, \eqref{x2Teq}.

\textit{$O(x_{\mathrm{VIX}})$ terms.} Expanding the equations \eqref{x1Teq}, \eqref{x2Teq} and the constraint \eqref{x12constvov} in $x_{\mathrm{VIX}}$ and keeping only the linear terms gives
\begin{align}
& c_{11} + \lambda_0 \alpha_\theta \frac{1}{\tau} \int_T^{T+\tau} \xi_0^u \left(\theta_1 e^{-k_1 (u-T)} v_T^1 +
\theta_2 e^{-k_2 (u-T)} v_T^{1,2}\right) du = 0\,, \\
& c_{21} + \lambda_0  \alpha_\theta \frac{1}{\tau} \int_T^{T+\tau} \xi_0^u \left( \theta_1 e^{-k_1 (u-T)} v_T^{1,2} +
\theta_2 e^{-k_2 (u-T)} v_T^2\right) du = 0\,.
\end{align}
The constraint \eqref{x12constvov} gives the additional equation
\begin{equation}
\alpha_\theta \left( \theta_1 \left(\frac{1}{\tau} \int_T^{T+\tau} \xi_0^u 
e^{-k_1 (u-T)}du\right) c_{11} +
\theta_2 \left(\frac{1}{\tau} \int_T^{T+\tau} \xi_0^u e^{-k_2 (u-T)}du\right) c_{21}  \right) =2F_0^2(T)\,.
\end{equation}
The integrals can be written in terms of the parameters $K_{ai}^2(T)$ 
introduced in \eqref{KabTdef}.
We obtain a system of three linear equations for $(c_{11},c_{21},\lambda_0)$, 
which has a unique solution, provided that the determinant 
\begin{equation}\label{Dvovdef}
D_{\rm vov}(T) = \left(\theta_1 K_{a1}^2(T)\right)^2 v_T^1 
+ \left(\theta_2 K_{a2}^2(T)\right)^2 v_T^2 + 2 \left(\theta_1 K_{a1}^2(T)\right) 
\left(\theta_2 K_{a2}^2(T)\right) v_T^{1,2}
\end{equation}
is non-vanishing.
The solution is 
\begin{align}
c_{11} &= \frac{2F_0^2(T)}{ \alpha_\theta} \cdot\frac{1}{D_{\rm vov}(T)} 
\left(\theta_1 K_{a1}^2(T) v_T^1 + \theta_2 K_{a2}^2(T) v_T^{1,2}\right)\,, \\
c_{21} &= \frac{2F_0^2(T)}{\alpha_\theta} \cdot\frac{1}{D_{\rm vov}(T)} 
\left(\theta_1 K_{a1}^2(T) v_T^{1,2} + \theta_2 K_{a2}^2(T) v_T^2\right)\,, \\
\lambda_0 &= -\frac{2F_0^2(T)}{ \alpha_\theta} \cdot\frac{1}{D_{\rm vov}(T)}\,. 
\end{align}

The VIX rate function can be expanded in $x_{\mathrm{VIX}}$ as
\begin{equation}
J_{\rm vov}(K,T) = j_0 x^2_{\mathrm{VIX}} + j_1 x^3_{\mathrm{VIX}} + O\left(x^4_{\mathrm{VIX}}\right)\,.
\end{equation}
The coefficient of the quadratic term is
\begin{align}
j_0 = \frac{1}{2\Delta_v} 
\left( c_{11}^2 v_T^2 + c_{21}^2 v_T^1 - 2 c_{11} c_{21} v_T^{1,2}\right) 
=  2 \left( \frac{F_0^2(T)}{\alpha_\theta} \right)^2 \cdot \frac{1}{D_{\rm vov}(T)}\,, 
\end{align}
which reproduces the quoted result for  $j_0$ in Proposition~\ref{prop:JV2Fvov}.

\textit{$O(x_{\mathrm{VIX}}^2)$ terms.} The computation can be extended to the quadratic terms in log-moneyness. 

Expanding the constraint \eqref{x12constvov} to quadratic order gives the linear equation for $c_{12},c_{22}$
\begin{align}\label{lam1eq:2vov}
& \alpha_\theta \left(\theta_1 F_{a1}^2(T) c_{12} + \theta_2 F_{a2}^2(T) c_{22} \right) \\
&\qquad\qquad+ \frac12 \alpha_\theta^2 
\Big\{ \theta_1^2 F_{b,11}^2(T) c_{11}^2 + 2 \theta_1\theta_2 F_{b,12}^2(T) c_{11}c_{21}
+ \theta_2^2 K_{b,22}^2(T) c_{21}^2\Big\}  = 2 F_0^2(T)\,. \nonumber
\end{align}
We will show that this constraint is sufficient in order to determine the coefficient of order $O(x^3_{\mathrm{VIX}})$ in the rate function $J_V(K,T)$.
This coefficient is  
\begin{align}\label{j1exp:2vov}
j_1 &= \frac{1}{\Delta_v} 
\left( c_{12} \left(c_{11} v_T^2- c_{21} v_T^{1,2}\right) +  c_{22} \left(c_{21} v_T^1-c_{11} v_T^{1,2} \right)\right) \\
&=  \frac{2F_0^2(T)}{ \alpha_\theta D_{\rm vov}(T)} 
\left(\theta_1 F_{a1}^2(T) c_{12} + \theta_2 F_{a2}^2(T) c_{22} \right)\,, \nonumber
\end{align}
where we used the solutions for $c_{11},c_{21}$ obtained in the previous step.
We note that the required combination of $c_{21},c_{22}$ is precisely the same as that which appears in the first line of the constraint \eqref{lam1eq:2vov}. 
Thus we can use \eqref{lam1eq:2vov} to determine this linear combination. Substituting into $j_1$ we obtain
\begin{align}
    j_1 = \frac{4F_0^4(T)}{\alpha_\theta^2 D_{\rm vov}(T)}
    \left( 1 - F_0^2(T) \frac{G_{\rm vov(T)}}{D^2_{\rm vov}(T)}\right)\,,
\end{align}
with
\begin{align}
G_{\rm vov}(T) &:= \theta_1^2 F_{b11}^2(T) \left(\theta_1 F_{a1}^2(T) v_T^1 + \theta_2 F_{a2}^2(T) v_T^{1,2} \right)^2 \nonumber \\
&\qquad+ 2 \theta_1 \theta_2 F_{b12}^2(T) \left(\theta_1 F_{a1}^2(T) v_T^1+  \theta_2 F_{a2}^2(T) v_T^{1,2}\right)
\left(\theta_1 F_{a1}^2(T) v_T^{1,2} +  \theta_2 F_{a2}^2(T) v_T^2\right) \nonumber\\
&\qquad\qquad + \theta_2^2 F_{b22}^2(T) \left(\theta_1 F_{a1}^2(T) v_T^{1,2} +  \theta_2 F_{a2}^2(T) v_T^2\right)^2\,.
\end{align}
This reproduces the stated result for $j_1$ and completes the proof of Proposition~\ref{prop:JV2Fvov}.
\end{proof}

\subsection{Proof of Proposition~\ref{prop:limvov}}

\begin{proof}
The VIX call option price is written as $C(K,T) = e^{-rT} F_0(T) c_{\mathrm{BS}}(x,v)$ with $x=\log(K/F_V)$ the log-moneyness and $v=\sigma_{\mathrm{VIX}}(K,T) \sqrt{T}$ the total volatility.
We denoted the undiscounted Black-Scholes call price with unit forward price as
\begin{equation}
c_{\mathrm{BS}}(x,v) = \Phi\left(-\frac{x}{v} + \frac12 v\right) - e^x \Phi\left(-\frac{x}{v} - \frac12 v\right)\,.
\end{equation}

As $\omega\to 0$, the VIX implied volatility approaches zero. Therefore we study the asymptotics of $c_{\mathrm{BS}}(x,v)$ as $v\to 0$. For an OTM call $x>0$ so we get the leading $
v\to 0$ asymptotics
\begin{equation}
c_{\mathrm{BS}}(x,v) = \phi\left(\frac{x}{v}-\frac{v}{2}\right) \frac{v^3}{x^2-\frac12 v^4}\left(1 + O(v^2)\right)\,,
\end{equation}
with $\phi(x) := \frac{1}{\sqrt{2\pi}} e^{-\frac12 x^2}$ being the density of the standard normal distribution.

We get
\begin{align}
\lim_{\omega \to 0} \omega^2 \log C(K,T)  &=
\lim_{\omega \to 0} \omega^2 \log c_{\mathrm{BS}}\left(x,\sigma_{\mathrm{VIX}}(x,T)\sqrt{T}\right)  \\
&= -\lim_{\omega\to 0} \frac{\omega^2}{2\sigma_{\mathrm{VIX}}^2 T} \left(x - \frac12 \sigma_{\mathrm{VIX}}^2 T\right)^2 \nonumber \\
&= -\lim_{\omega\to 0} \frac{\omega^2}{2\sigma_{\mathrm{VIX}}^2 T} x^2 = - J_{\rm vov}(K,T)\,. \nonumber
\end{align}
This completes the proof of Proposition~\ref{prop:limvov}.

\end{proof}

\end{document}